\documentclass[11pt]{article}

\usepackage{amsmath,amsfonts,bm}
\usepackage{amsthm}
\usepackage{mathtools}

\numberwithin{equation}{section}

\usepackage[figurename=Algorithm]{caption}

\def\Cov{\text{Cov}}

\newcommand{\cN}{\mathcal{N}}

\def\calA{\mathcal{A}}

\def\calC{\mathcal{C}}

\def\calH{\mathcal{H}}

\def\calM{\mathcal{M}}
\def\calN{\mathcal{N}}
\def\calO{\mathcal{O}}
\def\calP{\mathcal{P}}
\def\calR{\mathcal{R}}
\def\calS{\mathcal{S}}

\def\calX{\mathcal{X}}

\def\calY{\mathcal{Y}}
\def\reals{{\mathbb R}}

\def\pr{\mathbb{P}}
\renewcommand{\epsilon}{\varepsilon}
\newcommand{\dtv}{\text{TV}}

\newcommand{\cM}{\mathcal M}
\newcommand{\cO}{\mathcal O}
\newcommand{\cY}{\mathcal Y}

\newcommand{\mper}{\,.}
\newcommand{\iprod}[1]{\left\langle#1\right\rangle}

\newcommand{\Lap}{\mathrm{Lap}}
\newcommand{\tLap}{\mathrm{tLap}}

\newcommand{\cC}{\mathcal C}

\newcommand{\stab}{\mathrm{Stab}}
\newcommand{\score}{\textsf{score}}
\renewcommand{\stab}{\mathrm{stab}} 

\newcommand{\SelectionAlg}{\textsf{DP-Selection}} 

\newcommand{\cA}{\mathcal{A}}
\newcommand{\cB}{\mathcal{B}}

\newcommand{\set}[1]{\{#1\}}

\usepackage{paper}

\title{Privately Estimating a Gaussian: Efficient, Robust and Optimal
\footnote{Authors are listed in alphabetical order.}
}
\author[1]{Daniel Alabi\footnote{Email: \url{alabid@cs.columbia.edu}. Supported by a
Junior Fellowship from the Simons Foundation Society of Fellows and Cooperative Agreement
CB20ADR0160001 with the U.S. Census Bureau.}}
\affil[1]{Columbia University}
\author[2]{Pravesh K. Kothari\footnote{Email: \url{praveshk@cs.cmu.edu}. Supported by NSF CAREER Award \#2047933, NSF \#2211971, an Alfred P. Sloan Fellowship, and a Google Research Scholar Award.}}
\affil[2]{Carnegie Mellon University}

\author[3]{Pranay Tankala\footnote{Email: \url{pranay_tankala@g.harvard.edu}. Supported by Simons Foundation Grant 733782.}}
\author[3]{Prayaag Venkat\footnote{Email: \url{pvenkat@g.harvard.edu}. Supported by NSF grant CCF-1565264 and an NSF Graduate Fellowship under grant DGE1745303.}}
\affil[3]{Harvard University}

\author[4]{Fred Zhang\footnote{Email: \url{z0@berkeley.edu}.  Supported by NSF award CCF-1951384.}}
\affil[4]{University of California, Berkeley}

\date{}

\begin{document}

\maketitle

\begin{abstract}
In this work, we give efficient algorithms for privately estimating a Gaussian distribution in both pure and approximate differential privacy (DP) models with optimal dependence on the dimension in the sample complexity. 
\begin{itemize}
    \item 
In the pure DP setting, we give an efficient algorithm that estimates an unknown $d$-dimensional Gaussian distribution up to an arbitrary tiny total variation error using $\widetilde{O}(d^2 \log \kappa)$ samples while tolerating a constant fraction of adversarial outliers. Here, $\kappa$ is the condition number of the target covariance matrix. The sample bound matches best non-private estimators in the dependence on the dimension (up to a polylogarithmic factor). We prove a new lower bound on differentially private covariance estimation to show that the dependence on the condition number $\kappa$ in the above sample bound is also tight. Prior to our work, only identifiability results (yielding inefficient super-polynomial time algorithms) were known for the problem.
\item
In the approximate DP setting, we give an efficient algorithm to estimate an unknown Gaussian distribution up to an arbitrarily tiny total variation error using $\widetilde{O}(d^2)$ samples while tolerating a constant fraction of adversarial outliers. Prior to our work, all efficient approximate DP algorithms incurred a super-quadratic sample cost or were not outlier-robust. For the special case of mean estimation, our algorithm achieves the optimal sample complexity of $\widetilde O(d)$, improving on a $\widetilde O(d^{1.5})$ bound from prior work.
\end{itemize}

Our pure DP algorithm relies on a  recursive private preconditioning subroutine that utilizes the recent work~\cite{hopkins2021efficient} on private mean estimation. Our approximate DP algorithms are based on a substantial upgrade of the method of stabilizing convex relaxations introduced in~\cite{kothari2021private}. In particular, we improve on their mechanism by using a new unnormalized entropy regularization and a new and surprisingly simple mechanism for privately releasing covariances. 
\end{abstract}\clearpage



\section{Introduction}
Learning a high-dimensional Gaussian distribution is arguably the most basic task in statistical estimation. A long line of work has focused on finding algorithms for this fundamental problem that satisfy additional constraints such as robustness to adversarial outliers \cite{lai2016agnostic,diakonikolas2018robustly,diakonikolas2019robust,diakonikolas2019recent,kothari2022polynomial} and differential privacy \cite{KarwaV18,kamath2019privately,kamath2020private,brown2021covariance,LiuKKO21,BunKSW21,aden2021sample,tsfadia2022friendlycore,kothari2021private,LiuKO22,ashtiani2022private,kamath2022private}. An overarching goal in this line of work is to investigate the \emph{cost of privacy}---the overhead in sample complexity and running time that one must incur over and above the setting without the privacy constraints. Minimizing the cost of differential privacy---and ideally, achieving the same asymptotic dependence on the underlying dimension in sample complexity---is a challenging goal. 

To appreciate this goal, note that differentially private algorithms must necessarily have low sensitivity. That is, switching a single sample point in the input should lead to a small change to the estimate output by the algorithm. Importantly, this is a \emph{worst-case} guarantee: it must hold regardless of whether the input data satisfies the modeling assumption of being   independent Gaussian samples. Natural and simple estimators such as   the empirical mean and covariance of the input data have an unbounded sensitivity and result in no privacy guarantees.  

Without privacy constraints, standard concentration inequalities imply that the mean of an unknown Gaussian distribution can be estimated up to an $\ell_2$ error (more stringently, the ``Mahalanobis" error) of $\sim d/n$ with $n \gg d$ samples. For covariance estimation, a similar analysis yields an error rate of $\sim d^2/n$ for $n \gg d^2$.\footnote{Note that for obtaining vanishing total variation error guarantees, we need to estimate the covariance of the unknown Gaussian in the relative Frobenius distance.} The last few years have seen remarkable progress in finding private estimators  that come close to the above benchmarks. A recent work~\cite{hopkins2021efficient} provides an optimal estimator in $\ell_2$ norm for the mean of a  distribution with covariance bounded in spectral norm while satisfying \emph{pure} differential privacy---the strongest guarantee investigated in this setting. On the other hand, in the less stringent model of approximate differential privacy, an essentially optimal algorithm for estimating a Gaussian distribution was recently found in~\cite{ashtiani2022private}.

\paragraph{This Work} The main contribution of this work is providing   new sample-optimal algorithms for estimating a Gaussian distribution in both pure and approximate differential privacy models. For pure differential privacy, the main challenge is the task of covariance estimation for which no efficient algorithm was known. We resolve this challenge with dimension-dependence that asymptotically matches the optimal non-private bounds above. Our algorithms incur a logarithmic dependence on $\kappa$, the condition number of the unknown covariance. We improve prior lower bounds for the problem to show that such a dependence is in fact necessary for any pure differentially private algorithm thus concluding that our guarantees are asymptotically optimal. 

For approximate differential privacy, our main result builds   and improves upon the framework of~\cite{kothari2021private}. 
In \cite{kothari2021private}, the authors gave an efficient algorithm for Gaussian estimation that incurred dimension dependence of $\sim d^8$ in the sample complexity. As an immediate consequence of our upgraded guarantees, we obtain algorithms that use $\sim d^2$ samples to privately estimate a Gaussian distribution while being  resilient to a constant fraction of adversarial outliers. This improves on the previous best algorithm of~\cite{ashtiani2022private} that incurred a $\sim d^{3.5}$ sample cost. 

For our algorithms in the pure differential privacy model, our  main technique is a  preconditioning scheme for covariance estimation that iteratively constructs a rough estimate of the large eigenvalues of the unknown covariance. Previously, such a subroutine was proposed in \cite{kamath2019privately} with the weaker guarantees of approximate differential privacy. Our subroutine relies on the new sum-of-squares exponential mechanism introduced in the recent work of~\cite{hopkins2021efficient}.

For our algorithms in the approximate differential privacy model, our main contribution is a substantial upgrade of the method of \emph{stabilizing} convex relaxations introduced in~\cite{kothari2021private}. Their work makes a basic connection between outlier-robust and different private algorithms and introduces a general scheme that transforms convex programming-based outlier-robust algorithms for statistical estimation into (approximately) differentially private ones. Two main pieces in their scheme are a strongly convex regularizer on top of the starting convex relaxation and a noise addition mechanism for releasing the covariance that adds ``estimate dependent" noise --- i.e., the distribution of the noise depends on the computed covariance itself. They showed that despite this apparently leaky noise addition mechanism, their scheme leads to differentially private algorithms albeit with large polynomial factor losses in the sample complexity. The appealing aspect of their scheme is the generality and modularity in transforming robust estimators to private ones. However, a large polynomial factor loss, if inherent to this translation, could be prohibitive. Our main contribution is an upgrade on their scheme that leads to sample-optimal private algorithms. Our key innovation is an introduction and analysis of a new strongly convex regularizer based on \emph{unnormalized entropy} and a new and substantially simpler and optimal noise addition mechanism for releasing covariances. We believe that our improved version of the stabilization method will find further uses in the design of differentially private algorithms. 

\subsection{Our Results}
We now describe our results in more detail. We start by formally defining differential privacy. 

\begin{definition}[Differential Privacy~\cite{DworkMNS06, DKMMN06}]
For $\eps\geq 0$ and $\delta \in [0, 1]$, a (randomized) algorithm $\calM$ is
$(\eps, \delta)$-differentially private if for all pairs of neighboring databases
$x, x'$ that differ in exactly one row, and for all
output subsets $S$ of the range of $\calM$, the following holds:
$$
\pr[\calM(x)\in S] \leq e^\eps\cdot\pr[\calM(x')\in S] + \delta,
$$
where the probability is over the coin flips of the
algorithm $\calM$. We say that an algorithm satisfies $\epsilon$-pure DP if it satisfies $(\epsilon,\delta)$-DP with $\delta=0$. 
\label{def:dp}
\end{definition}


Our main contribution is finding sample-optimal private and outlier-robust algorithms for estimating a high-dimensional Gaussian distribution. In~Table \ref{tab:bounds}, we summarize the sample complexity,
privacy, and robustness guarantees of our algorithms
in relation to prior work. 

\begin{table}[t]\makeatletter\phantomsection\def\@currentlabel{(table)}\makeatother
\centering
\begin{tabular}{@{}llll@{}}
\toprule
Reference                           & Sample Complexity                                  & Privacy & Robustness \\ \midrule
Na\"ive estimator    & $\frac{d^2}{\alpha^2} + \frac{\kappa d^2}{\alpha \eps}$ & $(\eps,\delta)$   & Not robust
\\[4.5pt]
\cite{kamath2019privately}                  & $\frac{d^2}{\alpha^2} + \frac{d^2}{\alpha \eps} + \frac{d^{3/2}\log^{1/2} \kappa}{\eps}$ & $(\eps,\delta)$ & Not robust \\[6pt]
\cite{kamath2022private} & $\frac{d^2}{\alpha^2} + \frac{d^2}{\alpha \eps} + \frac{d^{5/2}}{\eps}$            & $(\eps,\delta)$ & Not robust \\[6pt]
\cite{ashtiani2022private} and  \cite{tsfadia2022friendlycore} & $\frac{d^2}{\alpha^2} + \frac{d^2}{\alpha \eps}$            & $(\eps,\delta)$ & Not robust\\[6pt]
\cite{ashtiani2022private} & $\frac{d^{3.5}}{\alpha^3 \eps}$ & $(\eps,\delta) $ & $\eta$\\[6pt]
\cite{kothari2021private}  & $\frac{d^8}{\alpha^4 \varepsilon^8}$            & $(\eps,\delta)$  & $\eta$ \\[4pt] \midrule
Our result (Theorem \ref{thm:pure-dp-cov-est}) & $\frac{d^2\log\kappa}{\eps} + \frac{d^2}{\alpha^2 \eps} $ & $(\eps,0)$ & $\sqrt{\eta}$ \\[4pt]
Our result (Theorem \ref{thm:dp-robust-cov-est}) & $\frac{d^2}{\alpha^2\epsilon^{4.5}}$ & $(\eps,\delta)$ & $\eta$
 \end{tabular}
\caption{Computationally-efficient sample complexity upper bounds for    DP covariance estimators of $\cN(0,\Sigma)$, with the utility guarantee   $\left\| \widehat\Sigma - \Sigma \right\|_{\Sigma} \leq \alpha$.  Robustness refers to additional error when the sample is $\eta$-corrupted.  Here, $d$ is the dimension, $\kappa$ is an a priori bound such that $\mathbb{I} \preceq \Sigma \preceq \kappa \mathbb{I}$.  The   bounds suppress polylogarithmic factors in $d,  \frac{1}{\eta}, \frac{1}{\alpha},$ and $\frac{1}{\delta}$. }
\label{tab:bounds}
\end{table}

Specifically, in the setting of pure DP, we prove:
\begin{theorem}[Pure DP Gaussian estimation; see Theorem \ref{thm:pure-dp-cov-est-restate} and Theorem \ref{thm:pure-dp-robust-gaussian}]
\label{thm:pure-dp-cov-est}
Fix $\alpha,\epsilon, \eta >0$.
Let $\mathbb{I} \preceq \Sigma_* \preceq \kappa \mathbb{I}$ and
$\norm{\mu_*}_2\leq R$.
There is an efficient $\epsilon$-pure DP algorithm that takes input an $\eta$-corrupted sample of size $n$ from $\cN(\mu_*, \Sigma_*)$ and with probability at least $0.99$ over the draw of the sample and the randomness of the algorithm, outputs estimates $\widehat{\mu}$, $\widehat{\Sigma}$ such that $TV(\cN(\mu_*, \Sigma_*), \cN(\widehat{\mu}, \hat{\Sigma})) \leq \alpha+O(\sqrt{\eta})$ so long as $n \geq \widetilde{O}\left(d^2 \log (\kappa)/\epsilon + d^2/\alpha^2 \epsilon + d\log R/\eps\right)$.

\end{theorem}

As discussed earlier, the case when $\Sigma=I$ (i.e., mean estimation) was resolved in the recent work of \cite{hopkins2021efficient} with essentially optimal $\widetilde{O}(d)$ sample complexity. For high dimensional covariance estimation all prior efficient algorithms only satisfied an approximate DP guarantee. We note that an identifiability argument was presented in~\cite{aden2021sample,BunKSW21}. The conceptual barrier here is that the exponential mechanism, a canonical mechanism for ensuring pure DP, does not admit a straightforward efficient implementation for private estimation.

Our work provides the first computationally-efficient algorithm for covariance estimation and achieves  near optimal sample complexity. We note that the dependence on $\eta$ (the fraction of outliers) in the total variation guarantee can likely be improved to optimal~\cite{diakonikolas2018robustly} $\widetilde{O}(\eta)$ though our current techniques do not yield such a guarantee. Our key idea is exploiting the pure DP mean estimation algorithm of \cite{hopkins2021efficient} to build a preconditioning scheme for covariances (see Section \ref{sec:pure-dp-tech}). Such a scheme was previously used in~\cite{kamath2019privately} for covariance estimation with approximate DP.



\paragraph{Approximate DP Algorithm}


We first consider approximate DP robust mean estimation of $\mathcal {N}(\mu_*, I)$.
\begin{theorem}[Approximate DP Robust Mean Estimation]
\label{thm:dp-robust-mean-est}
Let $d \in \N$ and $\epsilon, \delta\in(0, 1), \eta>0$. There is an algorithm that takes input an $\eta$-corrupted sample of size $n$ from a Gaussian distribution $\cN(\mu_*,I)$ and with probability at least $0.99$ over the draw of the sample and the randomness of the algorithm, outputs estimates $(\widehat{\mu}, I)$ such that $\norm{\widehat{\mu} - \mu_*}_2 \leq O(\eta \sqrt{\log (1/\eta)})$ so long as $n \geq \widetilde{O}\left(\frac{d\log^3 (1/\delta)}{\eps^3 \eta^2}\right)$. 
\end{theorem}
For the near optimal robustness guarantee that we get in the above result, the prior best algorithms require $\widetilde{O}(d^{1.5})$ samples ~\cite{LiuKKO21,ashtiani2022private}. On the other hand, a suboptimal dependence of $O(\sqrt{\eta})$ on the additive error as a function of the fraction of outliers can be obtained with linear number of samples~\cite{hopkins2021efficient}. Our contribution is showing a simple algorithm that achieves the best of all worlds: a private, robust mean estimation algorithm with optimal dependence on both fraction of outliers and the dimension $d$.


More generally, our method yields the following optimal guarantee for estimating an unknown Gaussian distribution (i.e., mean and covariance estimation).
\begin{theorem}[Approximate DP Robust Covariance Estimation]
\label{thm:dp-robust-cov-est}
Fix $d \in \N, \epsilon, \delta\in(0, 1), \eta>0$. There is an efficient $(\epsilon, \delta)$-DP algorithm that takes input an $\eta$-corrupted sample of size $n$ from a Gaussian distribution $\cN(\mu_*,\Sigma_*)$ and with probability at least $0.99$ over the draw of the sample and the randomness of the algorithm, outputs estimates $(\hat{\mu}, \hat{\Sigma})$ such that $\text{TV}(\cN(\mu_*,\Sigma_*), \cN(\hat{\mu}, \hat{\Sigma})) \leq \widetilde{O}(\eta)$ so long as  $n \geq \widetilde{O}\left(\frac{d^2 \log^4(1/\delta)}{\epsilon^{4.5} \eta^2} \right)$.
\end{theorem}
The dependence on the sample complexity as a function of the dimension $d$ is optimal. In addition, our algorithm gets the optimal dependence on the error rate as a function of $\eta$.

There is a long sequence of works on private robust Gaussian estimation with approximate DP. The best known prior work on private robust estimation is due to~\cite{ashtiani2022private}, which requires $\widetilde{O}( d^{3.5})$ samples. Concurrent work  of~\cite{kothari2021private} gave an algorithm with a polynomial but worse $\widetilde{O}(d^8)$ sample complexity.  
In addition, \cite{ashtiani2022private} gives a private yet non-robust algorithm with $\widetilde O(d^2)$ sample complexity.
Similar to that, the results of \cite{kamath2019privately, kamath2022private} obtain (nearly) optimal sample cost, but fail at being robust.

Our algorithm builds upon a general method introduced by \cite{kothari2021private} for stabilizing convex relaxations. However, their specific implementation seems to inherently suffer from large polynomial losses in sample complexity. We introduce an upgraded version of their framework via two important changes. First, we introduce a new unnormalized entropy based regularization that replaces the $\ell_2$ norm regularization used in \cite{kothari2021private}. 
The analysis of this regularization requires proving an inequality that improves on na\"ive strong convexity in order to obtain a tighter guarantee; see Section \ref{sec:outlier-rate}.
Second, the algorithm of \cite{kothari2021private} applies a sophisticated estimate-dependent noise-injection  mechanism to release the estimate from the convex programming solution. The  noise in turn hurts the accuracy of their final estimate. We give a surprisingly simple mechanism to release covariances --- our mechanism simply releases $\sim d$ samples from an estimate of the covariance computed by the stabilized relaxation. This mechanism turns out to be significantly more efficient in terms of sample complexity. We give a high-level overview of the approximate DP result in Section \ref{sec:approx-dp-tech}.

\paragraph{Pure DP Lower Bound}
\begin{table}[t]\makeatletter\phantomsection\def\@currentlabel{(table)}\makeatother
\centering
\begin{tabular}{@{}llll@{}}
\toprule
Reference                           & Sample Complexity                                  & Privacy  \\ \midrule
\cite{kamath2019privately}                  & $  \frac{d^2}{\alpha \eps}  $ & $(\eps,0)$   \\[6pt]
\cite{KMS22} & $  \frac{d^2}{\alpha \eps}$            & $(\eps,\delta)$   \\[4pt]\midrule
Our result (Theorem \ref{thm:pure-dp-cov-est-lower-bound}) & $\frac{d^2 \log(\kappa)}{\eps}  + \frac{d^2}{\alpha\eps}$  &$(\eps,0)$ &  
 \end{tabular}
\caption{Sample complexity lower bounds for    DP covariance estimators of $\cN(0,\Sigma)$, with the utility guarantee   $\left\| \widehat\Sigma - \Sigma \right\|_{\Sigma} \leq \alpha$.   Here, $d$ is the dimension, $\kappa$ is an a priori bound such that $\mathbb{I} \preceq \Sigma \preceq \kappa \mathbb{I}$.  The   bounds suppress polylogarithmic factors in $d,   \frac{1}{\alpha},$ and $\frac{1}{\delta}$. }
\label{tab:lower-bounds}
\end{table}

We complement our algorithmic result with a nearly matching  lower bound. The analysis is a standard packing lower bound, building  upon  \cite{kamath2019privately, BunKSW21}. 
Previously, \cite{kamath2019privately} shows that under $\eps$-DP, $\Omega\left(d^2/(\alpha\eps)\right)$ samples are required to estimate the covariance of a Gaussian up to an error $\alpha$ (in the Mahalanobis distance). 
In contrast, our upper bound  (Theorem \ref{thm:pure-dp-cov-est}) has a $\log (\kappa)$ term. We show that the dependence on the condition number   is necessary:

\renewcommand{\dtv}{\text{TV}}

\begin{theorem}[Pure DP lower bound for covariance estimation; see Theorem \ref{thm:pure-dp-cov-est-lower-bound-full}]
\label{thm:pure-dp-cov-est-lower-bound}
Let $\eps, \alpha \in (0, 1)$, $d\geq 2$. Any $\epsilon$-DP algorithm that, given $n$ i.i.d.\ samples $\calX = \{X_1, \ldots, X_n\}$ from $\calN(0, \Sigma)$ for an unknown $\Sigma \in \reals^{d \times d}$ satisfying $I \preceq \Sigma \preceq \kappa I$, outputs $\widehat{\Sigma} = \widehat{\Sigma}(\calX)$ such that, with probability 0.9,
\[
\dtv\left(\calN(0, \Sigma), \calN\left(0, \widehat{\Sigma}\right)\right) < O(\alpha),
\]
requires
\[
n = \widetilde \Omega\left(\frac{d^2 \log(\kappa)}{\eps}  + \frac{d^2}{\alpha\eps }\right).
\]
\end{theorem}

Comparing   this with  our upper bound (Theorem \ref{thm:pure-dp-cov-est}) shows that our pure DP algorithmic result achieves nearly optimal sample complexity, only up to a $1/\alpha$ factor in its second term.  

We note that there are known \emph{approximate DP} algorithms~\cite{kamath2022private,kothari2021private,ashtiani2022private,LiuKO22}  with sample complexity that do not depend on the condition number. Thus, our lower bound above shows a basic separation in the sample complexity for estimating a Gaussian distribution in pure and approximate DP models.  

\paragraph{Concurrent Work} 
In a concurrent work, Hopkins, Kamath,
Majid and Narayanan~\cite{HKMN22} also obtained near-optimal polynomial-time algorithms for learning a Gaussian subject to pure and approximate DP and robustness constraints. The sample complexities achieved in our paper for approximate DP 
are the same as those in \cite{HKMN22} up to $\poly(1/\epsilon, \log(1/\delta), 1/\alpha)$ factors.



\subsection{Additional Related Work}
\paragraph{Private Estimation}
Prior to our work, there is large body of work on parameter estimation under differential privacy.
Assuming the data is drawn from a one-dimensional 
normal distribution,
Karwa and Vadhan~\cite{KarwaV18} obtain the first (finite-sample)
guarantees for confidence interval estimation of the mean---whether or not
the variance of the population is known. Their bounds are tight up to
logarithmic factors.
In the high-dimensional cases,
Kamath, Li, Singhal, and Ullman~\cite{kamath2019privately} give the first $(\eps,\delta)$-DP algorithm for learning a Gaussian. Their sample complexity has a logarithmic dependence on the condition number of the unknown covariance. This result has since been improved and refined by multiple follow-up papers, including \cite{biswas2020coinpress,brown2021covariance,ashtiani2022private,kamath2022private,kothari2021private}, with a    focus  on removing the  dependence on conditioning.  We summarize these upper bounds in Table \ref{tab:bounds}. We stress that all of them study $(\eps,\delta)$-DP. In fact, our work is the first sample- and computationally-efficient algorithm for learning a multivariate Gaussian under the stronger $\eps$-DP. 

For pure DP, the most relevant work to us is Hopkins, Kamath and Majid \cite{hopkins2021efficient}, which gives an efficient pure DP procedure for mean estimation of bounded second moment distributions. 
Our pure DP  algorithm builds upon their main result.
In addition, Bun et al. \cite{BunKSW21} gives a general (inefficient) cover-based approach to pure DP estimation.


\paragraph{Robust Statistics and Connections with Privacy}
Robust statistics is a classic area of study, dating back to at least 1960s \cite{huber1964robust, tukey60}.
Computationally-efficient procedures for robust statistics have been developed more recently in the algorithms literature. We refer
the reader to \cite{diakonikolas2019recent} for a survey.

Robustness and privacy are spiritually similar.
For robustly
learning the mean of a Gaussian distribution subject to privacy, Liu, Kong, Kakade, and Oh
provide polynomial-time algorithms with
sub-optimal sample complexity~\cite{LiuKKO21}.
However, with optimal  sample complexity, they  only provide an exponential-time algorithm. 
For robustly learning a general Gaussian distribution subject to DP, 
Kothari, Manurangsi, and Velingker provide a 
polynomial-time approximate DP algorithm with sample complexity
$\widetilde{O}(d^8)$~\cite{kothari2021private}.
On the other hand,
Ashtiani and Liaw
obtain polynomial-time
private and robust algorithms with sample complexity
$\tilde{O}(d^{3.5})$~\cite{ashtiani2022private}. Liu, Kong, and Oh~\cite{LiuKO22}
rely on the observation that one-dimensional robust
statistical estimators have low sensitivity on resilient datasets.
Resilience, defined by 
Steinhardt, Charikar, and Valiant~\cite{SteinhardtCV18},
measures how stable the empirical mean is to 
deletion of a constant fraction of a dataset. This
definition is, intuitively, similar in spirit, but  
not equivalent to DP guarantees.
Using this framework and the exponential mechanism,
Liu, Kong, and Oh~\cite{LiuKO22} design computationally
inefficient algorithms for DP statistical estimation. 
More relevant to us, Kothari, Manurangsi and Velingker \cite{kothari2021private} studies transforming robust algorithms into private ones. 
Our approximate DP result leverages and strenthens    their technique  of stabilizing  convex-programming-based estimators.
Finally, a very recent work by  Georgiev and Hopkins \cite{georgiev2022privacy} provides a generic meta-theorem for obtaining robustness from DP. 



\paragraph{Lower Bounds for DP}
In the pure DP case, geometric packing-style arguments are
often used to prove lower bounds
~\cite{McSherryT07, BeimelKN10,HardtTa10}.
Our pure DP lower bound arguments build upon
previous works~\cite{kamath2019privately, BunKSW21}.
For the approximate DP case, the lower bounds are based
on fingerprinting codes, introduced by 
Boneh and Shaw~\cite{BonehS98}. The existence of such codes
imply lower bounds for approximate DP~\cite{BunUlVa14}. The technique can be applied to prove sample complexity lower bounds for   Gaussian estimation under approximate DP \cite{KMS22}. 



\subsection{Organization}
The remainder of the paper is organized as follows. 
In Section \ref{sec:technique}, we give a high-level overview of our techniques. 
In Section \ref{sec:prelim}, we introduce notations and  technical backgrounds. 
In Section \ref{sec:pure-dp-cov}, we provide our pure DP estimation results. 
We discuss an approximate DP mechanism, based on Gaussian sampling, for releasing a covariance estimate in Section \ref{sec:gauss-sampling}.  
Using this mechanism, we provide our results on approximate DP and robust estimation in Section \ref{sec:approx-dp}. 
In Section \ref{sec:pure-lower}, we prove a pure DP covariance estimation lower bound.

\section{Techniques}\label{sec:technique}
In this section, we provide a high-level technical overview of our algorithmic results. 

\subsection{Pure DP Estimation} 
\label{sec:pure-dp-tech}
To illustrate our main ideas, we focus on private covariance estimation where we are given samples from $\mathcal{N} (0,\Sigma)$ and a promise that $\mathbb{I} \preceq \Sigma \preceq \kappa \mathbb{I}$. At first blush, one could imagine obtaining an approximate DP algorithm by taking the empirical covariance matrix and ``privatize" it by the Gaussian mechanism that adds an independent Gaussian of an appropriate variance to every entry of the estimated matrix. This mechanism, however, runs into immediate trouble -- the variance needs to be ``calibrated" to the sensitivity of the algorithm to switching a single point in the input. The naive empirical covariance has an unbounded sensitivity and thus cannot be privatized at all!  

It turns out (and was first observed and used in~\cite{kamath2019privately}) that if we first delete input points $X$ such that  $\|X\|_2^2 > \kappa d$ and then use the above simple mechanism, the sensitivity can be naturally upper bounded. This yields a private covariance estimator but needs a sample complexity that grows linearly in $\kappa$.

In order to improve the dependence on $\kappa$, \cite{kamath2019privately} proposes a (recursive)  private conditioning scheme. Roughly speaking, the procedure and its analysis work in two steps.
\begin{enumerate}[(i)]
    \item First,  they observe that the na\"ive application of the Gaussian mechanism   to clamped samples  provides a rough estimate of the covariance already. In particular, the large eigendirections have their corresponding eigenvalues   preserved multiplicatively, and their algorithm explicitly  identifies them   from the rough estimate.
    \item  Second, once we approximate the large eigenvectors, their algorithm shrinks the covariance along these   directions, via partially projecting out the eigenvectors associated with large eigenvalues. Recursively doing this preconditions the covariance, reducing the  conditioning number of $\Sigma$ to a constant.  Under this condition, the nai\"ve method of adding Gaussian noise to the empirical estimate leads to good accuracy.
\end{enumerate}

However, due to the use of the Gaussian mechanism, \cite{kamath2019privately} only attains approximate (and concentrated) DP guarantees. 
To obtain pure DP from this construction, our main observation is that their first step can be precisely replaced by the pure DP \textit{mean estimation} algorithm of \cite{hopkins2021efficient} (on bounded second moment distributions). 
This may sound surprising at first, but the idea is straightforward. 
In particular, let $\{X_i\}_{i=1}^n$ be the i.i.d.\ samples from $\mathcal{N}(0,\Sigma)$. If we form the samples 
 $\calY = \left\{Y_i = X_i X_i^\top /   \kappa: i \in [n]\right\}$  for $n \gg d^2/ \eps$, then \cite{hopkins2021efficient} guarantees that we can get $\widehat{\Sigma}$ such that  
$ \norm{\widehat{\Sigma} - \Sigma}_F \leq 0.01 \kappa$. Note that this already preserves the   eigenvectors, up to additive $0.01\kappa$. This means that the large ones (say, with eigenvalues at least $\kappa/2$) are well preserved in a multiplicative sense. It turns out the guarantees are enough to plug into the  recursive preconditioning scheme.

The step above allows us to find a preconditioning matrix $A$
 so that $I \preceq A\Sigma A \preceq O(1)I$.
 With  the covariance matrix being well-conditioned,  
 we observe that  
 estimating it can be reduced to mean estimation. 
 Indeed, when $\Sigma$ has largest eigenvalue $\kappa$, a basic fact about the moments of a Gaussian random variable $X\sim \mathcal{N}(0,\Sigma)$ is that   $X X^\top$ (with mean $\Sigma$) has  covariance of spectral norm at most $2\kappa^2$.
 Now after preconditioning we have    $\kappa = O(1)$. Thus, we can apply the pure DP algorithm 
 of \cite{hopkins2021efficient}
 again on $X X^\top$. 
Using enough samples, this time we   get an estimate 
$\widehat{\Sigma}'$ such that  
$ \norm{\widehat{\Sigma}' - A\Sigma A}_F \leq 0.01$. 
Finally, by straightforward calculations,   $A^{-1}\widehat\Sigma' A^{-1}$ is a good estimate of $\Sigma$ in the {right} affline-invariant (i.e., relative Frobenius) norm.

The robustness guarantee of our algorithm is inherited from the use of \cite{hopkins2021efficient}. 
We leave it as an open question to get optimal robustness subject to pure DP.

\subsection{Approximate DP Robust Estimation}
\label{sec:approx-dp-tech}
Our approximate DP   results are built upon the framework of \cite{kothari2021private} that reduces private estimation to robust estimation obtained via what they call \emph{witness producing convex relaxations}. The algorithm of \cite{kothari2021private}, though robust and private,  requires $\widetilde{O}(d^8)$ samples, far cry from the (non-private) optimal bound of $\widetilde{O}(d^2)$ samples. 
For   learning a general Gaussian distribution, our algorithm achieves the same performance of \cite{kothari2021private} using only $\widetilde{O}(d^2)$ samples. 
Our approach significantly upgrades the methodology of \cite{kothari2021private}.
This involves two main technical innovations: a regularization scheme via entropy maximization and a new privacy  mechanism based on \textit{Gaussian sampling}. Along the way, we  review  several key techniques in \cite{kothari2021private} that we also employ. 

\paragraph{Randomizing the Outlier Rate}
A common recipe to design DP algorithm is to devise a stable algorithm whose outputs are close when given two adjacent datasets.     
Then standard DP mechanisms, such as the Gaussian mechanism, work by adding noise to the output with the noise level calibrated to the sensitivity of the algorithm. 
Therefore, the basic idea of \cite{kothari2021private}  is to \textit{stabilize} the following prototypical  convex programming approach from the robust statistics literature. Consider an inefficient algorithm which given an $\eta$-corrupted dataset $Y = \{Y_i\}_{i=1}^n$:\begin{enumerate}[(i)]
    \item Finds a witness set $X$ of $n$ points   such that the uniform distribution on $X$ has subgaussian fourth moments and $|Y \cap X| \geq (1-\eta)n$. Rejects if no such $X$ exists. 
    \item Returns the mean and covariance of $X$.
\end{enumerate}
Such (exponential-time) witness checking procedure is statistically optimal.  
It turns out, moreover, that this this scheme can be encoded  by an (efficiently-solvable) Sum-of-Squares (SoS) program. Moreover, the solution of the SoS relaxation retains   the  optimal statistical guarantee \cite{kothari2018robust, kothari2022polynomial}. However, the algorithm is unstable. Suppose $Y$ contains a {unique} witness set $X$ of exactly $(1-\eta)n$ points.  Then if $Y'$ is adjacent to $Y$ but missing a point from $X$, $Y'$ may not contain a feasible witness set at all and therefore will be rejected by the witness checking algorithm.  Hence,  there must be adjacent $Y,Y'$ such that the procedure accepts one but rejects the other. To patch this, \cite{kothari2021private} privately selects a outlier rate $\eta'$ (close to the true $\eta$) as an input parameter to witness checking. 

\paragraph{Stability from Entropy Regularization} 
Unfortunately, this alone does not resolve the problem. 
Another key barrier is that there may be multiple differing solutions to the witness checking constraints. 
Indeed, \cite{kothari2021private} observes that even in
one dimension, the simple witness checking algorithm has feasible solutions of   variance $1+O(\eta)$ and  of $1-O(\eta)$. 
To exclude such differing solutions, the crucial idea from \cite{kothari2021private} is to modify the convex relaxations for robust estimation to satisfy strong sensitivity upper bounds. 
This is accomplished via a strongly convex regularization on the weight vector on the input sample computed by the convex program.
Intuitively, strong convexity ensures that the optimal  solution is stable against slight perturbation in the input---particularly when given two neighboring datasets.  
For that, \cite{kothari2021private} simply chooses to minimize a $2$-norm objective $\|w\|_2^2$, where $w$ is the (soft) indicator variable for the witness set $X$ that the convex program optimizes for.
This is a natural choice of regularizer, since intuitively $\|\cdot\|^2_2$ 
encourages ``flat'' vectors and thus can be seen as a surrogate for entropy.
However, we observe that the strong convexity parameter of the $2$-norm objective is rather weak. This   leads to quadratically worse bounds in the sample complexity.

To improve upon this, we propose to directly use the (unnormalized) entropy function as an regularizing objective.
 To analyze the sensitivity of the stabilized convex program, \cite{kothari2021private} directly exploits the general strong convexity of its objective with respect to $2$-norm (which, via a Cauchy-Schwarz inequality, also yields a bound on the $1$-norm). 
 
 The unnormalized entropy already satisfies such a strong convexity property directly with respect to the $1$-norm -- this is a simple generalization of the Pinsker's inequality. However, a vanilla application of the generalized Pinsker's inequality yields a quadratically lossy bound in our case giving us no benefit on top of the  squared $\ell_2$  norm regularization in~\cite{kothari2021private}. 
 
Instead, we exploit the fact that in the regime we invoke it,  Pinsker's inequality is far from tight.  By a more direct analysis we get a essentially lossless strong convexity for the weight vectors encountered in our analysis. As a result,  we obtain a strong $\sim \widetilde{O}(1/n)$ bound on the sensitivity of our relaxation. This bound is almost optimal up to a logarithmic factor and improves quadratically on the $\approx \frac{1}{\sqrt{n}}$ bound proven in~\cite{kothari2021private}.

\paragraph{Replacing Noise-Injection by Gaussian Sampling} Notice that throughout the discussion above, we have only touched upon the sensitivity of the convex programming-based algorithm. 
To eventually privatize its output, typically one needs to add noise, for example, via the Gaussian mechanism.
Specifically, in the context of simpler problems such as estimating the mean of a spherical covariance distribution as studied in prior works, simple and natural mechanisms work by adding independent Gaussian noise to each entry of the  estimated mean. For covariance, however, such entrywise noise addition   can completely destroy the eigenstructure and thus lead to arbitrarily large total variation error. 
Therefore, as observed by \cite{kothari2021private}, one needs mechanisms that respect the eigenvalue structure---in particular, one must add smaller noise to the directions of lower variance. 
The issue is that this causes privacy leakage, since the estimate is non-privately computed from solving the convex relaxation. 
Nevertheless, \cite{kothari2021private} found a way to add such an estimate-dependent noise and still guarantee privacy. 
Their mechanism, however, is complicated and sub-optimal and results in a further quadratic loss in the sample complexity.

Instead, we replace the sophisticated noise-addition mechanism  of \cite{kothari2021private} by   a   novel and extremely simple mechanism: release  $\widetilde O(d^2)$ i.i.d.\ samples from Gaussian distribution $\mathcal{N}(0,\widehat{\Sigma})$, with $\widehat{\Sigma}$ being the estimated covariance. We show that this simple mechanism gives optimal approximate DP guarantees. 
By a standard concentration inequality, taking the empirical covariance of these released samples approximates $\widehat{\Sigma}$ well, which in turn is guaranteed to be close to the true $\Sigma$.
This results in a further quadratic improvement in the sample complexity over \cite{kothari2021private}, completing our final sample complexity bound.

Our overall analysis of the convex program builds upon \cite{kothari2021private} and prior work in the literature of robust statistics \cite{kothari2018robust,kothari2022polynomial}. At a conceptual level, our work demonstrates a quantitatively strengthened connection between robustness and privacy. 
\section{Preliminaries}
\label{sec:prelim}

\paragraph{Notation} Throughout this paper, we will use $X$ to denote an i.i.d. (uncorrupted) sample of $n$ points in $\R^d$ and $Y$ to denote its $(1-\alpha)$-corruption. For any finite set $S$ of points, we will use $\E_{s \sim S} f(s)$ to denote the empirical average of $f(s)$ as $s$ varies uniformly over $S$.  We denote the  $d$-by-$d$ identity matrix by $I_d$. Let $\mathbb{S}^d_+$ denote the set of $d$-by-$d$ positive semidefinite (PSD)   matrices. For a matrix $A$, we use $\|A\|,\|A\|_2$ to denote the spectral norm of $A$ and $\|A\|_F$ denotes its Frobenius norm. For PSD matrices $A,B$, we write $A \preceq B$ if $B -A$ is PSD.  For a PSD matrix $\Sigma$ and vector $x$, define $\|x\|_{\Sigma}=\left\|\Sigma^{-1 / 2} x\right\|_2$. For a matrix $X$, define $\|X\|_{\Sigma}=\left\|\Sigma^{-1 / 2} X \Sigma^{-1 / 2}\right\|_F$.  For    $M \in \R^{m\times n}$ and    $N \in \R^{m'\times n'}$, we define $M\otimes N$ to be the standard $mm' \times nn'$ matrix given by the {Kronecker product} of $M$ and $N$. Let $\Delta_d=\left\{p \in \mathbb{R}_{+}^d, \sum_{i=1}^d p_i=1\right\}$  denote the probability simplex in $\R^d$.

\subsection{Computational Model and Numerical Issues}  \label{sec:numerical-issues}
Our algorithms work in the standard word RAM model. 
In this model, the input sample $Y$ is given to the algorithm after
truncating the real numbers to rational numbers with $\poly(d)$ bits of precision.
The running time  of the algorithm is a function  of the total bit complexity of the representation of the input.

Throughout this work, we assume $\Sigma \geq 2^{-\operatorname{poly}(d)} I$.
This assumption is due to the numerical issue---namely, that truncation of $Y$ to rational numbers, in general, does not allow recovering eigenvalues of $\Sigma$
that are not representable in polynomially many bits of precision. 


\subsection{Entropy Potential Function}
\begin{definition}[Unnormalized entropy]
\label{def:ent} 
For $x \in \R^n_{\geq 0}$, define the unnormalized entropy function $\Ent (x) =  \sum_{i=1}^n x_i \log(1/x_i) + x_i$. Following standard convention, we take $0 \log (1/0) = 0$.
\end{definition}



\begin{lemma}[Lemma 3.21 of \cite{kothari2021private}]
\label{lem:renormalization}
Suppose $x, y\in[0, 1]^n$ such that $\sum_ix_i, \sum_i y_i\geq n/2$ and
$\norm{x-y}_1 \leq \beta n$ for $\beta\leq 1/10$. Let $\bar{x} = x/\norm{x}_1$ and
$\bar{y} = y/\norm{y}_1$ be the normalized versions of $x, y$. Then,
$$
\norm{\bar{x}-\bar{y}}_1 \leq 6\beta.
$$
\end{lemma}

\subsection{Differential Privacy}




\begin{theorem}[Basic Composition~\cite{DworkMNS06}]
\label{thm:basic-comp}
For every $\epsilon, \delta \geq  0$ and $k \in \N$, the class of $(\epsilon,\delta)$-DP mechanisms is $\left(k\epsilon, k\delta\right)$-DP under $k$-fold adaptive compositions.
\end{theorem}

\begin{theorem}[Variant of Parallel Composition \cite{McSherry10}]
\label{thm:parallel-comp}
Let $\cM : \cO \times \cY \to \cO$ be an $\epsilon$-DP mechanism that takes as input some parameters $w \in \cO$ and a dataset $Y \in \cY$ and outputs $\cM(w, Y) \in \cO$. For any $k \in \N$, let $\cM_k : \cO \times \cY^k \to \cO^{k+1}$ be a mechanism that takes as input parameters $w_0 \in \cO$ and a dataset $Y = (Y_1, \ldots, Y_k) \in \cY^k$ partitioned into $k$ disjoint components where $k$ does
not depend on $Y$.

Given $w_0\in \cO$, the mechanism $\cM_k$ (adaptively) computes $w_i = \cM(w_{i-1}, Y_i)$ for each $i \in [k]$, and outputs $(w_0, w_1, \ldots, w_k)$. Then $\cM_k$ satisfies $\epsilon$-DP.
\end{theorem}

\begin{proof}
Fix any sequence $w = (w_0, w_1, \ldots, w_k) \in \cO^{k+1}$ and let $Y = (Y_1, \ldots, Y_k)$ and $Y' = (Y_1', \ldots, Y_k')$ be neighboring datasets in $\cY^k$.
We need to show that the ratio
$\frac{\Pr[\cM_k(w_0, Y) = w]}{\Pr[\cM_k(w_0, Y') = w]}$
lies between $e^{-\epsilon}$ and $e^\epsilon$. 

Since $k$ does not depend on any dataset and $Y$ is partitioned into disjoint
subsets, there must exist at most one partition that is neighboring.
To this end and without loss of generality, let $j \in [k]$ be an index such that $Y_j,Y_j' \in \cY$ are neighboring datasets and $Y_i = Y_i'$ for all $i \neq j$. Therefore, by
independence,
\[\frac{\Pr[\cM_k(w_0, Y) = w]}{\Pr[\cM_k(w_0, Y') = w]}
= \prod_{i=1}^k \frac{\Pr[\cM(w_{i-1}, Y_i) = w_i]}{\Pr[\cM(w_{i-1}, Y_i') = w_i]}.\]
Thus,
in the rightmost product above, only the $j$\textsuperscript{th} of the $k$ factors may differ from $1$. Since $\cM$ is $\epsilon$-DP, this $j$\textsuperscript{th} factor lies between $e^{-\epsilon}$ and $e^{\epsilon}$. Thus, $\cM_k$ is also $\epsilon$-DP.
\end{proof}


\begin{theorem}[Advanced Composition, Theorem III.3 of \cite{dwork2010boosting}]
\label{thm:adv-comp}
For every $\epsilon, \delta, \delta' > 0$ and $k \in \N$, the class of $(\epsilon,\delta)$-DP mechanisms is $\left(\sqrt{2k \log(1/\delta')}\epsilon + k\epsilon(e^\epsilon -1), k\delta + \delta'\right)$-DP under $k$-fold adaptive compositions.
\end{theorem}

\begin{definition}[Differential Privacy under Condition] \label{def:dp-cond}
An algorithm $\cM: \cY \to \cO$ is said to be \emph{$(\eps, \delta)$-differentially private under condition $\Psi$} (or \emph{$(\eps, \delta)$-DP under condition $\Psi$}) for $\eps, \delta > 0$ iff, for every $S \subseteq \cO$ and every neighboring datasets $Y, Y'$ both satisfying $\Psi$, we have
$$\Pr[\cM(Y) \in S] \leq e^\eps \cdot \Pr[\cM(Y') \in S] + \delta.$$
\end{definition}

\begin{lemma}[Composition for Algorithm with Halting, Lemma 3.30 of \cite{kothari2021private}] \label{lem:composition-with-halt}
Let $\cM_1: \cY \to \cO_1 \cup \{\perp\}, \cM_2: \cO_1 \times \cY \to \cO_2 \cup \{\perp\}, \dots, \cM_k: \cO_{k-1} \times \cY \to \cO_k \cup \{\perp\}$ be algorithms. Furthermore, let $\cM$ denote the algorithm that proceeds as follows (with $o_0$ being empty): For $i = 1, \dots, k$, compute $o_i = \cM_i(o_{i - 1}, Y)$ and, if $o_i = \perp$, halt and output $\perp$. Finally, if the algorithm has not halted, then output $o_k$.

Suppose that:
\begin{itemize}
\item For any $1\leq i < k$, we say that $Y$ satisfies the condition $\Psi_i$ if running the algorithm on $Y$ does not result in halting after applying $\cM_1, \cM_2, \dots, \cM_i$.
\item $\cM_1$ is $(\eps_1, \delta_1)$-DP.
\item $\cM_i$ is $(\eps_i, \delta_i)$-DP (with respect to neighboring datasets in the second argument) under condition $\Psi_{i-1}$ for all $i = \{2, \dots, k\}$.
\end{itemize}
Then, $\cM$ is $\left(\sum_{i \in [k]} \eps_i, \sum_{i \in [k]} \delta_i\right)$-DP.
\end{lemma}

\begin{theorem}[Gaussian mechanism,~\cite{DworkMNS06, NikolovTZ13}]
\label{thm:gaussian-mech}
Let $f: (\R^d)^{\otimes n} \rightarrow \R^k$ satisfy
\[
\max_{\calX, \calX'} \norm{f(\calX) - f(\calX')}_2 \leq \Delta,
\]
where $\calX, \calX'$ are neighboring datasets. Then
for any $\eps\in(0, 1)$,
$M(\calX) = f(\calX) + \sigma Z$ is $(\epsilon,\delta)$-DP, where $Z \sim \calN(0,I_k)$ and $\sigma = \sqrt{2 \ln (1.25/\delta)} \Delta/\epsilon$.
\end{theorem}


\begin{definition}[$f$-Divergence Family~\cite{CI67, AS66}]
For any convex, lower semi-continuous function
$f:\reals^+\rightarrow\reals$, the 
{$f$-Divergence} between two probability measures
$P, Q$ is defined as
$$
D_f(P\Vert Q) = \int_\calX f\left(\frac{dP}{dQ}\right)dQ
= \int_\calX q(x)f\left(\frac{p(x)}{q(x)}\right)dx.
$$
\end{definition}

The family of $f$-divergences includes the
Jensen-Shannon, total variation, and hockey-stick
divergences. 

\begin{definition}[Hockey-Stick divergence~\cite{SasonV16}]
Let $p,q$ be probability density functions on $\R^d$ and $\gamma \geq 0$. The hockey-stick divergence $D_\gamma (p,q)$ between $p,q$ is defined as
\[
D_\gamma (p,q) = \int_{x \in \R^d} [p(x) - \gamma q(x)]_+ dx,
\]
where $[c]_+ = \max(c,0)$.
\label{def:hsd}
\end{definition}

We now state a consequence of the hockey-stick divergence bound for
differential privacy:

\begin{fact}
\label{fact:hockeystick-to-dp}
Let $\calM$ be a randomized algorithm whose output is in $\R^d$. Then $\calM$ is $(\epsilon, \delta)$-DP if
and only if
for any neighboring datasets $\calX, \calX'$ it holds that $D_{e^\eps}(\calM(\calX), \calM(\calX')) \leq \delta$
and
$D_{e^\eps}(\calM(\calX'), \calM(\calX)) \leq \delta$.
\end{fact}

The following condition is also useful for proving that a mechanism satisfies $(\epsilon, \delta)$-DP.

\begin{lemma}[Lemma 1.5 in~\cite{Vadhan17},
Section 1.1 of \cite{bun2016concentrated}]
\label{thm:privacy-loss-tail-implies-dp}
For a (randomized) mechanism ${\cal M}$ and datasets $x, y$, define the function \[f_{xy}(z) = \log\left(\frac{\Pr[{\cal M}(x) = z]}{\Pr[{\cal M}(y) = z]}\right).\]
If \(\Pr[f_{xy}({\cal M}(x)) > \varepsilon] \le \delta\) for all adjacent datasets $x, y$, then ${\cal M}$ is $(\varepsilon, \delta)$-DP.
\end{lemma}


One of the most generic mechanisms used to satisfy pure differential privacy is the
\textit{Exponential Mechanism}:

\begin{theorem}[Exponential Mechanism~\cite{McSherryT07}]
Let $\calX\sim\calX'\in\calO^k$ denote two neighboring datasets.
Consider any arbitrary utility function $u:\calO^k\times\calR\rightarrow\reals$
with global sensitivity $\Delta_u = \max_{\calX\sim\calX', r}|u(\calX, r)-u(\calX', r)|$. For any dataset $\calX$, 
the exponential mechanism outputs $r\in\calR$ with probability
$\propto\exp(\frac{\eps\cdot u(\calX, r)}{2\Delta_u})$.

Furthermore, the exponential mechanism satisfies $\eps$-DP.

\end{theorem}

In general, the exponential mechanism is not computationally
efficient to implement.
We cite a recent result on pure DP mean estimation (on bounded second moment distributions), due to \cite{hopkins2021efficient},
that presents computationally efficient implementations of the exponential
mechanism via a Sum-of-Squares approach.
Their main procedure is  outlier-robust with a corruption level of $\eta$, at the cost of increasing the estimation error by an additive $O(\sqrt{\eta})$.
\begin{theorem}[Pure DP mean estimation, Theorem 1.2 of \cite{hopkins2021efficient}]
\label{thm:pure-dp-mean-est}
For every $n,d \in \N$ and $R, \alpha, \eps, \beta > 0$ there is a polynomial-time $\eps$-DP algorithm $\textsf{PureDPMean}$ such that for every distribution $D$ on $\R^d$ such that $\norm{\E_{X \sim D} X}_2 \leq R$ and $\Cov_{X \sim D} (X) \preceq I_d$, given $X_1, \ldots, X_n \sim D$, with probability at least $1-\beta$ the algorithm outputs $\widehat{\mu}$ such that $\norm{\widehat{\mu} - \E_{X \sim D} X}_2 \leq \alpha$ so long as
\[
n \geq \widetilde{O}\left( \frac{d + \log(1/\beta)}{\alpha^2 \epsilon} + \frac{d \log (R) + \min(d, \log(R)) \cdot \log (1/\beta)}{\epsilon}\right).
\]
Furthermore, if an $\eta$-fraction of the samples are adversarially corrupted, the algorithm
maintains the same guarantee, at the cost now that $\|\widehat{\mu} -\E_{X\sim D} X \| \leq \alpha + O(\sqrt{\eta})$.
\end{theorem}



%

\subsection{Basic Tools from Probability Theory}

\begin{definition}[Total Variation Distance]\label{def:TV}
For any two distributions $P,Q$ over $\reals^d$,
the \textit{total variation distance}  $\dtv$  is defined as 
\[
\dtv(P, Q) = \sup_{S\subseteq\reals^d}|P(S) - Q(S)|.
\]
Moreover, it can be verified that
\[
\dtv(P, Q) = \frac{1}{2}\int_{\reals^d}|p(x) - q(x)|dx.
\]
\end{definition}

The following lemma shows how to convert parameter estimation to distribution estimation (in total variation distance), for the multidimensional Gaussian distribution.
\begin{lemma}[Parameter closeness implies distribution closeness;  see Lemma 2.9 of \cite{kamath2019privately}]
\label{lemma:dtv-gaussians}
Let $\alpha \geq 0$, $\mu, \widehat{\mu} \in \R^d$ and $\Sigma, \widehat{\Sigma} \in \R^{d \times d}$ be PSD. Suppose that \[\norm{\Sigma^{-1/2} (\mu - \widehat{\mu})}_2 \leq \alpha \text{  and  } \norm{\Sigma^{-1/2} \widehat{\Sigma}  \Sigma^{-1/2} - I}_F \leq \alpha.\] Then $\dtv \left(\calN(\mu, \Sigma), \calN\left(\widehat{\mu}, \widehat{\Sigma}\right)\right) \leq O(\alpha)$.
\end{lemma}

The  fact below follows from  Theorem 4.12 of \cite{diakonikolas2019robust}.
\begin{fact}
\label{fact:cov-outer-product}
Let $X \sim \calN(0, \Sigma)$ and $Y = XX^T$. Then $\Cov(Y) \preceq 3 \Sigma \otimes \Sigma$. 
\end{fact}

For the Gaussian distribution, the empirical estimator of its covariance attains the following statistical accuracy. This fact is needed for analyzing our Gaussian sampling mechanism.
\begin{theorem}[Empirical covariance estimator for Gaussian \cite{vershynin2018high}]
\label{thm:emp-cov-concentration}
Let $\Sigma \in \R^{d \times d}$ be PSD, $X_1,\ldots,X_n \sim \calN(0,\Sigma)$ be \iid and $\widehat{\Sigma} = \frac{1}{n} \sum_{i=1}^n X_i X_i^T$. Then with probability $1-\gamma$, it holds that
\[
\norm{\Sigma^{-1/2} \widehat{\Sigma} \Sigma^{-1/2} - I}_F \leq \rho,
\]
for some $\rho = O\left(\sqrt{\frac{d^2 + \log(1/\gamma)}{n}} + \frac{d^2 + \log(1/\gamma)}{n} \right)$.
\end{theorem}

\paragraph{Concentration inequalities} We cite some standard concentration inequalities.
\begin{theorem}[Hanson-Wright Inequality \cite{rudelson2013hanson}]
\label{thm:hanson-wright}
Let $g \sim \calN(0, I_d)$ and $A \in \R^{d \times d}$. Then for some absolute constant $c > 0$ and every $t \geq 0$ it holds that
\[
\Pr (|g^T A g - \E g^T A g| \geq t) \leq 2\exp\left(-c \min \left(  \frac{t^2}{\norm{A}_F^2}, \frac{t}{\norm{A}_2} \right) \right).
\]
\end{theorem}

\begin{lemma}[Chi-squared tail bound]
\label{lem:chi-squared-tail}
Let $Z \sim \calN(0,I_d)$. Then there is some constant $c > 0$ such that for all $ t \geq 0$,
\[
\Pr \left(\norm{Z}_2^2 \geq d + t\right) \leq \exp\left(-c \min\left\{t^2/d, t\right\}\right).
\]
\end{lemma}

\paragraph{Sub-exponential random variables} We now cite some properties of sub-exponential random variables.

\begin{definition}[Sub-exponential random variable; Definition 2.7 in \cite{wainwright2019high}]
\label{def:sub-exponential}
A random variable $X$ with mean $\mu = \E[X]$ is \emph{sub-exponential} if there are non-negative parameters $(\nu, \alpha)$ such that \[\E\left[e^{\lambda(X - \mu)}\right] \le e^{\frac{\nu^2\lambda^2}{2}} \qquad \text{for all }|\lambda| < \frac{1}{\alpha}.\]
\end{definition}

\begin{lemma}[Sub-exponential tail bound, Proposition 2.9 in \cite{wainwright2019high}]
\label{thm:sub-exponential-tail}
Suppose that $X$ is sub-exponential with parameters $(\nu, \alpha)$. Then \[\Pr[X - \mu \ge t] \le \max\left\{e^{-\frac{t^2}{2\nu^2}}, e^{-\frac{t}{2\alpha}}\right\}.\]
\end{lemma}

Let $\chi_d^2$ denote a Chi-squared random variable with $d$ degrees of freedom.
\begin{lemma}[$\chi^2_1$ sub-exponential parameters, Example 2.11 in \cite{wainwright2019high}]
\label{thm:chi-squared-sub-exponential}
A chi-squared random variable with $1$ degree of freedom ($\chi_1^2$)  is sub-exponential with parameters $(\nu, \alpha) = (2, 4)$.
\end{lemma}

\begin{lemma}[Sub-exponential parameters of independent sum, Chapter 2 of \cite{wainwright2019high}]
\label{thm:sub-exponential-sum}
Consider an independent sequence $X_1, \ldots, X_k$ of random variables, such that $X_i$ is sub-exponential with parameters $(\nu_i, \alpha_i)$. Then the variable $\sum_{i=1}^k X_i$ is sub-exponential with the parameters $(\nu_*, \alpha_*)$, where \[\alpha_* = \max_{i \in [k]} \alpha_i \qquad\text{and}\qquad \nu_* = \sqrt{\sum_{i=1}^k \nu_i^2}.\]
\end{lemma}

\subsection{Sum-of-Squares Optimization}


We refer the reader to the monograph~\cite{TCS-086} for a detailed exposition 
of the sum-of-squares method and its usage in average-case algorithm design. A \emph{degree-$\ell$ pseudo-distribution} is a finitely-supported function $D:\R^n \rightarrow \R$ such that $\sum_{x} D(x) = 1$ and $\sum_{x} D(x) f(x)^2 \geq 0$ for every polynomial $f$ of degree at most $\ell/2$. We define the \emph{pseudo-expectation} of a function $f$ on $\R^d$ with respect to a pseudo-distribution $D$, denoted $\pE_{D(x)} f(x)$, as $\pE_{D(x)} f(x) = \sum_{x} D(x) f(x)$. 

The degree-$\ell$ pseudo-moment of a pseudo-distribution $D$ is the tensor $\E_{D(x)} (1,x_1, x_2,\ldots, x_n)^{\otimes \ell}$ with entries corresponding to pseudo-expectations of monomials of degree at most $\ell$ in $x$. The set of all degree-$\ell$ moment tensors of degree $d$ pseudo-distributions is also closed and convex.

\begin{definition}[Constrained pseudo-distributions]
  Let $D$ be a degree-$\ell$ pseudo-distribution over $\R^n$.
  Let $\cA = \{f_1\ge 0, f_2\ge 0, \ldots, f_m\ge 0\}$ be a system of $m$ polynomial inequality constraints.
  We say that \emph{$D$ satisfies the system of constraints $\cA$ at degree $r$} (satisfies it $\eta$-approximately, respectively), if for every $S\subseteq[m]$ and every sum-of-squares polynomial $h$ with $\deg h + \sum_{i\in S} \max\set{\deg f_i,r} \leq \ell$, $\pE_{D} h \cdot \prod _{i\in S}f_i  \ge 0$.
  We say that $D$ satisfies (similarly for approximately satisfying) $\cA$ (without mentioning degree) if $D$ satisfies $\cA$ at degree $r$.
\end{definition}

\paragraph{Sum-of-squares proofs} A \emph{sum-of-squares proof} that the constraints $\{f_1 \geq 0, \ldots, f_m \geq 0\}$ imply the constraint $\{g \geq 0\}$ consists of  polynomials $(p_S)_{S \subseteq [m]}$ such that $g = \sum_{S \subseteq [m]} p_S \cdot \Pi_{i \in S} f_i$.

We say that this proof has \emph{degree $\ell$} if for every set $S \subseteq [m]$, the polynomial $p_S \Pi_{i \in S} f_i$ has degree at most $\ell$ and write: 
\begin{equation}
  \{f_i \geq 0 \mid i \leq r\} \sststile{\ell}{}\{g \geq 0\}
  \mper
\end{equation}



\begin{fact}[Soundness]
  \label{fact:sos-soundness}
  If $D$ satisfies $\cA$ for a degree-$\ell$ pseudo-distribution $D$ and there exists a sum-of-squares proof $\cA \sststile{r'}{} \cB$, then $D$ satisfies $\cB$ at degree $rr' +r'$.
\end{fact}


\begin{definition}[Total bit complexity of Sum-of-Squares Proofs]
Let $p_1, p_2, \ldots, p_m$ be polynomials in indeterminate $x$ with rational coefficients. 
For a polynomial $p$ with rational coefficients, we say that $\{p_i \geq 0\}$ derives $\{p\geq 0\}$ in degree $k$ and total bit complexity $B$ if $p = \sum_i q_i^2 + \sum_i r_i p_i$ where each $q_i^2,r_i$ are polynomials with rational coefficients of degree at most $k$ and $k-deg(p_i)$ for every $i$, and the total number number of bits required to describe all the coefficients of all the polynomials $q_i, r_i, p_i$ is at most $B$.
\end{definition}

There is an efficient separation oracle for moment tensors of pseudo-distributions that allows approximate optimization  of linear functions of pseudo-moment tensors approximately satisfying constraints. %
The \emph{degree-$\ell$ sum-of-squares algorithm} optimizes over the space of all degree-$\ell$ pseudo-distributions that approximately satisfy a given set of polynomial constraints:

\begin{fact}[Efficient Optimization over Pseudo-distributions \cite{MR939596-Shor87,parrilo2000structured,MR1748764-Nesterov00,MR1846160-Lasserre01}]
Let $\eta>0$. There exist an algorithm that for $n, m\in \N$ runs in time $(n+ m)^{O(\ell)} \poly \log 1/\eta$, takes input an explicitly bounded and satisfiable system of $m$ polynomial constraints $\cA$ in $n$ variables with rational coefficients and outputs a level-$\ell$ pseudo-distribution that satisfies $\cA$ $\eta$-approximately. \label{fact:eff-pseudo-distribution}
\end{fact}

\subsection{Analytic Properties of Probability Distributions} \label{sec:analytic-props}

\paragraph{Certifiable Subgaussianity}
We define certifiable subgaussianity and it will be used in analysis of  our approximate DP robust mean estimation algorithm.
\begin{definition}[Certifiable Subgaussianity] \label{def:cert-subgaussianity}
A distribution $D$ on $\R^d$ with mean $\mu_*$ is said to be $2k$-certifiably $C$-subgaussian if there is a degree $2k$ sum-of-squares proof of the following polynomial inequality in $d$-dimensional vector-valued indeterminate $v$:
\[
\E_{x \sim D} \iprod{x-\mu_*,v}^{2k} \leq (Ck)^k \Paren{\E_{x \sim D} \iprod{x-\mu_*,v}^2}^k\mper
\]
Furthermore, we say that $D$ is certifiable $C$-subgaussian if it is $2k$-certifiably $C$-subgaussian for every $k\in \N$.

A finite set $X \subseteq \R^d$ is said to be $2k$-certifiable $C$-subgaussian if the uniform distribution on $X$ is $2k$-certifiably $C$-subgaussian. 
\end{definition}

\paragraph{Certifiable Hypercontractivity of Degree 2 Polynomials} 
Next, we define \emph{certifiable hypercontractivity} of degree-$2$ polynomials that formulates (within SoS)
the fact that higher moments of degree-$2$ polynomials of distributions (such as Gaussians)
can be bounded in terms of appropriate powers of their 2nd moment.

\begin{definition}[Certifiably hypercontractive]\label{def:hc}
A distribution $\mathcal D$ on $\R^d$ with mean $\mu$ and covariance $\Sigma$ is said to have  $2h$-certifiably $C$-hypercontractive 
degree-$2$ polynomials
if for a $d \times d$ matrix-valued indeterminate $Q$ and $\overline{x}=x-\mu$,
\[
\sststile{2h}{Q} \Set{\E_{x \sim D} \left(\overline{x}^{\top}Q\overline{x}-\E_{x \sim D} \overline{x}^{\top}Q\overline{x}\right)^{2h} \leq (Ch)^{2h} \norm{\Sigma^{1/2}Q\Sigma^{1/2}}_F^{2h}}.
\]
\end{definition}



The Gaussian distribution  and its     affine transforms are  known to satisfy $2t$-certfiable $C$-
hypercontractivity  with an absolute constant $C$ for every  $t$~\cite{KauersOTZ14}.

Certifiable hypercontractivity strictly generalizes the better known {certifiable subgaussianity} property (formalized and studied first in~\cite{KS17}) that is the special case of certifiable hypercontractivity of (squares of) linear polynomials, or, equivalently, when $Q = vv^{\top}$ for a vector-valued indeterminate $v$.

\section{Pure DP Covariance Estimation}
\label{sec:pure-dp-cov}
In this section, we give an efficient algorithm for Gaussian covariance estimation under pure differential privacy. 

\paragraph{High-level overview}
The procedure builds upon the recent work by Hopkins, Kamath and Majid \cite{hopkins2021efficient} on pure DP mean estimation.
First, we exploit their algorithm     to precondition the unknown covariance matrix such that it is approximately identity. A key step is  a weak preconditioning algorithm that uses  the result of \cite{hopkins2021efficient}  to improve the conditioning of $\Sigma$ by a constant factor. We recursively apply the construction to strengthen the conditioning of $\Sigma$.  
(This recursive scheme  was first proposed by \cite{kamath2019privately} for approximate DP estimation, but here we use it to obtain pure DP guarantees.)
Once $\Sigma$ is nearly identity, we show that appealing to a pure DP mean estimation algorithm would suffice. 
For that purpose, we resort to \cite{hopkins2021efficient} again and complete the proof.

Our algorithm can be seen as  a reduction to black-box applications to \cite{hopkins2021efficient}. As a result,   our estimator retains the same robustness property of \cite{hopkins2021efficient}, albeit it is sub-optimal for the Gaussian distribution.
\subsection{Weak Private Preconditioning}

As we discussed, the key subroutine  of our algorithm  is a private conditioning procedure. 
Given the  samples, its goal is to output a preconditioning matrix $A$ such that $I \preceq A\Sigma A \preceq 0.99 \kappa I$, where $\kappa$ is the condition number of  the known covariance $\Sigma$. 
In other words, the condition number of $A\Sigma A$ improves over that of $\Sigma$, by a constant factor. 
This guarantee is similar to what appears in the previous literature on private covaraince estimation and subspace recovery \cite{kamath2019privately,singhal2021privately, kamath2022private}. 
However, the algorithms from prior work crucially rely upon the Gaussian mechanism, which only  ensure  \textit{approximate} (or concentrated)  DP.  These results, therefore, do not translate into  pure DP guarantees. 

In this section, we describe and analyze a weak pure DP algorithm for preconditioning the covariance.  
puThe procedure reduces the condition number of the covariance (multiplicatively) by a constant factor. Towards this goal, a simple observation is that the algorithm for {pure} DP mean estimation from \cite{hopkins2021efficient}, applied  na\"ively, can be used  for covariance estimation with an absolute Frobenius norm error guarantee.

\begin{figure}[htbp]
    \centering
    \begin{mdframed}[style=algo]
\begin{enumerate}
    \item \textbf{Input:} Samples $\calX = \{X_1, \ldots, X_n\} \subset \R^d$, condition number $\kappa \geq 1$, accuracy parameter $\alpha > 0$, failure probability $\beta > 0$, privacy parameter $\epsilon > 0$.
    \item Form the set of samples $\calY = \left\{Y_i = \frac{X_i \otimes X_i}{ \sqrt{3} \kappa} : i \in [n]\right\}$.
    \item  Run the algorithm \textsf{PureDPMean} in Theorem \ref{thm:pure-dp-mean-est} (the main result of \cite{hopkins2021efficient}) on input $\calY$ with $R = \sqrt{d/3}$ and $\alpha/\sqrt{3}, \beta, \epsilon > 0$ to obtain an estimate $\tilde{\Sigma}$.
    \item \textbf{Output:} Covariance matrix estimate $\widehat{\Sigma} = \sqrt{3}\kappa \tilde{\Sigma}$.
\end{enumerate}
\end{mdframed}
\caption{\textsf{PureDPMatrixMean} (based on \textsf{PureDPMean} in Theorem \ref{thm:pure-dp-mean-est})}
\label{alg:pure-dp-matrix-mean}
\end{figure}

\begin{theorem}[Pure DP covariance estimation in absolute Frobenius norm]
\label{thm:pure-dp-well-cond}
Let $\alpha > 0$ be an error parameter, $\epsilon > 0$ be a privacy parameter and $d \in \N$. There is a polynomial-time $\epsilon$-DP algorithm \textsf{PureDPMatrixMean} that, given $\eps,\alpha$ and \[n \geq \widetilde{O}\left( \frac{d^2 + \log(1/\beta)}{\alpha^2 \epsilon}\right)\] i.i.d.\ samples $\calX = \{X_1,X_2, \ldots, X_n\}$ from $\calN(0, \Sigma)$ for an unknown $\Sigma \in \R^{d \times d}$ satisfying $\Sigma \preceq \kappa I$, outputs $\widehat{\Sigma} = \widehat{\Sigma}(\calX)$ satisfying \[\norm{\widehat{\Sigma} - \Sigma}_F \leq \alpha \kappa\] with probability at least $1-\beta$.
\end{theorem}

\begin{proof}
Observe that \[\norm{\E Y_i}_F \leq \frac{1}{ \sqrt{3} \kappa} \norm{\Sigma}_F \leq \sqrt{\frac{d}{3}} = R\] by the definition of $Y_i$, and \[\Cov (Y_i) \preceq \frac{1}{\kappa^2} \Sigma \otimes \Sigma \preceq I\] by Fact \ref{fact:cov-outer-product}. By the choice of $n$ and  the error guarantee of \textsf{PureDPMean}, we have that the matrix $\widehat{\Sigma}$ output by Algorithm \ref{alg:pure-dp-matrix-mean} (\textsf{PureDPMatrixMean}) satisfies $\norm{\widehat{\Sigma} - \Sigma}_F \leq \alpha \kappa$ with probability at least $1-\beta$. Finally, the algorithm runs in polynomial time, since \textsf{PureDPMean}    is in polynomial time and it takes linear time to form and rescale the samples.  
\end{proof}

\begin{figure}[htbp]
    \centering
    \begin{mdframed}[style=algo]
\begin{enumerate}
    \item \textbf{Input:} Samples $\calX = \{X_1, \ldots, X_n\} \subset \R^d$, condition number $\kappa \geq 1$, accuracy parameter $\alpha > 0$, failure probability $\beta > 0$, privacy parameter $\epsilon > 0$.
    \item Run Algorithm \ref{alg:pure-dp-matrix-mean} \textsf{PureDPMatrixMean} on input $\calX$ and with parameters $(\alpha = 0.01, \beta, \epsilon, \kappa)$ to obtain an estimate $\widehat{\Sigma}$ of the covariance matrix.
    \item Let $V$ be the span of all eigenvectors of $\widehat{\Sigma}$ attaining eigenvalue at least $\kappa/2$, $\Pi$ be the projector onto $V$ and $\Pi_\perp$ be the projector onto the orthogonal complement of $V$.
    \item \textbf{Output:} Weak preconditioner $A = 1.19\cdot \left(0.9 \Pi + \Pi_\perp\right)$.
\end{enumerate}
\end{mdframed}
\caption{One round, weak private preconditioning algorithm}
\label{alg:precondition-weak}
\end{figure}

We leverage the above observation to design a weak preconditioning algorithm (Algorithm \ref{alg:precondition-weak}) that reduces the condition number of $\Sigma$ by a constant factor. 
Algorithm \ref{alg:precondition-weak} first runs Algorithm \ref{alg:pure-dp-matrix-mean} with an error parameter $\alpha=0.01$. 
Then from the error guarantee of Theorem \ref{thm:pure-dp-well-cond},  we can privately estimate all eigenvalues of $\Sigma$ up to an additive factor of $0.01\kappa$. Finally, 
we run a \textit{partial projection} step, a technique  from \cite{kamath2019privately}.
Informally speaking, the algorithm   partially projects out the eigenvectors associated with large eigenvalues. Intuitively, this shrinks the directions of large variance more so than those of small variance, and thus reduces conditioning number. 

The algorithm is formally described by Algorithm \ref{alg:precondition-weak}  and its   guarantees given below.

\begin{lemma}[Private preconditioning, one round]
\label{thm:pure-dp-preconditioner-one-round}
Let $\epsilon > 0$ be a privacy parameter, $d \in \N$, $\kappa \geq 20$, and $\beta > 0$ be a failure probability. There is a polynomial-time $\epsilon$-DP algorithm that, given $\eps,\beta$ and \[n \geq \widetilde{O}\left(\frac{d^2 + \log (1/\beta)}{\epsilon}\right)\] i.i.d.\ samples $\calX = \{X_1,X_2, \ldots, X_n\}$ from $\calN(0, \Sigma)$ for an unknown $\Sigma \in \R^{d \times d}$ satisfying $I \preceq \Sigma \preceq \kappa I$, outputs $A \in \R^{d \times d}$ such that \[I \preceq A \Sigma A^T \preceq 0.99 \kappa I\] with probability at least $1-\beta$.
\end{lemma}
\begin{proof}
We will show that Algorithm~\ref{alg:precondition-weak} satisfies the claims. First, note that it is $\epsilon$-DP because the algorithm, based on \textsf{PureDPMean}, in Theorem \ref{thm:pure-dp-well-cond} is $\epsilon$-DP and  Algorithm~\ref{alg:precondition-weak} post-preprocesses its output. By our distributional assumptions, with probability $1- \beta$ it holds that $\widehat{\Sigma}$, computed in the second step of Algorithm~\ref{alg:precondition-weak}, satisfies
\begin{equation}\label{eqn:error-puredp}
    \norm{\widehat{\Sigma} - \Sigma}_F \leq 0.01 \kappa.
\end{equation}
This implies that, for any unit vectors $u,v$ it holds that \[\left|u^T \widehat{\Sigma} v - u^T \Sigma v\right| \leq 0.01 \kappa.\] Let $V$ be the subspace spanned by all eigenvectors of $\widehat{\Sigma}$ with corresponding eigenvalue at least $\kappa/2$, $\Pi$ be the projector onto $V$ and $\Pi_\perp$ be the projector onto $V^\perp$, the orthogonal complement of $V$. We will now show that the matrix $A = \gamma \Pi + \Pi_\perp$ satisfies $0.85 I \preceq A \Sigma A \preceq 0.83 \kappa I$ for $\gamma = 0.9$. Then rescaling $A$ by $1.19$  ensures that the conclusion of the lemma holds.

For the upper bound, we have that
\begin{align*}
\norm{A \Sigma A}_2 &\leq \norm{A \widehat{\Sigma} A}_2 + \norm{A \left(\Sigma-\widehat{\Sigma}\right) A}_2 \\
&\leq \norm{A \widehat{\Sigma} A}_2 + \norm{\Sigma-\widehat{\Sigma}}_2 \norm{A}_2^2 \\
&\leq \max \left(\kappa/2,  \gamma^2 \cdot \norm{\widehat{\Sigma}}_2\right) + \norm{\Sigma-\widehat{\Sigma}}_2 \norm{A}_2^2 \\
&\leq \max\left(\kappa/2, \gamma^2 \cdot 1.01 \kappa\right) + 0.01 \kappa \\ &\leq \left(1.01 \gamma^2 + 0.01\right) \kappa \\
&\leq 0.83 \kappa,
\end{align*}
where the fist line is via triangle inequality, the second follows   since the spectral norm is sub-multiplicative, the third  is by the choice of $A$, the fourth line  by the error guarantee of $\widehat \Sigma$ (Equation \ref{eqn:error-puredp}) and $\|A\|_2\leq 1$, and the last two inequalities follow from simple algebra and $\gamma=0.9$.

For the lower bound, consider any unit vector $u \in \R^d$. We will lower bound $u^T A \Sigma A u$ in two different ways and maximize over the two. First, because $\Sigma \succeq I$, we have that
\[
u^T A \Sigma A u \geq u^T A^2 u = \gamma^2 \norm{\Pi u}_2^2 + \norm{\Pi_\perp u}_2^2 \geq \norm{\Pi_\perp u}_2^2.
\]
For the second lower bound, we have
\[
u^T A \Sigma A u \geq \gamma^2 u^T \Pi \Sigma \Pi u + \gamma u^T \Pi_\perp \Sigma \Pi u + \gamma u^T \Pi \Sigma \Pi_\perp u.
\]
Since $\Pi u \in V$, the first term is lower bounded by $.49 \gamma^2 \kappa \norm{\Pi u}_2^2$. For the second (and similarly for the third) term, we have that
\[
\left|\gamma u^T \Pi_\perp \Sigma \Pi u\right| \leq 0.01 \gamma \kappa \norm{\Pi u}_2 \norm{\Pi_\perp u}_2 \leq 0.01 \gamma \kappa \norm{\Pi u}_2. 
\]
Aggregating these bounds, we have that for any unit vector $u \in \R^d$, it holds that
\[
u^T A \Sigma A u \geq \max \left(\norm{\Pi_\perp u}_2^2, \left(0.49\gamma^2 \norm{\Pi u}_2^2 - 0.02 \gamma \norm{\Pi u}_2\right) \kappa \right).
\]
Using the facts that $\kappa \geq 20$ and $\norm{\Pi u}_2^2 + \norm{\Pi_\perp u}_2^2 = 1$, it is straightforward to verify that this lower bound is always at least $0.85$. This completes the proof.
\end{proof}

\subsection{Recursive Private Preconditioning}
Given the weak conditioning algorithm, the natural next step is to recurse. Applying the weak preconditioner for $O(\log \kappa)$ times suffices to put the covariance nearly into identity. 

\begin{figure}[htbp]
    \centering\begin{mdframed}[style=algo]
\begin{enumerate}
    \item \textbf{Input:} Samples $\calX = \{X_1,X_2, \ldots, X_n\} \subset \R^d$, condition number $\kappa \geq 1$, privacy parameter $\epsilon> 0$.
    \item Set $L = O(\log \kappa)$ and partition $\calX = \calX_1 \sqcup \cdots \sqcup \calX_L$ into $L$ subsets of size $n/L$ each.
    \item Set $A_{0} = I_d$ and $\kappa_1 = \kappa$. 
    \item For each $j \in [L]$:
        \begin{enumerate}
            \item Set $A_{<j} = \prod_{k=0}^{j-1} A_{k}$.
            \item Run Algorithm~\ref{alg:precondition-weak} on samples $\left\{A_{<j} X : X \in \calX_j\right\}$ with parameters $\left(\epsilon, \beta = 1/1000L, \kappa_j\right)$ to obtain a matrix $A_j \in \R^{d \times d}$.
            \item Set $\kappa_{j+1} =0.99\kappa_j$.
        \end{enumerate}
    \item Let $A=\prod_{i=j}^L A_i $
    \item \textbf{Output:}   the preconditioning matrix $A$.
\end{enumerate}
\end{mdframed}
\caption{Recursive private preconditioning algorithm}
\label{alg:precondition}
\end{figure}

   Specifically, we show that Algorithm \ref{alg:precondition} satisfies the following guarantees.

\begin{theorem}[Private preconditioning, recursive]
\label{thm:pure-dp-preconditioner}
Let $\epsilon > 0$ be a privacy parameter and $d \in \N$. There is a polynomial-time $\epsilon$-DP algorithm that, given
\[n \geq \widetilde{O}\left(\frac{d^2 \log(\kappa)}{\epsilon}\right)\] 
\iid samples $\calX = \{X_1, \ldots, X_n\}$ from $\calN(0, \Sigma)$ for an unknown $\Sigma \in \R^{d \times d}$ satisfying $I \preceq \Sigma \preceq \kappa I$, outputs a matrix $A \in \R^{d \times d}$ such that \[I\preceq A \Sigma A \preceq 20 I\] with probability at least $0.99$.
\end{theorem}

\begin{proof} 
We will show that Algorithm~\ref{alg:precondition} satisfies the claims. It follows     from  parallel composition (Theorem \ref{thm:parallel-comp}) and the privacy   of the weak preconditioner (Theorem \ref{thm:pure-dp-preconditioner-one-round}) that the algorithm is $\epsilon$-DP.  Assume without loss of generality that $\kappa \geq 20$. By the choice of $L = O(\log(\kappa))$ and $\beta = 1/1000L$ and an application of union bound, with probability as least $0.999$, all $L$ invocations of Algorithm~\ref{alg:precondition-weak}, in step 4(b), succeed. We now condition on this success event.

Let $\Sigma_1= \Sigma$ and $\Sigma_j = A_{<j} \Sigma_j A_{<j}$ for every $j \leq L$. 
By   the guarantee of Algorithm \ref{alg:precondition-weak},  in the $j$-th  iteration we get a preconditioning matrix $A_{j}$  such that 
\begin{equation}
    I \preceq A_{j} \Sigma_j A_{j} \preceq 0.99\kappa_j I.
\end{equation}
By   induction on $j$ and the choice of $L$, we have that $A$ satisfies $I\preceq A \Sigma A \preceq 20 I$.  
\end{proof}

\subsection{Putting it Together}

We can now put everything together and prove one of our primary results, the main statement of Theorem \ref{thm:pure-dp-cov-est}. For convenience, we restate it below as Theorem \ref{thm:pure-dp-cov-est-restate}.

\begin{figure}[htbp]
    \centering\begin{mdframed}[style=algo]
\begin{enumerate}
    \item \textbf{Input:} Samples $\calX = \{X_1, \ldots, X_n\} \subset \R^d$, condition number $\kappa \geq 1$, privacy parameter $\epsilon > 0$.
    \item Compute  the preconditioning matrix $A \in \R^{d \times d}$ using $\{X_1, \ldots, X_{n/2}\}$ as input to Algorithm~\ref{alg:precondition} and privacy parameter $\epsilon/2$.
    \item Run Algorithm \ref{alg:pure-dp-matrix-mean} \textsf{PureDPMatrixMean} on samples $\{A X_{n/2 + 1}, \ldots, A X_n\}$ with privacy parameter $\epsilon/2$ and error parameter $\alpha/20$ to obtain a covariance matrix $\Sigma_1 \in \R^{d \times d}$.
    \item \textbf{Output:} Covariance estimate $\widehat{\Sigma} = A^{-1} \Sigma_1 A^{-1}$.
\end{enumerate}
\end{mdframed}
\caption{Private covariance estimation algorithm}
\label{alg:pure-dp-cov-est}
\end{figure}

\begin{theorem}[Pure DP covariance estimation]
\label{thm:pure-dp-cov-est-restate}
Let $\alpha > 0$ be an error parameter, $\epsilon> 0$ be a privacy parameter, and $d \in \N$. There is a polynomial-time $\epsilon$-DP algorithm that, given $\eps$ and
\[n \geq \widetilde{O}\left(\frac{d^2\log(\kappa)}{\epsilon} + \frac{d^2}{\alpha^2 \epsilon} + \frac{d\log(R)}{\epsilon}\right)\] 
i.i.d.\ samples $\calX = \{X_1, \ldots, X_n\}$ from $\calN(0, \Sigma)$ for an unknown $\Sigma \in \R^{d \times d}$ satisfying $I \preceq \Sigma \preceq \kappa I$, outputs $\widehat{\Sigma} = \widehat{\Sigma}(\calX)$ satisfying \[\norm{\Sigma^{-1/2} \widehat{\Sigma} \Sigma^{-1/2} - I}_F \leq \alpha\] with probability at least $0.99$.
\end{theorem}
\begin{proof}
We will show that Algorithm~\ref{alg:pure-dp-cov-est} satisfies the claims. First, note that the algorithm is $\epsilon$-DP by applying basic composition (Theorem \ref{thm:basic-comp}) to the privacy guarantees  of Theorem \ref{thm:pure-dp-preconditioner} and Theorem \ref{thm:pure-dp-well-cond}. Moreover, the algorithm is in polynomial time, since Algorithms \ref{alg:pure-dp-matrix-mean} and \ref{alg:precondition} both run in polynomial time.
 
It now suffices to prove the    utility  guarantees. Let  $\calX$ consist of $n$ \iid samples from $\calN(0, \Sigma)$, for an unknown $\Sigma$ satisfying $I \preceq \Sigma \preceq \kappa I$. By~Theorem \ref{thm:pure-dp-preconditioner}, if
\[n \geq \widetilde{O}\left(\frac{d^2\log(\kappa)}{\epsilon}\right)\] 
the preconditioner $A$ computed in step 2 of Algorithm~\ref{alg:pure-dp-cov-est} satisfies $I \preceq A\Sigma A \preceq  20I$ with probability 0.99. Conditioned on this event, the samples $\{A X_{n/2+1}, \ldots, A X_n\}$ are \iid according to $\calN(0, A \Sigma A)$. Since $A \Sigma A\preceq 20 I$ and $n \geq \widetilde{O}\left(\frac{d^2}{\alpha^2 \epsilon}\right)$, Algorithm \ref{alg:pure-dp-matrix-mean} will return a covariance matrix $\Sigma_1$ satisfying\[
\norm{\Sigma_1 - A \Sigma A}_F \leq \alpha,
\]
with probability at least $0.99$.
We translate this into a relative Frobenius distance guarantee as follows. Notice that
\begin{align*}
\alpha \geq \norm{\Sigma_1 - A \Sigma A}_F &= \norm{A \Sigma^{1/2} \left(\Sigma^{-1/2} A^{-1} \Sigma_1 A^{-1} \Sigma^{-1/2} - I\right) \Sigma^{1/2} A}_F \\
&\geq \lambda_{\min}\left(A\Sigma A\right) \norm{\Sigma^{-1/2} A^{-1} \Sigma_1 A^{-1} \Sigma^{-1/2} - I}_F,
\end{align*}
where the last step uses the  sub-multiplicativity of the Frobenius norm.
By definition of the algorithm, the final estimate of $\Sigma$ is $\widehat{\Sigma} = A^{-1} \Sigma_1 A^{-1}$. 
Plugging this into the inequality above, we have 
\begin{align}\label{eqn:translate-2}
\alpha \geq \norm{\Sigma_1 - A \Sigma A}_F  
\geq \lambda_{\min}\left(A\Sigma A\right) \norm{\Sigma^{-1/2} \widehat{\Sigma} \Sigma^{-1/2} - I}_F.
\end{align}
Observe that  $\lambda_{\min}\left(A\Sigma A\right) \geq  1$,  since $I \preceq A \Sigma A$. 
Applying this fact and  rearranging the inequality \ref{eqn:translate-2}, we get that
\begin{equation}
\norm{\Sigma^{-1/2} \widehat{\Sigma} \Sigma^{-1/2} - I}_F \leq \alpha, 
\end{equation}
completing the proof.
\end{proof}

\subsection{Application: General Pure DP Gaussian Estimation}
We now show how to estimate a high-dimensional Gaussian with unknown mean and covariance in statistical distance, under pure DP. 
This is by combining our result on   private covaraince estimation and the prior work on mean estimation \cite{hopkins2021efficient}. The argument is standard: we simply  estimate the mean and covariance separately and  apply  Lemma \ref{lemma:dtv-gaussians} that converts closeness in parameters  to closeness in distribution.

\paragraph{Pure DP mean estimation} The first step of the algorithm is to privately estimate the mean.
The idea is simple and similar to \cite{kamath2019privately}. If $\Sigma$ were known, then we can apply $\Sigma^{-1/2}$ to the samples and run the mean estimation algorithm (\textsf{PureDPMean}) of \cite{hopkins2021efficient} on input $\left\{\Sigma^{-1/2} X_i\right\}_{i=1}^n$. For $X_i \sim \cN(\mu, \Sigma)$, we have $\Sigma^{-1 / 2} X_i \sim \mathcal{N}\left(\Sigma^{-1 / 2} \mu, I\right)$. Thus, the output $\widehat{\mu}$ of \textsf{PureDPMean} satisfies that $\left\|\Sigma^{-1 / 2}(\mu-\widehat{\mu})\right\|_2 \leq \alpha$, which is what we need for distribution estimation (Lemma \ref{lemma:dtv-gaussians}). 

In the setting when $\Sigma$ is unknown, we apply our our preconditioning algorithm to privately learn a   matrix $A$ that is spectrally close to $\Sigma^{-1/2}$. This effectively sets the samples to have near identity covariance. We show it suffices for our purpose.

\begin{figure}[htbp]
    \centering\begin{mdframed}[style=algo]
\begin{enumerate}
    \item \textbf{Input:} Samples $\calX = \{X_1, X_2 , \ldots, X_{3n}\} \subset \R^d$,   privacy parameter $\epsilon > 0$.
    \item For each $i=1,2,\ldots, n$, let $Y_i = \frac{1}{\sqrt{2}} \left(X_{2i} - X_{2i-1}\right)$.
    \item Compute  the preconditioning matrix $A \in \R^{d \times d}$ using $\{Y_1, \ldots, Y_{n}\}$ as input to Algorithm~\ref{alg:precondition} and privacy parameter $\epsilon/2$.
    \item Run the algorithm  \textsf{PureDPMean} in Theorem \ref{thm:pure-dp-mean-est} \cite{hopkins2021efficient} on samples $\{A X_{2n + 1}/\sqrt{20}, \ldots, A X_{3n}/\sqrt{20}\}$ with privacy parameter $\epsilon/2$, error parameter $\alpha/\sqrt{20}$, and failure rate $\beta=0.01$ to obtain a mean  estimate $\tilde{\mu} \in \R^{ d}$.
    \item \textbf{Output:} Mean estimate $\widehat{\mu} = \sqrt{20}A^{-1}\tilde{\mu}$.
\end{enumerate}
\end{mdframed}
\caption{Private mean estimation algorithm}
\label{alg:pure-dp-mean-est}
\end{figure}

Specifically, our private mean estimation procedure is described by Algorithm \ref{alg:pure-dp-mean-est} and its guarantees given below.
\begin{lemma}[Pure DP Gaussian mean estimation]\label{lem:pure-dp-mean-gauss-ours}
Let $\alpha > 0$ be an error parameter, $\epsilon> 0$ be a privacy parameter, $R \in \R$ and $d \in \N$. There is a polynomial-time $\epsilon$-DP algorithm that, given $\eps$ and  
\[n \ge \widetilde{O}\left(\frac{d^2\log(\kappa)}{\epsilon} + \frac{d}{\alpha^2 \epsilon} + \frac{d\log(R)}{\epsilon}\right)\] 
i.i.d.\ samples $\calX = \{X_1, \ldots, X_n\}$ from $\calN(\mu, \Sigma)$ for an unknown $\mu$ satisfying $\| \mu\|_2 \leq R$ and an unknown $\Sigma \in \R^{d \times d}$ satisfying $I \preceq \Sigma \preceq \kappa I$, outputs $\widehat{\mu}$ satisfying \[\left\|\Sigma^{-1 / 2}(\mu-\widehat{\mu})\right\|_2 \leq \alpha\] with probability at least $0.9$.
\end{lemma}

\begin{proof}
    The privacy follows from basic   composition of the   privacy property of Algorithm \ref{alg:precondition} and \textsf{PureDPMean}. We focus on the utility analysis proving that $\left\|\Sigma^{-1 / 2}(\mu-\widehat{\mu})\right\|_2 \leq \alpha$.

    Since input samples $\{X_1,X_2,\ldots,X_{2n}\}$ are \iid from $\cN(\mu, \Sigma)$, the random vectors $\{Y_1,Y_2,\ldots, Y_n\}$ are \iid according to $\cN(0, \Sigma)$. By Theorem \ref{thm:pure-dp-preconditioner}, the choice of $n$ and our assumption on $\Sigma$,   step 2 of Algorithm \ref{alg:pure-dp-mean-est} outputs a preconditioning matrix $A \in \R^{d \times d}$ such that
    \begin{equation}\label{eqn:precond-mean}
            I\preceq A \Sigma A \preceq 20 I
    \end{equation}
    with probability at least $0.99$. Since $I\preceq \Sigma$ and $A\Sigma A \preceq 20I$, we have $\|A\|_2 \leq \sqrt{20}$. Since   $\{  X_{2n + 1}, \ldots,  X_{3n}\}$ are \iid from $\cN(\mu,\Sigma$), then   in step 3,   $\{A X_{2n + 1} / \sqrt{20}, \ldots, A X_{3n} / \sqrt{20}\}$ are i.i.d.\ according to $\cN(A\mu/\sqrt{20}, A\Sigma A/20)$. Recall that  the guarantee from $A$ ensures   $A \Sigma A / 20 \preceq I$.  Moreover, $\|A \mu / \sqrt{20}\|_2 \leq\|A\|_2\|\mu\|_2 / \sqrt{20} \leq R$, since $\|A\|_2\leq \sqrt{20}$. Therefore, conditioned on these events, Theorem \ref{thm:pure-dp-mean-est} implies that the mean estimate $\tilde{\mu}$ in step 3 satisfies that $\|A\mu/\sqrt{20} -\tilde{\mu}\|_2 \leq \alpha/\sqrt{20}$, and hence $\|A(\mu - \widehat{\mu})\|_2 \le \alpha$, with probability at least $0.99$. Now since $I\preceq A \Sigma A$, we have $\|\Sigma^{-1/2}A^{-1}\|_2 \le 1$. Hence, $\left\|\Sigma^{-1/2}(\mu -\widehat{\mu})\right\|_2 \leq \|\Sigma^{-1/2} A^{-1}\|_2 \cdot \|A(\mu - \widehat{\mu})\|_2 \le \alpha$,  with probability at least $0.99$.
    The proof follows by applying a union bound over the failure probability of step 2 and 3 of   Algorithm \ref{alg:pure-dp-mean-est}.

\end{proof}

\paragraph{Putting it Together}
We now put together Lemma \ref{lem:pure-dp-mean-gauss-ours} on mean estimation and Theorem \ref{thm:pure-dp-cov-est-restate} on covariance estimation to show:

\begin{theorem}[Pure DP Gaussian estimation]
    Let $\alpha > 0$ be an error parameter, $\epsilon> 0$ be a privacy parameter, $R \in \R$ and $d \in \N$. There is a polynomial-time $\epsilon$-DP algorithm that, given $\eps$ and
\[n \geq \widetilde{O}\left(\frac{d^2\log(\kappa)}{\epsilon} + \frac{d^2}{\alpha^2\epsilon} + \frac{d\log(R)}{\epsilon}\right)\] 
i.i.d.\ samples $\calX = \{X_1, \ldots, X_n\}$ from $\calN(\mu, \Sigma)$ for an unknown $\mu$ satisfying $\| \mu\|_2 \leq R$ and an unknown $\Sigma \in \R^{d \times d}$ satisfying $I \preceq \Sigma \preceq \kappa I$, outputs $\widehat{\mu}, \widehat{\Sigma}$  such that 
\begin{equation}\label{eqn:tv-36}
\dtv\left(\cN(\mu,\Sigma), \cN\left(\widehat{\mu}, \widehat{\Sigma}\right)\right)  \leq O(\alpha) 
\end{equation}
with probability at least $0.8$.
\begin{proof}
For simplicity, assume that $n$ is even.
We   use $n/2$ samples as input to the private mean estimation algorithm (Algorithm \ref{alg:pure-dp-mean-est}) and another $n/2$ samples for   covariance estimation (Algorithm \ref{alg:pure-dp-cov-est}), both with a privacy parameter $\eps/2$. 
Privacy follows from basic composition (Theorem \ref{thm:basic-comp}). 
For utility, by our choice of $n$, Lemma \ref{lem:pure-dp-mean-gauss-ours} implies  that \[\left\|\Sigma^{-1 / 2}(\mu-\widehat{\mu})\right\|_2 \leq \alpha\]with probability  $0.9$, and Theorem \ref{thm:pure-dp-cov-est-restate} implies that \[\norm{\Sigma^{-1/2} \widehat{\Sigma} \Sigma^{-1/2} - I}_F \leq \alpha\] with probability $0.9$.  
Conditioned on the success of both steps, Lemma \ref{lem:parameter-closeness} yields the desired total variation  guarantee, Equation \ref{eqn:tv-36}. The failure probability follows from an application of union bound.
 \end{proof}
\end{theorem}

\subsection{Robustness}\label{sec:pure-dp-robust}
We now argue that our algorithms are robust to adversarial corruptions, with the cost that estimation error generally   is worsened to $\alpha +O(\sqrt{\eta})$, where $\eta$ is the fraction of corrupted samples. 
As we discussed, our algorithm for learning Gaussian   under pure DP is by reducing the problem to black-box applications of the main procedure  from \cite{hopkins2021efficient}, namely, \textsf{PureDPMean} in Theorem \ref{thm:pure-dp-mean-est}. 
We exploit the robustness property of \textsf{PureDPMean} to show:
\begin{theorem}[Robust pure DP Gaussian estimation]\label{thm:pure-dp-robust-gaussian}
Let $\alpha > 0$ be an error parameter, $\epsilon> 0$ be a privacy parameter, $R \in \R$ and $d \in \N$. For a sufficiently small constant $\eta$, there is a polynomial-time $\epsilon$-DP algorithm that, given $\eps$ and
\[n \ge \widetilde{O}\left(\frac{d^2\log(\kappa)}{\epsilon} + \frac{d^2}{\alpha^2 \epsilon} + \frac{d\log(R)}{\epsilon}\right) \] 
$\eta$-corrupted samples $\calX = \{X_1, \ldots, X_n\}$ from $\calN(\mu, \Sigma)$ for an unknown $\mu$ satisfying $\| \mu\|_2 \leq R$ and an unknown $\Sigma \in \R^{d \times d}$ satisfying $I \preceq \Sigma \preceq \kappa I$, outputs $\widehat{\mu}$ and $ \widehat{\Sigma}$  such that 
\begin{equation}\label{eqn:tv-37}
\dtv\left(\cN(\mu,\Sigma), \cN\left(\widehat{\mu}, \widehat{\Sigma}\right)\right)  \leq O(\alpha +\sqrt{\eta}) 
\end{equation}
with probability at least $0.8$.
\end{theorem}

To give a proof sketch, the key step is to observe that our recursive preconditioning algorithm  (Algorithm \ref{alg:precondition}) is robust. Recall that the algorithm simply calls our weak preconditioning scheme (Algorithm \ref{alg:precondition-weak}) recursively (for $\log \kappa$ times). This weak scheme, in turn, runs \textsf{PureDPMean} to roughly estimate   $\Sigma$, up to an additive error of $0.01\kappa$ in the (absolute) Frobenius norm. We observe that the robustness property of  \textsf{PureDPMean} suffices to yield the same error guarantee even under $\eta$-corruption. Hence, the recursive preconditioning algorithm retains its performance under corruption. Finally,  the remaining steps of our algorithms for mean and covariance estimation simply calls \textsf{PureDPMean}  on the preconditioned samples. We lose the extra factor of $\sqrt{\eta}$ from there.

We remark that in the analysis we make no attempt to optimize the breakdown point of the algorithm. (In fact, it   depends on the hidden constant in the $O(\sqrt{\eta})$ error term of \textsf{PureDPMean}.)
\begin{proof}[Proof of Theorem \ref{thm:pure-dp-robust-gaussian}]
We start by modifying step 2 of Algorithm \ref{alg:precondition-weak} to invoke Algorithm \ref{alg:pure-dp-matrix-mean} with $\alpha = 0.0099$ instead. By the choice of $n,\alpha$ and for a sufficiently small $\eta$, Theorem \ref{thm:pure-dp-mean-est} implies that the rough estimate  $\widehat{\Sigma}$, computed in the second step of Algorithm~\ref{alg:precondition-weak}, satisfies
\begin{equation}\label{eqn:error-puredp-2}
    \norm{\widehat{\Sigma} - \Sigma}_F \leq 0.01 \kappa
\end{equation} 
 with probability $1- \beta$. 
 Observe that with the error bound above, the rest of the proof of Theorem \ref{thm:pure-dp-preconditioner-one-round}  remains valid. 
 Inspecting the analysis of the recursively preconditioning algorithm, we note that the error guarantee of Theorem \ref{thm:pure-dp-preconditioner-one-round}  suffices to imply Theorem \ref{thm:pure-dp-preconditioner}. 

 We now argue for the error rate on covariance  and mean estimation, separately.
 \begin{itemize}
     \item For covariance estimation, consider Algorithm \ref{alg:pure-dp-cov-est} and its guarantees Theorem \ref{thm:pure-dp-cov-est}.  Step 2, the preconditioning step, of Algorithm \ref{alg:pure-dp-cov-est} retains its performance exactly. In step 3, we instead get an estimate $\widehat{\Sigma}$ such that $\norm{\Sigma^{-1/2} \widehat{\Sigma} \Sigma^{-1/2} - I}_F \leq \alpha+O(\sqrt{\eta})$, due to the robustness property of \textsf{PureDPMean}.
    \item For mean estimation, consider Algorithm \ref{alg:pure-dp-mean-est} and its guarantees Lemma \ref{lem:pure-dp-mean-gauss-ours}. Similarly, step 3 of the algorithm retains its performance exactly, and we lose an extra $O(\sqrt{\eta})$ factor in step 4. Hence, the algorithm outputs an $\widehat \mu$ such that $\|\mu -\widehat{\mu}\| \leq \alpha +O(\sqrt{\eta})$.
 \end{itemize}
 Applying the Lemma \ref{lem:parameter-closeness} converts the parameter closeness to distribution closeness, and this finishes the proof.
 \end{proof}
\section{Gaussian Sampling Mechanism}
\label{sec:gauss-sampling}
 In this section, we provide an approximate DP mechanism for releasing a   covariance estimate. The mechanism is simple and natural. It  works by approximating a PSD matrix  $\Sigma$ empirically with independent samples from $\mathcal{N}(0,\Sigma)$. 
 \paragraph{Setting}
 Consider a setting where there is an estimation   algorithm $\mathcal{M}: \left(\R^{d}\right)^n \rightarrow \mathbb{S}_+^d$   that given  $n$ samples in $\mathbb{R}^d$ outputs a PSD matrix. For us, the algorithm  $\mathcal{M}$, roughly speaking, will be the solution of a \textit{stabilized} convex program (minimizing a strongly convex potential).
 In particular, we can guarantee a sensitivity bound in the relative Frobenius distance such that \[\left\|\mathcal M({\cal Y})^{1/2}\mathcal M({\cal Y}')^{-1}\mathcal M({\cal Y})^{1/2} - I\right\|_F \le \Delta\] for some small $\Delta$, on any neighboring datasets  $\mathcal{Y}, \mathcal{Y}'$. 
Moreover, we can show that   $\mathcal M(\mathcal{Y})$ is close to the true covariance $\Sigma$, if $\mathcal{Y}$ are samples from $\mathcal{N}(0,\Sigma)$. Yet, $\mathcal{M}$ by itself provides no privacy. The goal, therefore, is to privately release the estimate  $\mathcal M(\mathcal{Y})$, while retaining its statistical performance. 

In light of standard DP mechanisms, it is a natural idea to   design an explicit noise-injection mechanism to privatize $\mathcal{M}(\mathcal{Y})$. This is indeed the approach of \cite{kothari2021private}. To ensure that the noise level does not hamper accuracy, though, their mechanism depends on the estimate itself, which leads to much technical complications.  We now give a significantly simpler solution to this problem, which we believe may find applications elsewhere.

\subsection{Algorithm}

The algorithm we propose is extremely simple. Let $\Sigma = \mathcal M({\mathcal{Y}})$ be the covariance we intend to release.  We approximate by empirical samples. That is, given $\Sigma$ and an accuracy parameter $k$,  we sample $g_i \sim \mathcal{N}(0,\Sigma)$ for each $i\in[k]$ and release $\widehat{\Sigma} = \frac{1}{k} \sum_{i=1}^k g_i g_i^T$. The algorithm is formally given in Algorithm \ref{alg:cov-noise-mechanism}.

\begin{figure}[htbp]    
    \centering\begin{mdframed}[style=algo]
\begin{enumerate}
    \item \textbf{Input:} A PSD matrix $\Sigma \in \R^{d \times d}$ and parameter $k \in \N$
    \item Obtain vectors $g_1,g_2, \ldots, g_k$ by sampling $g_i \sim \calN(0,\Sigma)$, independently for each $i\in [k]$.
    \item \textbf{Output:} Covariance estimate $\widehat{\Sigma} = \frac{1}{k} \sum_{i=1}^k g_i g_i^T$.
\end{enumerate}
\end{mdframed}
\caption{The Gaussian Sampling Mechanism}
\label{alg:cov-noise-mechanism}
\end{figure}

\subsection{Analysis}
We now give a privacy and utility analysis of the algorithm. 
Intuitively, when $k\rightarrow \infty$, the output $\widehat{\Sigma} \rightarrow \Sigma$, which means good utility. However, this leads to no privacy, since $\Sigma = \mathcal{M}(\mathcal{Y})$   is not private (w.r.t.\ $\mathcal{Y}$). Our theorem characterizes the utility-privacy trade-off quantitatively (when $k$ is finite). 
Note that for fixed $\eps,\delta$, the smaller the sensitivity bound  $\Delta$ is, the higher we get to choose $k$, which leads to better approximation accuracy.

\begin{theorem}[Analysis of the Gaussian Sampling Mechanism]\label{thm:gaussian-sampling-mech}
Fix $\varepsilon, \delta \in (0, 1)$ and $k \in \mathbb{N}$ and let \[\Delta = \min\left(\frac{\varepsilon}{\sqrt{8k\log(1/\delta)}}, \frac{\varepsilon}{8\log(1/\delta)}\right) <1.\] 
Let $\mathcal{M}: \left(\R^{d}\right)^n \rightarrow \mathbb{S}_+^d$ be a (randomized) algorithm that given a dataset of $n$ points in $\mathbb{R}^d$ outputs a PSD matrix. 
Suppose that $\mathcal{M}$ satisfies a sensitivity bound  that \[\left\|\mathcal M({\cal Y})^{1/2}\mathcal M({\cal Y}')^{-1}\mathcal M({\cal Y})^{1/2} - I\right\|_F \le \Delta\] for any neighboring datasets  $\mathcal{Y}, \mathcal{Y}'$. Then given an input $\Sigma = \mathcal{M}(\mathcal{Y})$, 
\begin{itemize}
    \item Algorithm \ref{alg:cov-noise-mechanism} is $(\varepsilon, \delta)$-DP (with respect to the original dataset $\mathcal{Y}$); and
    \item  Algorithm \ref{alg:cov-noise-mechanism} outputs $\widehat{\Sigma} \in \mathbb{S}_+^d$ such that  with probability at least $1-\gamma$,\[\norm{\Sigma
    ^{-1/2} \widehat{\Sigma} \Sigma^{-1/2} - I}_{F} \leq \rho\]
for $\rho = O \left(\sqrt{\frac{d^2 + \log(1/\gamma)}{k}} + \frac{d^2 + \log(1/\gamma)}{k} \right)$.
\end{itemize}
\end{theorem}

\begin{proof}
The utility guarantee is immediately implied by Theorem \ref{thm:emp-cov-concentration}. For the proof of privacy, let \[f_{\Sigma}(x) = (2\pi)^{-\frac{d}{2}} \det(\Sigma)^{-\frac{1}{2}} \exp\left(-\frac{1}{2}x^\top \Sigma^{-1} x\right)\] denote the ${\cal N}(0, \Sigma)$ density function, let $\Sigma_1 = \Sigma({\cal Y})$ and $\Sigma_2 = \Sigma({\cal Y'})$ for neighboring datasets ${\cal Y}, {\cal Y'}$, and let $g_1, \ldots, g_k$ be the i.i.d.\ samples from ${\cal N}(0, \Sigma_1)$ output by Algorithm \ref{alg:cov-noise-mechanism}. 
By Theorem \ref{thm:privacy-loss-tail-implies-dp}, it suffices to show that \(\Pr\left[Z > \varepsilon\right] \le \delta,\) where \[Z = \sum_{i=1}^k \log\left(\frac{f_{\Sigma_1}(g_i)}{f_{\Sigma_2}(g_i)}\right)\] is the privacy loss random variable. 
To this end, define $A = \Sigma_1^{1/2}\Sigma_2^{-1}\Sigma_1^{1/2}$ and let $A = \sum_{j=1}^d \lambda_j v_j v_j^\top$ be its spectral decomposition. 
By assumption, we have \[\sqrt{\sum_{j=1}^d (\lambda_j-1)^2} = \|A - I\|_F \le \Delta.\] Similarly, letting $B = \Sigma_2^{1/2}\Sigma_1^{-1}\Sigma_2^{-1/2}$ and switching the roles of ${\cal Y}$ and ${\cal Y}'$, we see that \[\sqrt{\sum_{j=1}^d (1 - 1/\lambda_j)^2} = \|I - A^{-1}\|_F = \|I-B\|_F\le \Delta.\] In the above calculation, we have used the fact that $A^{-1} = CC^\top$ and $B = C^\top C$ for the matrix $C = \Sigma_1^{-1/2}\Sigma_2^{1/2}$, and hence $A^{-1}$ and $B$ have the same eigenvalues.
Next, let $$h_{ij} = \left\langle \Sigma_1^{-1/2} g_i, v_j \right\rangle,$$ and note that the random variables $h_{ij}$ are i.i.d.\ copies of ${\cal N}(0, 1)$. 
Using this notation, the privacy loss random variable $Z$ can be expressed as follows:
\begin{align*}
    Z &= \sum_{i=1}^k \log\left(\frac{\det(\Sigma_1)^{-\frac{1}{2}}\exp\left(-\frac{1}{2}g_i^\top\Sigma_1^{-1}g_i\right)}{\det(\Sigma_2)^{-\frac{1}{2}}\exp\left(-\frac{1}{2}g_i^\top\Sigma_2^{-1}g_i\right)}\right)\\
    &= \sum_{i=1}^k \left(\frac{1}{2}g_i^\top\left(\Sigma_2^{-1}-\Sigma_1^{-1}\right)g_i - \frac{1}{2} \log\left(\frac{\det(\Sigma_1)}{\det(\Sigma_2)}\right)\right)\\
    &= \frac{1}{2}\sum_{i=1}^k \left(\left(\Sigma_1^{-1/2}g_i\right)^\top(A - I)\left(\Sigma_1^{-1/2}g_i\right) - \log\det(A)\right)\\
    &= \frac{1}{2}\sum_{i=1}^k \sum_{j=1}^d \Big((\lambda_j - 1) h_{ij}^2 - \log(\lambda_j)\Big),
\end{align*}
where the last step is by taking the spectral decomposition of $A$.
Since $h_{ij}^2 \sim \chi^2_1$, Theorem \ref{thm:chi-squared-sub-exponential} (sub-exponential parameters of $\chi^2_1$) and Theorem \ref{thm:sub-exponential-sum} (sub-exponential parameters of an independent sum) imply that $Z$ is a sub-exponential random variable with parameters $\nu = \sqrt{k}\|A - I\|_F \le \sqrt{k}\Delta$ and $\alpha = 2\|A - I\|_2 \le 2\Delta$. Also, the mean of $Z$ is
\begin{align*}
    \E[Z] &= \frac{k}{2}\sum_{j=1}^d\left(\lambda_j - 1 - \log\left(\lambda_j\right)\right) &\text{(by linearity of expectation)}\\
    &\le \frac{k}{2}\sum_{j=1}^d \left(\lambda_j - 2 + 1/\lambda_j\right) &\text{(since $\lambda_j>0$ and}\log(x) \ge 1 - 1/x \text{ for $x> 0$})\\
    &= \frac{k}{2}\sum_{j=1}^d (\lambda_j - 1)(1 - 1/\lambda_j)\\
    &\le \frac{k}{2}\|A - I\|_F\left\|I - A^{-1}\right\|_F &\text{(by Cauchy-Schwarz)}\\
    &\le \frac{k}{2}\Delta^2 &\text{(by the above discussion)} \\
    &\le \frac{\varepsilon}{2} &\text{(because $\Delta < \varepsilon/\sqrt{k}$ and $\varepsilon < 1$).}
\end{align*}
Finally, using Theorem \ref{thm:sub-exponential-tail} (sub-exponential tail bound), we conclude that \[\Pr[Z > \varepsilon] \le \Pr\left[Z - \E[Z] > \frac{\varepsilon}{2}\right] \le \max\left\{e^{-\frac{(\varepsilon/2)^2}{2\nu^2}}, e^{-\frac{\varepsilon/2}{2\alpha}}\right\} \le \delta\] by plugging in the values of $\nu$, $\alpha$, and $\Delta$. This proves that Algorithm \ref{alg:cov-noise-mechanism} is $(\varepsilon, \delta)$-DP.
\end{proof}

\section{Approximate DP Robust Estimation}
\label{sec:approx-dp}
In this section, we give an efficient algorithm for private and robust estimation of mean and covariance of a high-dimensional Gaussian. The main idea is to improve the sensitivity bound of \cite{kothari2021private} via minimizing an entropy 
potential instead of the 2-norm objective. 
Moreover, we   use our Gaussian sampling mechanism to replace the noise-addition mechanism of \cite{kothari2021private} to improve the sample cost.

Throughout we let
 $\Ent (x) =  \sum_{i=1}^n x_i \log(1/x_i) + x_i$ 
be the (unnormalized) negative entropy function, $\mathcal{Y}= \{Y_i\}_{i=1}^n$ be the input data and $\eta$ be the (fixed) corruption rate.

\subsection{Stable Outlier Rate Selection}
\label{sec:outlier-rate}
The main subroutine of our algorithm is to search for pseudo-distributions that, in addition to satisfying    certain witness checking constraints $\calA_{C, \eta, n} (\mathcal{Y})$, also minimize a   strongly convex potential.  

\begin{definition}[Potential]
\label{def:pot}
Let $C > 0$, $\eta \in [0,1]$, $t \in \N$, $\mathcal{Y} \subset \R^d$ be a dataset of size $n$ and $f: \R^n \rightarrow \R$ be convex. For any degree-$2t$ pseudo-distribution $\zeta$ consistent with $\calA_{C, \eta, n} (\mathcal{Y})$, let $\Pot_{C, \eta, n, \zeta}(\mathcal{Y}) = f \left(\pE_\zeta [w]\right) $. Let $\Pot_{C, \eta, n}(\mathcal{Y}) = \min_{\zeta} \Pot_{C, \eta, n, \zeta}(\mathcal{Y})$, where the minimum is over all degree-$2t$ pseudo-distributions $\zeta$ satisfying $\calA_{C, \eta, n} (\mathcal{Y})$. If no such $\zeta$ exists, we define $\Pot_{C, \eta, n}(\mathcal{Y}) = \infty$.
\end{definition}

We choose $f(x) = -\Ent(x)/\log(n)$.
Moreover, we will define different constraint systems $\calA^{\textsf{mean}}_{C, \eta, n} (\mathcal{Y})$ and $\calA^{\textsf{cov}}_{C, \eta, n} (\mathcal{Y})$  for mean and covariance estimation, respectively. 
Before  that, the first step of our estimator  is   to randomly select an outlier rate $\eta'$  
such that    the potential function on the pseudo-distribution 
we compute  is close  on two neighboring input datasets. 

\paragraph{DP Selection}
Our algorithm of approximate DP estimation uses the private outlier rate selection procedure of \cite{kothari2021private}. 
For that, we need to apply the exponential mechanism to select an
outlier rate that satisfies a stability threshold.
With the vanilla exponential mechanism, however, one is not
guaranteed to get candidates with score above a certain
threshold. 
Theorem \ref{thm:dp-apx-selection} provides an
$(\eps, \delta)$-DP selection mechanism, where the candidate selected
has utility score that is above a certain threshold.
The caveat is that, with probability at most $\beta$, 
the procedure might output $\perp$.
 
\begin{lemma}[Theorem 3.34 of \cite{kothari2021private}] \label{thm:dp-apx-selection}

Fix $\beta\in(0, 1]$. 
Let $\eps, \delta\in (0, 1]$ be privacy parameters.
$\Delta$ is the sensitivity bound for the scoring
function.
Let
$\cC$ be a set of candidates and let $\score: \cC \times \cY\rightarrow\reals$ be a scoring function for candidates as a function of the databases $Y \in \cY$, such that its sensitivity (w.r.t. $Y$) is at most $\Delta$. 

There exists an algorithm \SelectionAlg\ with input parameter $\kappa$ that satisfies the following properties:
\begin{enumerate}
\item\label{cond:selectdp} \SelectionAlg~is $(\eps, \delta)$-DP.
\item\label{cond:selectgeqkappa} If the output of \SelectionAlg~is $c^* \neq \perp$, then $\score(c^*,Y) \geq \kappa$.
\item\label{cond:selectrejectprob} If there exists $c \in \cC$ such that $\score(c, Y) \geq \kappa + O\left(\frac{\Delta}{\eps} \cdot \log\left(\frac{|\cC|}{\beta \delta}\right)\right)$, then \SelectionAlg~output $\perp$ with probability at most $\beta$.
\end{enumerate}
\end{lemma}
While the exact implementation of \SelectionAlg\ is unimportant for us, it's simply a variant of the exponential mechanism, and we  describe it in detail in Section \ref{sec:dp-selection} for the sake of completeness.

\paragraph{Outlier Rate Selection Mechanism}
Now we describe the private outlier rate selection mechanism  $\mathcal{M}_1$  in detail. The procedure is via running the \SelectionAlg\ mechanism with a   score function tied to the stability of the convex program.

In the following, the constraint system  $\mathcal A(\eta')$   generally refers to  $\calA^{\textsf{mean}}_{C, \eta', n} (\mathcal{Y})$ or $\calA^{\textsf{cov}}_{C, \eta', n} (\mathcal{Y})$, for fixed $C,n$ and  input data $\mathcal{Y}$.
\begin{definition}[Stability] \label{def:stability}
Fix $L \in \N$. Let $\tau, \gamma \in \{0, \dots, n\}$ such that $\gamma \leq \tau, n - \tau$. Suppose for some $Y \subseteq \R^d$ of size $n$, the constraint system $\mathcal A((\tau-\gamma)/n)$ is feasible. We define the stability of the $2\gamma$ length interval centered at $\tau$ to be 
\begin{align*}
\stab_Y(\tau,\gamma) = \Pot_{(\tau-\gamma)/n}(Y)- \Pot_{(\tau+\gamma)/n}(Y)
\end{align*}
\end{definition}
\begin{definition}[Score function] \label{def:em-score}
Fix $n,k\in\N$ and $C > 0$. Let $Y \subseteq \R^d$ be a set of size $n$. For a parameter $L$, we define the following score function for every integer $\tau\in [n]$:
\begin{align*}
\score_{n,C,k}(\tau, Y) = \begin{cases}
0 & \text{ if } \mathcal A_{C,k,\tau/n,n}(Y)\text{ is infeasible}, \\
\max_{\gamma \atop \mathcal A_{C,k,(\tau-\gamma)/n,n}(Y) \text{ is feasible}} \min\{\gamma, 20L - n \cdot \stab_Y(\tau, \gamma)\} &\text{ otherwise}.
\end{cases}
\end{align*}
In the second case, we define $\gamma^*_Y(\tau) := \arg\max_{\gamma \atop \mathcal A_{C,k,(\tau-\gamma)/n,n}(Y) \text{ is feasible}} \min\{\gamma, L - \stab_Y(\tau, \gamma)\}$.
\end{definition}

\begin{figure}[htbp]    
    \centering\begin{mdframed}[style=algo]
\begin{enumerate}
    \item \textbf{Input:} The true outlier rate $\eta$, privacy loss $\eps,\delta$, threshold $L$
    \item Run the $(\eps,\delta)$-DP \SelectionAlg\ mechanism with $\kappa = L$ and score function defined in Definition \ref{def:em-score} to sample $\tau \in [0, \eta n]$ or get $\perp$
    \item \textbf{Output:} An integer $\tau \in [0, \eta n]$ or $\perp$
\end{enumerate}
\end{mdframed}
\caption{The private outlier rate selection mechanism $\mathcal{ M }_1$}
\label{alg:outlier-rate}
\end{figure}

Given a stability threshold
(intuitively indicates how unstable the convex 
program selecting the witness is), 
privacy parameters $\eps, \delta\in(0, 1)$,
and outlier rate $\eta\in[0, 1/2)$,
Algorithm~\ref{alg:outlier-rate} returns $\tau$, a value
that defines a ``stable'' interval of outlier rates:
$[\tau/n - 0.5L/n, \tau/n - 0.5L/n]$.

\paragraph{Analysis}
We analyze the outlier rate selection mechanism $\mathcal{ M }_1$ as follows. 
Our choice of entropy objective makes the analysis significantly differ   from \cite{kothari2021private}.

Moreover, for our privacy analysis, we rely on the following
simple observation about adjacent pseudo-distributions that are
induced by adjacent datasets:

\begin{lemma}[Adjacent Pseudo-distributions, Lemma 4.9 of \cite{kothari2021private}] \label{lem:adjacent-pseudo-distribution}
Let $\zeta$ be a pseudo-distribution of degree $2k$ that satisfies all the constraints in $\calA_{C,\eta,n}(\calY)$ on dataset $\calY$ with outlier rate $\eta$. Let $\calY' \subseteq \R^d$ be adjacent to $Y$. Define an \emph{adjacent pseudo-distribution} $\zeta'$ (that ``zeroes out $w_i$'') by $\pE_{\zeta'}[ w_S p(X',\cdots)] = \pE_{\zeta}[w_S p(X',\cdots)]$ if $i \not \in S$ and $\pE_{\zeta'}[ w_S p(X',\cdots)] = 0$ if $i \in S$ for every polynomial $p$ in $X'$ and other auxiliary indeterminates in $\calA_{C,\eta,n}(\calY)$. Then, $\zeta'$ is a pseudo-distribution of degree $2k$ that satisfies all the constraints in $\calA_{C,\eta,n}(\calY)$ on both inputs $\calY'$ and $\calY$ with outlier parameter $\eta+1/n$. 
\end{lemma}

To establish the sensitivity of the outlier rate selection
procedure, we require certain properties of the
$\Ent$ objective:

\begin{lemma}[Properties of $\Ent$]
\label{lem:entropy-properties}
The function $f(x) = -\Ent(x)/\log(n)$ satisfies the following properties:
\begin{enumerate}[(i)]
    \item For any $x \in [0,1/n]^n$ which satisfies $1 \geq \norm{x}_1 \geq (1- \eta)$, and $j \in [n]$, define $\Tilde{x} \in \R^n$ by $\Tilde{x}_j = 0$ and $\Tilde{x}_i = x_i$ for $i \neq j$. Then it holds that $f(\Tilde{x}) \leq f(x) + \frac{1}{n} + \frac{1}{n \log (n)}$.
    \item For any $x \in [0,1]^n$ which satisfies $1 \geq \norm{x}_1 \geq (1- \eta)$, it holds that
    \[
    - 1 \leq f(x) \leq -(1-\eta) - \frac{1-\eta}{\log(n)}.
    \]
\end{enumerate}
\end{lemma}
\begin{proof}
\begin{enumerate}[(i)]
    \item Observe that $\log(n)(f(\Tilde{x}) - f(x)) = \Ent(x) - \Ent(\Tilde{x}) = x_j \log(1/x_j) + x_j \leq \frac{\log(n) + 1}{n}$. Rearranging completes the proof.
    \item Define the function $h(x) = \sum_{i=1}^n \frac{x_i}{\norm{x}_1} \log \left(\frac{\norm{x}_1}{x_i} \right)$ and observe that $\norm{x}_1 h(x) = \norm{x}_1 \log(\norm{x}_1) + \Ent(x) - \norm{x}_1$. Using the fact that $ h(x) \leq \log(n)$ for all $x \in [0,1]^n$ and rearranging completes the proof of the lower bound on $f(x)$.
    
    For the upper bound, begin by noting that $h(x) = h(x/\norm{x}_1)$ and that $x/\norm{x}_1 \in [0,1/ (1-\eta)n]^n$. It is a standard fact that the min-entropy of a distribution lower bounds the Shannon entropy. Hence, $h(x/\norm{x}_1) \geq \log((1-\eta)n)$. Using the fact that $\norm{x}_1 \geq (1-\eta)$ and rearranging completes the proof of the upper bound on $f(x)$.
\end{enumerate}
\end{proof}

We now show that for any stable interval defined around 
$\tau\in[0, \eta n]$, the interval cannot have a significantly
different value than another smaller-length interval defined around
$\tau$:

\begin{lemma}[see Lemma 4.13 of \cite{kothari2021private}] \label{lem:comparison-of-stability}
Let $\tau, \gamma \in [n]$ such that $\gamma \leq \tau, n - \tau$. Suppose for some $\calY \subseteq \R^d$ of size $n$, the constraint system $\calA_{C,(\tau-\gamma)/n, n}(\calY)$ is feasible for both $\calY. \calY'$, where $\calY'$ is a neighboring dataset. Then, for any $\tau, \gamma$, 
\[
\stab_{\calY'}(\tau,\gamma-1) \leq \stab_{\calY}(\tau,\gamma) + \frac{2}{n} + \frac{2}{n \log (n)}
\]
\end{lemma}

\begin{proof}
Using Lemmas~\ref{lem:adjacent-pseudo-distribution}, \ref{lem:entropy-properties} and noting that if $\zeta'$ is adjacent to $\zeta$ then 
\[
f\left(\pE_{\zeta'}[w]\right) \leq f\left(\pE_{\zeta}[w]\right) + \frac{1}{n} + \frac{1}{n \log (n)},
\]
we have:
\[
 \Pot_{(\tau-\gamma+1)/n}(\calY') \leq \Pot_{(\tau-\gamma)/n}(\calY) + \frac{1}{n} + \frac{1}{n \log (n)},
\]
and
\[
 \Pot_{(\tau+\gamma)/n}(\calY) \leq \Pot_{(\tau+\gamma-1)/n}(\calY') + \frac{1}{n} + \frac{1}{n \log (n)}.
\]
Combining the two equations yields
\begin{align*}
 \stab_{\calY'}(\tau,\gamma-1) &= \Pot_{(\tau-\gamma+1)/n}(\calY')-\Pot_{(\tau+\gamma-1)/n}(\calY')\\
 &\leq \Pot_{(\tau-\gamma)/n}(\calY)-\Pot_{(\tau+\gamma)/n}(\calY) + \frac{2}{n} + \frac{2}{n \log (n)}\\
& = \stab_{\calY}(\tau,\gamma) + \frac{2}{n} + \frac{2}{n \log (n)}. \qedhere
\end{align*}
\end{proof}

To use the exponential mechanism
(with any reasonable amount of utility), we need a small
sensitivity on the score function. We show that the $\score$ 
function, as defined, has a small sensitivity:

\begin{lemma}[Sensitivity of Score Function, Lemma 4.14 of \cite{kothari2021private}] \label{lem:sensitivity-score-function}
Let $\calY, \calY'$ be set of $n$ points in $\R^d$ differing at most in one point, and $\tau \in [n]$. Then, for every $\tau >0$,
\begin{equation}
 |\score(\tau, \calY)-\score(\tau, \calY')| \leq 6. \label{eq:scoresensitivity}
\end{equation}
\end{lemma}
\begin{proof}
The proof of this lemma is identical to that of Lemma 4.14 of \cite{kothari2021private}, but invokes Lemma \ref{lem:comparison-of-stability} instead of Lemma 4.13 in \cite{kothari2021private}.
\end{proof}

The following lemma shows that not too many intervals can
be unstable. This implies that there must exist at least one
stable interval:

\begin{lemma}[Existence of a Good Stable Interval, Lemma 4.15 of \cite{kothari2021private}] \label{lem:existence-of-good-interval}
Suppose $\calA_{C, \eta/2, n} (\calY)$ is feasible . For every $L \in [0, 0.25 \eta n]$, there is a $\tau \in [0,\eta n]$ such that $\score(\tau, \calY) \geq L$. 
\end{lemma}

\begin{proof}
Consider $\Pot_{\eta/2}, \Pot_{\eta/2 + 2L/n}, \dots, \Pot_{\eta/2 + 2Lr/n}$ where $r := \lfloor 0.25\eta n / L \rfloor$. Observe that $\Pot_{\eta/2}(Y)-\Pot_{\eta}(Y) \leq \eta$ by Lemma~\ref{lem:entropy-properties}. Therefore, there must exist $r^* \in [r]$ such that
$$\Pot_{\eta/2 + 2L(r^* - 1)/n} - \Pot_{\eta/2 + 2Lr^*/n} \leq \frac{\eta }{r} \leq \frac{8L}{n}.$$
Let $\tau = \eta/2 + (2Lr^* - 1)/n$ and $\gamma = L$. Then, we have $\stab(\tau, \gamma) \leq \frac{8L}{n}$ and, thus,
\begin{align*}
n \cdot \score(\tau, \calY) \geq \min\{\gamma, 20L - 8L\} \geq L,
\end{align*}
and this finishes the proof.
\end{proof}

Since we are using the exponential mechanism to select a stable
interval, we now show that with high probability, we will
succeed:

\begin{lemma}[Utility of Score Function, Lemma 4.16 of \cite{kothari2021private}] \label{lem:util-score}
Suppose $\calA_{C, \eta/2,n}(\calY)$ is feasible. Let $\epsilon, \delta, \beta \in (0, 1]$. For every $L \in [0, 0.25 \eta n]$, if $L \geq O\left(\frac{1}{\eps} \cdot \log\left(\frac{n}{\beta \delta}\right)\right)$, then with probability $1 - \beta$, Theorem \ref{thm:dp-apx-selection}, invoked with the score function in Definition \ref{def:em-score} and $\kappa = L/2$, does not reject, and the output $\tau$ satisfies $\stab_{\calY}(\tau, L/2) < \frac{20L}{n}$.
\end{lemma}

\begin{proof}
This follows from the guarantee of \SelectionAlg\ (Theorem \ref{thm:dp-apx-selection}), Lemma \ref{lem:existence-of-good-interval} and the definition of $\score$.
\end{proof}

\begin{lemma}[Sensitivity from potential stability]
\label{lem:derivative-argument}
Let $x,x' \in [0,2/n]^n$ such that $|\Ent(x) - \Ent(x')| \leq L/n$ for some $L \geq 0$. Then $\norm{x - x'}_1 \leq O(L/n \log(n))$.
\end{lemma}
\begin{proof}
We begin by writing $x = y + z$ and $x' = y' + z'$, where:
\begin{itemize}
    \item $y'_i = x'_i$ and $y_i = x_i$ for all coordinates $i$ for which $x'_i \geq x_i$ and $y'_i = y_i = 0$ for all other coordinates. Hence, $y' \geq y$ coordinate-wise.
    \item $z'_i = x'_i$ and $z_i = x_i$ for all coordinates $i$ for which $x'_i < x_i$ and $z'_i = z_i = 0$ for all other coordinates. Hence, $z' < z$ coordinate-wise.
\end{itemize}
Next, note that $\Ent(x) = \Ent(y) + \Ent(z), \Ent(x') = \Ent(y') + \Ent(z')$. Also, it holds that $\Ent(y') \geq \Ent(y)$ and $\Ent(z') \leq \Ent(z)$ by monotonicity of the function $g(x) = x \log(1/x) + x$ on the domain $[0,1]$. Assembling these facts, we have by the triangle inequality that 
\[
\max \left(|\Ent(y) - \Ent(y')|, |\Ent(z) - \Ent(z')|\right) \leq |\Ent(x) - \Ent(x')| \leq O(L/n).
\]
Thus, to complete the proof, it suffices to show that $\norm{y-y'}_1 \leq O(L/n)$, as the case for $z,z'$ is symmetric and $\norm{x-x'}_1 = \norm{y-y'}_1 + \norm{z-z'}_1$. Without loss of generality, since $y' \geq y$, we may assume that coordinates of $y'$ are positive. Otherwise, we may decompose both $\Ent(y') - \Ent(y)$ and $y'-y$ into a sum of terms depending on the support and the complement of the support and apply the argument below only on the support of $y'$. 

Because $\Ent$ is a concave function, we have that
\[
\Ent(y') - \Ent(y) \geq \ip{\nabla \Ent(y')}{y'-y},
\]
where $(\nabla \Ent(y'))_i = \log(1/y'_i)$\footnote{Note that $\Ent$ is differentiable at $y'$ since $y' > 0$.}. Since $y'_i \leq 2/n$ for all $i \in [n]$, it holds that $\ip{\nabla \Ent(y')}{y'-y} \geq \log(n/2) \norm{y'-y}_1$. Rearranging and using the fact that $\Ent(y') - \Ent(y) \leq O(L/n)$ gives the desired bound.
\end{proof}

We now state the main guarantee of our private outlier rate selection procedure.

\begin{lemma}[Guarantees of private outlier rate selection]
\label{lem:outlier-rate-select}
Let $\beta \in (0,1)$ be a failure probability parameter, $\epsilon, \delta$ be privacy parameters and $\eta\in [0,1/2)$ be a sufficiently small constant. Let the parameters satisfy $\eta n \geq L$, where $L = \Omega\left(\frac{1}{\epsilon} \log\left(\frac{n}{\beta \delta}\right)\right)$. Then, there is an $(\epsilon,\delta)$-DP mechanism $\calM_1$ that on input dataset $\mathcal{Y} \subset \R^d$ of size $n$, outlier rate $\eta \in (0,1/2)$ and parameter $C$ outputs $\eta' \in [\eta/2, \eta] \cup \{\bot\}$ with the following guarantees:
\begin{enumerate}[(i)]
    \item Let $\mathcal{Y}'$ be a neighboring dataset to $\mathcal{Y}$. If the mechanism does not output $\bot$ on either $\mathcal{Y}$ or $\mathcal{Y}'$, then $|\Pot_{\eta'}(\mathcal{Y}) - \Pot_{\eta'}(\mathcal{Y}')| \leq O(L)$. Furthermore, let $\zeta$ and $\zeta'$ degree-$2t$ pseudo-distributions satisfying $\Pot_{C, \eta, n}(\mathcal{Y}) = \Pot_{C, \eta, n, \zeta}(\mathcal{Y})$ and $\Pot_{C, \eta, n}(\mathcal{Y}') = \Pot_{C, \eta, n, \zeta'}(\mathcal{Y}')$ and $p = \pE_\zeta [w] / \norm{\pE_\zeta [w]}_1, p' = \pE_\zeta [w'] / \norm{\pE_\zeta [w']}_1$. Then, it holds that
    \[
    \norm{p-p'}_1 \leq \frac{120 L}{n}.
    \]
    \item The mechanism outputs $\bot$ with probability at most $\beta$.
\end{enumerate}
\end{lemma}
\begin{proof}
An argument nearly identical to the proof of Lemma 4.17 of \cite{kothari2021private} ensures that the second part of the claim holds and that $|\Pot_{\eta'}(\mathcal{Y}) - Pot_{\eta'}(\mathcal{Y}')| \leq O(L/n)$ .
We will now show that $\norm{p-p'}_1$ satisfies the desired inequality. By assumption, we have that
\[
\left|f\left(\pE_\zeta[w]\right) - f\left(\pE_{\zeta'}[w]\right)\right| \leq O(L/n).
\]
In other words,
\[
\left|\Ent\left(\pE_\zeta[w]\right) - \Ent\left(\pE_{\zeta'}[w]\right)\right| \leq O(L \log(n)/n).
\]
By the constraint $w_i^2 = w_i/n$ in $\calA$, it is straightforward to derive that $\norm{\pE_\zeta[w]}_\infty, \norm{\pE_{\zeta'}[w]}_\infty \leq 2/n$. Thus, $\pE_{\zeta}[w], \pE_{\zeta'}[w]$ satisfy the hypotheses of Lemma \ref{lem:derivative-argument} and we   conclude that 
\[
\norm{\pE_\zeta[w] - \pE_{\zeta'}[w]}_1 \leq O(L/n).
\]
Applying Lemma \ref{lem:renormalization} to the vectors $\pE_\zeta[w], \pE_{\zeta'}[w]$ completes the proof.
\end{proof}

\subsection{Private Robust Mean Estimation}
The building block of our mean estimator is a polynomial constraint system $\calA^{\textsf{mean}}_{C, \eta, n}(\{Y_1, \ldots, Y_n\})$. Our algorithm will solve its  SoS relaxation. 
Algorithm~\ref{alg:robust-mean-conv-program} encodes the
process of non-privately
finding $X'$, a set of samples that intersects with the
corrupted sample in
$(1-\eta)n$ points.

\begin{figure}[htbp]    \centering\begin{mdframed}[style=algo]
    $\calA^{\textsf{mean}}_{C, \eta, n}(\{Y_1, \ldots, Y_n\})$: \textbf{Constraint system for mean estimation}
\begin{enumerate} 
    \item $w_i^2 = w_i/n$ for all $i \in [n]$.
    \item $\sum_{i\in [n]} w_i \geq 1-\eta $.
    \item $w_i \left({X}'_i - Y_i\right) = 0$ for all $i \in [n]$.
    \item $\widetilde{\mu} = \frac{1}{n} \sum_{i \in [n]} {X}'_i$. 
     \item $\frac{1}{n} \sum_{i \in [n]} \ip{{X}'_i - \widetilde{\mu}}{v}^2 \leq C \norm{v}_2^2$.

\end{enumerate}
\end{mdframed}
\caption{Robust mean estimation convex program $\calA^{\textsf{mean}}_{C, \eta, n}(\{Y_1, \ldots, Y_n\})$.}
\label{alg:robust-mean-conv-program}
\end{figure}


Following the framework of
\cite{kothari2021private}, in
Algorithm~\ref{alg:approx-dp-mean-est}, we
provide an end-to-end
$(\eps, \delta)$-DP algorithm for robustly estimating the
mean. The three major steps (stable outlier rate selection,
witness checking, and noise injection) each satisfy
$(\eps/3, \delta/3)$-DP so that the result follows by routine
composition.

\begin{figure}[htbp]
    \centering\begin{mdframed}[style=algo]
\begin{enumerate}
    \item \textbf{Input:} Samples $\calY = \{Y_1, \ldots, Y_n\} \subset \R^d$, outlier rate $\eta \in [0,1/2)$, privacy parameters $\epsilon, \delta > 0$, subgaussian parameter $C > 0$, pseudo-distribution degree parameter $t$.
    
    \item \textbf{Stable outlier rate selection:} Invoke the $(\epsilon/3,\delta/3)$-DP mechanism $\calM_1$ from Lemma \ref{lem:outlier-rate-select} on input $(\calY,\eta, C, \beta/3)$ to sample an outlier rate $\eta'$. If $\eta' = \bot$, then halt.
    
    \item \textbf{Witness checking:} Compute a degree-$2t$ pseudo-distribution $\zeta$ satisfying $\calA^{\textsf{mean}}_{C, \eta', n}(\calY)$ that minimizes $\Ent \left(\pE_\zeta[w]\right)$. Invoke the $(\epsilon/3, \delta/3)$-DP mechanism $\calM_2$ from Lemma \ref{lem:witness-checking} on input $(\calY, \eta,C, \zeta, \beta/3)$ to sample $C'$. If $C' = \bot$, then halt. Otherwise, let $p = \pE_{\zeta} [w] / \norm{\pE_{\zeta} [w]}_1$ and $\widetilde{\mu} = \sum_{i=1}^n p_i Y_i$.
    
    \item \textbf{Noise injection:} Apply the Gaussian mechanism from~Theorem \ref{thm:gaussian-mech} with privacy parameters $\epsilon/3, \delta/3$ and noise level $\sigma = \widetilde{O}\left(C' \sqrt{\log(1/\delta)} / \epsilon \sqrt{n}\right)$ on $\widetilde{\mu}$ to obtain a vector $\widehat{\mu} \in \R^{d}$. 
    
    \item \textbf{Output:} Mean estimate $\widehat{\mu}$.
\end{enumerate}
\end{mdframed}
\caption{Robust mean estimation algorithm}
\label{alg:approx-dp-mean-est}
\end{figure}

The privacy analysis of the algorithm uses the following lemma:

\begin{lemma}[Fact 3.23 of \cite{kothari2021private}]
\label{lem:parameter-closeness}
 Let $Y$ be a set of $n$ points in $\R^d$. Let $p,p' \in [0,1]^n$ be probability vectors such that  $\norm{p-p'}_1 = \tau$. Suppose that   $p$ and $p'$ are distributions   on $Y$, where $p$ is  $2k$-certifiably $C_1$-subgaussian and $p'$ is $2k$-certifiably $C_2$-subgaussian. Let $\mu_p = \sum_i p_i y_i$ and $\Sigma_p = \sum_i p_i (y_i - \mu_p) (y_i-\mu_p)^{\top}$ for every $t \in \N$ be the mean and covariance   distribution defined $p$. Define $\mu_{p'}, \Sigma_{p'}, M^{(t)}_{p'}$ similarly for the distribution corresponding to $p'$.
Then, for every $\tau \leq \eta_0$ for some absolute constant $\eta_0$, for every $u \in \R^d$, $C' = C_1 + C_2$ and $t \leq k$:
\[
\iprod{\mu_p - \mu_{p'},u} \leq \tau^{1-1/2k} \cdot O\left(\sqrt{C k}\sqrt{u^{\top} \Sigma_{p} u}\right).
\]
\end{lemma}

\begin{theorem}[Privacy]
\label{thm:mean-est-approx-dp}
Given $\eps \in (0,1),\delta >0$, any subgaussian parameter $C$, SoS degree $t$ and outlier rate $\eta$, Algorithm~\ref{alg:approx-dp-mean-est} is $(\epsilon,\delta)$-DP provided that $n \geq O \left( \frac{\log^3(n/\delta \beta)}{\epsilon^3} \right)$.
\end{theorem}
\begin{proof}
Observe that Algorithm~\ref{alg:approx-dp-mean-est} is an adaptive composition of $3$ steps, the first two of which may halt early. By Lemma \ref{lem:composition-with-halt}, it suffices to show that each of the 3 steps is $(\epsilon/3, \delta/3)$-DP to conclude that their composition is $(\epsilon, \delta)$-DP. We now verify the privacy of each of these steps:
\begin{itemize}
    \item \textbf{Stable outlier rate selection:} In this step, we invoke the mechanism $\calM_1$ of Lemma \ref{lem:outlier-rate-select}, from which we may conclude that this step is $(\epsilon/3, \delta/3)$-DP.
    \item \textbf{Witness checking:} We may now assume that the algorithm did not halt in the previous step, so that the convex program $\calA^{\textsf{mean}}_{C, \eta', n}$ is feasible. By Lemma \ref{lem:witness-checking} and assuumption on $n$, we may immediately conclude this step is $(\epsilon/3, \delta/3)$-DP.
    \item \textbf{Noise injection:} Let $\calY, \calY'$ be neighboring datasets and assume that the algorithm did not halt in either of the previous two steps for either of $\calY, \calY'$. To invoke the guarantee of the Gaussian mechanism (Theorem \ref{thm:gaussian-mech}), we now verify that $\norm{\widetilde{\mu}(\calY) - \widetilde{\mu}(\calY')}_2 \leq \Delta$ for some $\Delta = \widetilde{O}(1/\sqrt{n})$. Since the algorithm did not halt in the first step, Lemma \ref{lem:outlier-rate-select} implies that $\norm{p(\calY) - p(\calY')}_1 = O(L/n)$, where $L = O(\frac{1}{\epsilon}\log(\frac{n}{\beta \delta}))$. Since the algorithm did not halt in the second step, $p(\calY)$ and $p(\calY')$ induce $C'$-subgaussian distributions on $\calY$ and $\calY'$, respectively. Hence, by Lemma \ref{lem:parameter-closeness}, it holds that $\norm{\widetilde{\mu}(\calY) - \widetilde{\mu}(\calY')}_2 = O(C' \sqrt{L/n})$, which completes the proof.
\end{itemize}
This completes our privacy analysis.
\end{proof}

We now proceed to the utility analysis. It relies upon the following   bound:

\begin{lemma}[Special case of Theorem 1 of \cite{kothari2022polynomial}]
\label{thm:kmz-robust-gaussian-mean}
Let $\calY$ be an $\eta$-corrupted sample of size $n$ from $\calN(\mu, I)$, $C$ be a large enough absolute constant and suppose $n \geq \widetilde{O}(d \log^5(1/\beta)/\eta^2)$. Then with probability $1-\beta$, the degree-$12$ SoS relaxation of $\calA_{C, \eta, n}(\calY)$ is feasible and any satisfying pseudo-distribution $\zeta$ satisfies:
\[
\norm{\sum_{i=1}^n \pE_\zeta [w_i] Y_i - \mu}_2 \leq \widetilde{O}(\eta).
\]
\end{lemma}

With this result in hand, we are ready to show:
\begin{theorem}[Utility]
Let $\calY = \{Y_1, \ldots, Y_n\}$ be an $\eta$-corrupted sample from $\calN(\mu, I_d)$ and assume that 
\[
n \geq \widetilde{O} \left( \max \left(\frac{d \log^5(1/\beta)}{\eta^2}, \frac{\log^3(n/\delta \beta)}{\epsilon^3}\right) \right)
\]
Then with probability $1-\beta$, the output of Algorithm~\ref{alg:approx-dp-mean-est} on input $\calY$ and parameters $\eta \in (0,1/2)$, $\epsilon,\delta > 0   $, $t =12$, $C = O(1)$ satisfies:
\[
\norm{\mu - \widehat{\mu}}_2 \leq O(\eta \log(1/\eta)) + O \left( \frac{1}{\epsilon^{3/2}} \sqrt{\frac{(d + \log(1/\beta)) \log(1/\delta) \log(n/\beta \delta)}{n}} \right).
\]
In particular, if $n \geq \widetilde{O}(d\log^3(n/\beta \delta)/\eta^2 \eps^3)$ as well, we have $\|\mu -\widehat{\mu}\| \leq O(\eta \left(\log(1/\eta)\right)$
\end{theorem}
\begin{proof}
By the guarantees of Lemma \ref{lem:outlier-rate-select} and Lemma \ref{lem:witness-checking}, the algorithm fails to reach the noise injection step with probability at most $2\beta/3$. On the event that Algorithm~\ref{alg:approx-dp-mean-est} does not output halt in the first two steps, we have that $\widehat{\mu} = \widetilde{\mu} + \sigma Z$, where $\sigma = \sqrt{2 \ln (1.25/\delta)} \Delta/\epsilon$, $\Delta = \widetilde{O}(C' \sqrt{L/n}) = O(\sqrt{L/n})$ (by assumption on $C$ and the utility guarantee of Lemma~\ref{lem:witness-checking}), $L = O(\frac{1}{\epsilon}\log(\frac{n}{\beta \delta}))$, and $Z \sim \calN(0,I_d)$, as in the proof of~Theorem \ref{thm:mean-est-approx-dp}. By the triangle inequality, we have
\[
\norm{\mu - \widehat{\mu}}_2 \leq \norm{\mu - \widetilde{\mu}}_2 + \sigma \norm{Z}_2.
\]
By Lemma \ref{lem:chi-squared-tail}, with probability $1 - \frac{\beta}{6}$, it holds that $\norm{Z}_2 \leq \sqrt{d + O(\log (1/\beta))}$. By Theorem \ref{thm:kmz-robust-gaussian-mean} and our assumption on $n$, with probability $1 - \frac{\beta}{6}$, it holds that $\norm{\mu - \widehat{\mu}}_2 \leq O(\eta \log (1/\eta))$.
\end{proof}

\subsection{Private Robust Covariance Estimation}
We now give an algorithm for private and robust Gaussian covariance estimation. The algorithm builds on the  constraint system $\calA^{\textsf{cov}}_{C, \eta, n}(\{Y_1, \ldots, Y_n\})$ and its SoS relaxation.
For any $d$-by-$d$ matrix intermediate $Q$, define $\overline{{x}_i'}^{\top} Q \overline{x_i^{\prime}}=x^{\prime \top} Q x^{\prime}-\frac{1}{n} \sum_{i=1}^n x_i^{\prime \top} Q x_i^{\prime}$.

\begin{figure}[htbp]
    \centering\begin{mdframed}[style=algo]
\begin{enumerate}
    \item $w_i^2 = w_i/n$ for all $i \in [n]$.
    \item $\sum_{i\in [n]} w_i \geq 1-\eta$.
    \item $w_i ({X}'_i - Y_i) = 0$ for all $i \in [n]$.
    \item $\widetilde{\mu} = \frac{1}{n} \sum_{i \in [n]} {X}'_i$. 
        \item  $\Pi^2=\frac{1}{n} \sum_{i=1}^n\left(X^{\prime}_i-\widetilde{\mu}\right)\left(X^{\prime}_i-\widetilde\mu\right)^{\top}$,
    \item $\frac{1}{n} \sum_{i=1}^n X_i^{\prime \top} Q X_i^{\prime}\leq C\|\Pi Q \Pi\|_F^2$
\end{enumerate}
\end{mdframed}
\caption{Robust covariance estimation convex program $\calA^{\textsf{cov}}_{C, \eta, n}(\{Y_1, \ldots, Y_n\})$.}
\label{alg:cov-conv-program}
\end{figure}

Algorithm~\ref{alg:cov-conv-program} encodes non-privately
finding $X'$, a set of samples that intersects with the
corrupted sample in
$(1-\eta)n$ points
and has 4th moments upper bounded
in terms of the squared second moments in all directions. We cite the following property of this constraint system. 
\begin{lemma}[analog of Lemma 5.5 of \cite{kothari2021private}]
\label{lem:rel-frob-sensitivity}
Let $\eta, \epsilon, \delta > 0$ and $L\in \N$ be the input parameters to Algorithm \ref{alg:approx-dp-cov-est} such that $0.25\eta n \geq L = \Omega\left(\frac{1}{\epsilon}\cdot \log\left(\frac{n}{\beta\delta}\right)\right)$.   Let $Y,Y'$ be adjacent datasets. 
Suppose that the algorithm does not halt in any of the   steps and  samples the constant $C'$ in Step 3 and chooses $\eta'$ in Step 2 on input $Y$ and $Y'$. Then, for $\theta = {L/n}$, we have: 
\[
\norm{\Sigma_p^{-1/2} \Sigma_{p'} \Sigma_p^{-1/2} - I}_F \leq O(C') \theta^{1/2}.
\]
\end{lemma}

\begin{lemma}[Theorem 1.3 of \cite{kothari2021private}]
\label{thm:kmz-robust-gaussian-cov}
Let $\eps \in (0,1),\delta > 0$ and $\eta$ be a sufficiently small constant. Let $\calY$ be an $\eta$-corrupted sample of size $n$ from $\calN(0, \Sigma)$, $C$ be larger than some absolute constant and suppose $n \geq \widetilde{O}(d^2 \log^5(1/\beta)/\eta^2)$. Then with probability $1-\beta$, the solution $\widetilde \Sigma = \sum_{i=1}^n p_i Y_iY_i^\top$ in Step 3 of Algorithm~\ref{alg:approx-dp-cov-est} satisfies
\[
 \norm{\Sigma^{-1/2} \widetilde{\Sigma} \Sigma^{-1/2} - I}_F \leq \widetilde{O}\left(C \eta \right).
\]
\end{lemma}
\begin{proof}
This result follows immediately from combining Theorem 1.3 of \cite{kothari2021private}, Corollary 1.3 of \cite{kothari2022polynomial} and the utility guarantee of Lemma~\ref{lem:witness-checking}.    
\end{proof}

\begin{figure}[htbp]
    \centering\begin{mdframed}[style=algo]
\begin{enumerate}
    \item \textbf{Input:} Samples $\calY = \{Y_1, \ldots, Y_n\} \subset \R^d$, outlier rate $\eta \in [0,1/2)$, privacy parameters $\epsilon, \delta > 0$, certifiable hypercontractivity parameter $C > 0$, pseudo-distribution degree parameter $t$,  and Gaussian sampling parameter $k$.
    
    \item \textbf{Stable outlier rate selection:} Invoke the $(\epsilon/3,\delta/3)$-DP mechanism $\calM_1$ from Lemma \ref{lem:outlier-rate-select} on input $(\calY_2,\eta, C, \beta/3)$ to sample an outlier rate $\eta'$. If $\eta' = \bot$, then halt.
    
    \item \textbf{Witness checking:} Compute a degree-$2t$ pseudo-distribution $\zeta$ satisfying $\calA^{\textsf{cov}}_{C, \eta', n}(\calY_2)$ that minimizes $\Ent \left(\pE_\zeta[w]\right)$. Invoke the $\left(\epsilon/3, \delta/3\right)$-DP mechanism $\calM_2$ from Lemma \ref{lem:witness-checking} on input $(\calY_2, \eta,C, \zeta, \beta/3)$ to sample $C'$. If $C' = \bot$, then halt. Otherwise, let $p = \pE_{\zeta} [w] / \norm{\pE_{\zeta} [w]}_1$ and $\widetilde{\Sigma} = \sum_{i=1}^n p_i Y_i Y_i^T$.
    
    \item \textbf{Noise injection:} Invoke the $(\epsilon/3, \delta/3)$-DP mechanism in Algorithm~\ref{alg:cov-noise-mechanism} on input $\widetilde{\Sigma}$ to obtain a matrix $\widehat{\Sigma} \in \R^{d \times d}$ with parameter $k$.
    
    \item \textbf{Output:} Covariance estimate $\widehat{\Sigma}$.
\end{enumerate}
\end{mdframed}
\caption{Covariance estimation algorithm}
\label{alg:approx-dp-cov-est}
\end{figure}

Following the framework of
\cite{kothari2021private}, in
Algorithm~\ref{alg:approx-dp-cov-est}, we
provide an end-to-end
$(\eps, \delta)$-DP algorithm for robustly estimating the
covariance matrix.
The three major steps (stable outlier rate selection,
witness checking, and noise injection) each satisfy
$(\eps/3, \delta/3)$-DP so that the privacy guarantee follows by routine
composition.

\begin{theorem}[Privacy]
Given   $\eps, \delta\in(0, 1]$, any SoS degree $t$ and outlier rate $\eta$, Algorithm~\ref{alg:approx-dp-cov-est} with parameter $k$ is $(\epsilon, \delta)$-DP provided that
\[
n \geq O \left( \frac{(C' k \log (1/\delta) + \log^2 (1/\delta)) \cdot \log(n/\beta \delta)}{\epsilon^{2.5}} + \frac{\log^3(n/\delta \beta)}{\epsilon^3}\right)
\]
\end{theorem}
\begin{proof}
Observe that Algorithm~\ref{alg:approx-dp-cov-est} is an adaptive composition of 3 steps, each of which may halt early. By Lemma \ref{lem:composition-with-halt}, it suffices to show that each of the 3 steps is $(\epsilon/3, \delta/3)$-DP to conclude that their composition is $(\epsilon, \delta)$-DP. The parts of Algorithm~\ref{alg:approx-dp-cov-est} that involve \textbf{Stable outlier rate selection} and \textbf{Witness checking} can be shown to satisfy the desired privacy guarantees in an identical way to the proof of privacy of Algorithm~\ref{alg:approx-dp-mean-est} for mean estimation in Theorem \ref{thm:mean-est-approx-dp}. We now verify the privacy of the remaining step:
\begin{itemize}
    \item \textbf{Noise injection:}  Let $\calY, \calY'$ be neighboring datasets and assume that the algorithm did not halt in any of the previous steps for either of $\calY, \calY'$. To invoke guarantee of the Gaussian sampling mechanism (Theorem \ref{thm:gaussian-sampling-mech}), it suffices to verify that $\norm{\widetilde{\Sigma}^{1/2}(\calY) \widetilde{\Sigma}(\calY')^{-1} \widetilde{\Sigma}(\calY)^{1/2} - I}_F \leq \Delta$ for $\Delta$ sufficiently small. 
    
    Since the algorithm did not halt after \textbf{Stable outlier rate selection}, Lemma \ref{lem:outlier-rate-select} implies that $\norm{p(\calY) - p(\calY')}_1 = O(L/n)$, where $L = O(\frac{1}{\epsilon}\log(\frac{n}{\beta \delta}))$. Since the algorithm did not halt \textbf{Witness checking}, $p(\calY)$ and $p(\calY')$ induce $2t$-certifiably $C'$-hypercontractive distributions on $\calY$ and $\calY'$, respectively. 
    Hence, by Lemma \ref{lem:rel-frob-sensitivity}, it holds that \[\norm{\widetilde{\Sigma}^{1/2}(\calY) \widetilde{\Sigma}(\calY')^{-1} \widetilde{\Sigma}(\calY)^{1/2} - I}_F \leq \Delta = O\left(C' \sqrt{L/n} \right).\] This value of $\Delta$ satisfies the conditions of Theorem \ref{thm:gaussian-sampling-mech} by our assumption on $n$.
    \end{itemize}
\end{proof}

\begin{theorem}[Utility]
Let $\calY = \{Y_1, \ldots, Y_n\}$ be an $\eta$-corrupted sample from $\calN(0, \Sigma)$ and assume that 
\[
n \geq \widetilde{O}\left( \max \left(\frac{d^2 \log^5(1/\beta)}{\eta^2}, \frac{(k \log (1/\delta) + \log^2 (1/\delta)) \cdot \log(n/\beta \delta)}{\epsilon^{2.5}}, \frac{\log^3(n/\delta \beta)}{\epsilon^3} \right) \right).
\]
Then with probability $1-\beta$, the output of Algorithm~\ref{alg:approx-dp-cov-est} on input $\calY$ and parameters $\eta \in (0,1/2)$, $\epsilon,\delta\in(0, 1]$, $t =12$, $C = O(1)$ large enough and any positive integer $k$ satisfies:  
\[
\norm{\Sigma^{-1/2} \widehat{\Sigma} \Sigma^{-1/2} - I}_F \leq \widetilde{O}\left(\eta \right) + O \left( \sqrt{\frac{d^2 + \log(1/\beta)}{k}}\right)
\]
In particular, if $k\geq \widetilde{O}\left( (d^2 + \log (1/\beta))/\eta^2\right)$ as well, we have 
\[
\norm{\Sigma^{-1/2} \widehat{\Sigma} \Sigma^{-1/2} - I}_F \leq \widetilde{O}\left(\eta \right).
\]
\end{theorem}
\begin{proof}
By the guarantees of Lemmas~\ref{lem:outlier-rate-select}, \ref{lem:witness-checking}, the algorithm fails to reach the noise injection step with probability at most $2\beta/3$. Assume now that Algorithm~\ref{alg:approx-dp-cov-est} does not output halt in the first two steps. Then we have that:
\begin{align*}
\norm{\Sigma^{-1/2} \widehat{\Sigma} \Sigma^{-1/2} - I}_F &\leq \norm{\Sigma^{-1/2} \widetilde{\Sigma} \Sigma^{-1/2} - I}_F + \norm{\Sigma^{-1/2} (\widetilde{\Sigma} - \widehat{\Sigma}) \Sigma^{-1/2}}_F \\
&\leq \norm{\Sigma^{-1/2} \widetilde{\Sigma} \Sigma^{-1/2} - I}_F + \norm{\Sigma^{-1/2} \widetilde{\Sigma}^{1/2}}_{2}^2 \norm{\widetilde{\Sigma}^{-1/2} \widehat{\Sigma} \widetilde{\Sigma}^{-1/2} - I}_F.
\end{align*}
By Theorem \ref{thm:kmz-robust-gaussian-cov} and choice  of $n$, the first term is at most $\widetilde{O}\left(\eta \right)$, and $\norm{\Sigma^{-1/2} \widetilde{\Sigma}^{1/2}}_{2}^2$ is at most $O(1)$. Finally, $\norm{\widetilde{\Sigma}^{-1/2} \widehat{\Sigma} \widetilde{\Sigma}^{-1/2} - I}_F$ is bounded by invoking the utility guarantee of Theorem \ref{thm:gaussian-sampling-mech}. 
\end{proof}

\section{Pure DP Covariance Estimation Lower Bound}
\label{sec:pure-lower}

In this section, we detail an information-theoretic
lower bound for Gaussian covariance estimation under
differential privacy constraints. Like our upper bound,
our lower bound has a logarithmic dependence on $\kappa$.
Our proof builds on previous packing-style lower bound
arguments for pure DP estimation \cite{HardtTa10, aden2021sample, BunKSW21, kamath2019privately}.

\paragraph{High-level Overview}
For any integer $d\geq 2$,
given any covariance matrix
$\Sigma\in\reals^{d\times d}$ that satisfies
$I \preceq \Sigma \preceq \kappa I$,
we show that to learn $\calN(0, \Sigma)$ within
total variation $O(\alpha)$, we must require sample complexity
$\Omega\left(\frac{d^2}{\eps}\log\left(\frac{d\kappa}{\alpha}\right) + \frac{d^2}{\eps\alpha}\right)$.
We proceed in two steps:
first we show a lower bound of
$\Omega\left(\frac{d^2}{\eps}\log\left(\frac{d\kappa}{\alpha}\right)\right)$
which does not depend polynomially on $1/\alpha$.  Then we combine it with a previous lower bound of
$\Omega\left(\frac{d^2}{\eps\alpha}\right)$, due to \cite{kamath2019privately}.
Taken together, this gives us a lower bound of
$\Omega\left(\frac{d^2}{\eps}\log\left(\frac{d\kappa}{\alpha}\right) + \frac{d^2}{\eps\alpha}\right)$:

\begin{theorem}
\label{thm:pure-dp-cov-est-lower-bound-full}
Let $\eps, \alpha \in (0, 1)$, $d\geq 2$.
Any $\epsilon$-DP algorithm that, given $n$ i.i.d.\ samples $\calX = \{X_1,X_2, \ldots, X_n\}$ from $\calN(0, \Sigma)$ for an unknown $\Sigma \in \R^{d \times d}$ satisfying
$I \preceq \Sigma \preceq \kappa I$,
outputs $\widehat{\Sigma} = \widehat{\Sigma}(\calX)$ such that, with probability at least $0.9$,
\[
\dtv\left(\calN(0, \Sigma), \calN(0, \widehat{\Sigma})\right) < O(\alpha),
\]
must require
\[
n = \Omega\left(\frac{d^2}{\eps} \log\left(\frac{d \kappa}{\alpha}\right) + \frac{d^2}{\eps \alpha}\right).
\]
\end{theorem}

\subsection{Condition Number Lower Bound}

We rely on previous work on differentially
private hypothesis selection~\cite{BunKSW21}:
given samples from some unknown distribution $P$
(e.g., defined by $\calN(0, \Sigma)$ for some
$\Sigma\in\reals^{d\times d}$), what is the closest 
distribution to $P$ in some set $\calH$?

Crucial to the derivation of the lower bound of
$\Omega\left(\frac{d^2}{\eps}\log\left(\frac{d\kappa}{\alpha}\right)\right)$ is the notion of covers and packings:

\begin{definition}[$\gamma$-Cover]

A $\gamma$-cover of a set of distributions $\calH$
is a set of distributions $\calC_\gamma$, such that
for every $H\in\calH$ there exists $P\in\calC_\gamma$
with the following property:
$\dtv(P, H)\leq\gamma$.

\end{definition}

\begin{definition}[$\gamma$-Packing]

A $\gamma$-packing of a set of distributions $\calH$
is a set of distributions
$\calP_\gamma\subseteq\calH$, such that
for every pair $P, Q\in\calP_\gamma$,
$\dtv(P, Q)>\gamma$.

\end{definition}

The following lemma states that, provided we can find an
$\alpha$-packing, we can get a sample complexity lower bound
for pure DP:

\begin{lemma}[Lemma 5.1 in~\cite{BunKSW21}]

Let $\calP_\alpha$ be an $\alpha$-packing of a set of
distributions $\calH$. Then for any $P\in\calH$, any
$\eps$-differentially private algorithm that takes
sample $X_1, \ldots, X_n\sim P$ produces,
with probability at least 0.99, a distribution
$\widehat{H}$ such that $\dtv(P, \widehat{H})\leq\alpha/2$
requires
$$
n = \Omega\left(\frac{\log |\calP_\alpha|}{\eps}\right).
$$
\label{lem:bksw51}
\end{lemma}

Note that the sample complexity lower bound
in~Lemma \ref{lem:bksw51} is not of the
form 
$\Omega\left(\frac{\log |\calP_\alpha|}{\eps\alpha}\right)$
since such a lower bound would contradict already-existing
upper bounds.

Lemma \ref{lem:bksw68} shows the existence of an
$\alpha$-cover of a set of $d$-dimensional
Gaussian distributions:

\begin{lemma}[Lemma 6.8 in~\cite{BunKSW21}]

Let $\mu\in\reals^d, \Sigma\in\reals^{d\times d}$
such that $\norm{\mu}_2\leq R$ and
$I \preceq \Sigma \preceq \kappa I$. 
Then there exists an $\alpha$-cover of the set of 
Gaussian distributions $\calN(\mu, \Sigma)$ of size
$$
O\left(\frac{dR}{\alpha}\right)^d\cdot O\left(\frac{d\kappa}{\alpha}\right)^{d(d+1)/2}.
$$
\label{lem:bksw68}
\end{lemma}
To use the lower bound from
Lemma \ref{lem:bksw51}, we need an $\alpha$-packing and not
an $\alpha$-cover. The following lemma relates the
size of the largest $\alpha$-packing to the smallest
$\alpha$-cover:

\begin{lemma}[Lemma 5.2 in~\cite{BunKSW21}]\label{lem:bksw52}
Let $\calH$ be a set of distributions.
If $p_\alpha$ and $c_\alpha$ are the size of the
largest $\alpha$-packing and smallest $\alpha$-cover
of $\calH$, respectively, then
\[
p_{2\alpha} \leq c_\alpha \leq p_\alpha.
\]
\end{lemma}
We can now obtain the following corollary of a sample
complexity lower bound that depends on $\kappa$:

\begin{corollary}\label{cor:lbkappa}
Fix $\eps,\alpha\in(0, 1)$.
For any integer $d\geq 2$ and covariance matrix
$\Sigma\in\reals^{d\times d}$ satisfying
$I \preceq \Sigma \preceq \kappa I$,
let $\Hat\Sigma = \Hat\Sigma(\calX)$ be any $\eps$-DP
algorithm such that with
probability at least 0.99,
given $n$ i.i.d. samples
$\calX = \{X_1, \ldots, X_n\}$ from $\calN(0, \Sigma)$,
the algorithm has the following guarantee:
$$
\dtv(\calN(0, \Sigma), \calN(0, \Hat\Sigma)) < \alpha/2.
$$
Then
$$n = \Omega\left(\frac{d^2}{\eps}\log\left(\frac{d\kappa}{\alpha}\right)\right).$$
\end{corollary}
\begin{proof}
Set $R = O(\alpha/d)$ since the
mean of the Gaussian is a constant.
By~Lemma \ref{lem:bksw52} and~Lemma \ref{lem:bksw68},
there exists an $\alpha$-packing of size
$O((\frac{d\kappa}{\alpha})^{d^2})$.
And by~Lemma \ref{lem:bksw51}, this gives us a sample complexity
lower bound of
$n = \Omega(\frac{d^2}{\eps}\log(\frac{d\kappa}{\alpha}))$.

\end{proof}

\subsection{Precision Matrix Lower Bound}

We now proceed to show the lower bound of
$\Omega(\frac{d^2}{\eps\alpha})$.
The proof of the theorem relies on the following technical lemma on the TV distance between two mean-zero Gaussians with different covariance.
By Theorem~\ref{thm:dtv-eigen}, it suffices to derive
a lower bound on
$\norm{\Sigma_1^{1/2}\Sigma_2^{-1}\Sigma_1^{1/2} - I_d}_F$
where $\Sigma_1$ and $\Sigma_2$ are the unknown
covariance matrix and the output from the $\eps$-DP algorithm,
respectively.

\begin{lemma}[Lemma 3.5 in~\cite{DMR18b}, Theorem 1.1 in~\cite{DMR18a}]
Let $\mu\in\R^d$ and let
$\Sigma_1, \Sigma_2$ be positive definite symmetric $d\times d$ matrices.
Use $\lambda_1, \ldots, \lambda_d$ to denote the eigenvalues of
$\Sigma_1^{-1}\Sigma_2 - I_d$. Then,
$$
0.01 \leq \frac{\dtv(\calN(\mu, \Sigma_1), \calN(\mu, \Sigma_2))}{\min\left\{1, \sqrt{\sum_{i=1}^d\lambda_i^2}\right\}} \leq 1.5.
$$
Also,
$$
\dtv(\calN(0, \Sigma_1), \calN(0, \Sigma_2))
\geq \frac{1}{100} \min \left\{ 1, \norm{\Sigma_1^{1/2}\Sigma_2^{-1}\Sigma_1^{1/2} - I_d}_F \right\}.$$
\label{thm:dtv-eigen}
\end{lemma}

For two Gaussians with the same mean,~\cite{DMR18a} gives closed-form
lower and upper bounds in TV distance. The upper and lower bounds are within
(small) constants of each other.

Note that since $\Sigma_1^{-1}\Sigma_2$ and
$\Sigma_1^{-1/2}\Sigma_2\Sigma_1^{-1/2}$ have the same spectrum, we have
\[\sum_{i=1}^d\lambda_i^2 = \norm{\Sigma_1^{-1/2}\Sigma_2\Sigma_1^{-1/2} - I_d}_F^2.\]

The following statement, although not explicitly stated in
\cite{kamath2019privately}, can be inferred from \cite{kamath2019privately}. We
combine the statement with cleaner versions of~\cite{DMR18b} to 
derive our lower bound.

\begin{proposition}[see also \cite{kamath2019privately}]\label{lem:frob-sc}
Let $d\geq 2$.
Let $\widehat\Sigma$ be an $\eps$-DP algorithm that outputs an approximation of
the covariance matrix $\Sigma$, where $\frac{1}{2}I \preceq \Sigma \preceq 2 I$,
such that
$$
\E_{X\sim\calN(0, \Sigma)^{\otimes n}}\left[\norm{\left[\widehat\Sigma(X)\right]^{-1} - \Sigma^{-1}}_F^2\right] \leq \frac{\alpha^2}{64}.
$$
Then $n = \Omega\left(\frac{d^2}{\eps\alpha}\right)$.
\end{proposition}

\begin{proof}

We can equivalently prove that
$\E_{X\sim\calN(0, \Sigma)^{\otimes n}}\left[\norm{\widehat\Sigma(X) - \Sigma}_F^2\right] \leq \frac{\alpha^2}{64}$.
We assume that $n = O(\frac{d^2}{\eps\alpha})$ and will reach a contradiction.

Let $\mathcal{M}_d$ denote the set of $d$-by-$d$ real symmetric matrices.
Consider $$\calS_n = \left\{S\in\mathcal{M}_d\,:\,S \text{ has diagonal entries that are } 0 \text{ and  non-diagonals   either }-\frac{\alpha}{2d} \text{ or }+\frac{\alpha}{2d}\right\}.$$
Clearly, $$|\calS_n|= 2^{(d^2-d)/2}.$$
Let $U \sim\Unif(\calS_n)$ be the uniform distribution over $\calS_n$.
For any $V\in\calS_n$, define
$\Sigma = \Sigma(V) = I + V$. The $\eps$-DP algorithm $\widehat\Sigma$ aims to
output a matrix as close to $\Sigma$ as possible.

Define $Z$ and $Z'$ such that
\begin{equation}
Z = \iprod{\widehat\Sigma(X), V} = 2\sum_{i < j}\widehat\Sigma(X)_{ij}\cdot V_{ij},
\label{eq:zsum}
\end{equation}
and
\begin{equation}
Z' = \iprod{\widehat\Sigma(X'), V} = 2\sum_{i < j}\widehat\Sigma(X')_{ij}\cdot V_{ij}.
\label{eq:zsump}
\end{equation}

Let $V, V'$ be independent samples from $\Unif(\calS_n)$.
Then sample $X\sim\calN(0, \Sigma(V))^{\otimes n}$ and
$X'\sim\calN(0, \Sigma(V'))^{\otimes n}$.
Also, note that
$\norm{\widehat\Sigma(X) - \Sigma(V)}_F^2 = \sum_{i < j}2(\widehat\Sigma_{ij}(X) - \Sigma_{ij}(V))^2$.

Then, by Lemma \ref{lem:lowerz},
$\E[Z] \geq \frac{\alpha^2}{16} - \frac{1}{2}\norm{\widehat\Sigma(X) - \Sigma(V)}_F^2 \geq \frac{7\alpha^2}{128}$.

Finally, observe that by Lemma \ref{lem:zneighbor}
and~Lemma \ref{lem:lowerz},
$\pr[Z > \alpha^2/32] = \Omega(1)$ and
$\pr[Z' > \alpha^2/32] \leq \exp(-\Omega(d^2))$ but Lemma \ref{lem:zneighbor} leads to
$\Omega(1) \leq \exp(4\alpha\eps n) - \Omega(d^2)$ which would require
$n = \Omega(\frac{d^2}{\eps\alpha})$.
\end{proof}

The following lemma is used to lower bound the expected value of
$Z = \iprod{\widehat\Sigma(X), V}$ in terms of
$\norm{\widehat\Sigma(X) - \Sigma(V)}_F^2$:

\begin{lemma}[Claim 6.12 of \cite{kamath2019privately}]
For any $d\geq 2$,
$$
\E[Z] \geq \frac{\alpha^2}{16} - \frac{1}{2}\norm{\widehat\Sigma(X) - \Sigma(V)}_F^2 \geq \frac{7\alpha^2}{128}.
$$
\label{lem:lowerz}
\end{lemma}


The following lemma relates $Z$ to $Z'$:

\begin{lemma}[Claim 6.13 and 6.14 of \cite{kamath2019privately}]

For $Z, Z'$ as defined in Equations~(\ref{eq:zsum}) and
~(\ref{eq:zsump}), we have:

\begin{enumerate}

\item
$\pr[Z > \alpha^2/32] \leq \exp(4(\alpha\eps n))\cdot \pr[Z' > \alpha^2/32]$.

\item
$\pr[Z' > \alpha^2/32] \leq \exp(-\Omega(d^2))$.

\end{enumerate}
\label{lem:zneighbor}
\end{lemma}

\subsection{Putting it Together}

We now prove Theorem \ref{thm:pure-dp-cov-est-lower-bound-full}:

\begin{proof}[Proof of~Theorem \ref{thm:pure-dp-cov-est-lower-bound-full}]

By Theorem~\ref{thm:dtv-eigen},
to prove~\autoref{thm:pure-dp-cov-est-lower-bound}, it suffices to show
that it is impossible to have both
$\norm{\Sigma^{1/2}\widehat\Sigma^{-1}\Sigma^{1/2} - I}_F \leq \alpha$
and $n = O\left(\frac{d^2}{\eps\alpha}\right)$ using an
$\eps$-DP algorithm to compute $\widehat\Sigma = \widehat\Sigma(\calX)$.

First note that
\begin{align}
\norm{\Sigma^{1/2}\widehat\Sigma^{-1}\Sigma^{1/2} - I}_F 
&= \norm{\Sigma^{1/2}(\widehat\Sigma^{-1} - \Sigma^{-1})\Sigma^{1/2}}_F  \\
& \geq \sigma_d(\Sigma^{1/2})^2\norm{\widehat\Sigma^{-1} - \Sigma^{-1}}_F,
\end{align}
where
$\sigma_d(\Sigma^{1/2})$ denotes the smallest singular value of $\Sigma^{1/2}$.
Because $\Sigma^{1/2}$ is symmetric and positive definite, we know that
eigenvalues coincide with its singular values.
Also, the eigenvalues 
of $\Sigma^{1/2}$ are larger than $1/\sqrt{2}$ so that
\[
\norm{\Sigma^{1/2}\widehat\Sigma^{-1}\Sigma^{1/2} - I}_F 
\geq \frac{1}{2}\norm{\widehat\Sigma^{-1} - \Sigma^{-1}}_F.
\]
By Lemma~\ref{lem:frob-sc}, 
$\norm{\widehat\Sigma^{-1} - \Sigma^{-1}}_F^2 \geq \frac{\alpha^2}{64}$.
As a result, 
$\norm{\Sigma^{1/2}\widehat\Sigma^{-1}\Sigma^{1/2} - I}_F \geq \alpha/16$, contradicting 
the assumption that 
\[\dtv\left(\calN(0, \Sigma), \calN(0, \widehat{\Sigma})\right) < \alpha/1600.\]
Thus, $n = \Omega\left(\frac{d^2}{\eps\alpha}\right)$.
By Corollary~\ref{cor:lbkappa}, we obtain a lower bound
of $n = \Omega(\frac{d^2}{\eps}\log(\frac{d\kappa}{\alpha}))$.
As a result, the lower bound is a max of
$n = \Omega\left(\frac{d^2}{\eps\alpha}\right)$
and 
$n = \Omega\left(\frac{d^2}{\eps}\log(d \kappa/\alpha)\right)$ which asymptotically is
$$
n = \Omega\left(\frac{d^2}{\eps} \log\left(\frac{d \kappa}{\alpha}\right) + \frac{d^2}{\eps \alpha}\right).
$$
This completes the proof.
\end{proof}

\section{Conclusion}
\label{sec:conclusion}


In this work, we have designed differentially private 
algorithms that achieve
the optimal sample complexity for privately estimating a Gaussian in high
dimensions. For the pure DP setting, via the use of a recursive
preconditioning approach, we present the first polynomial-time algorithm
that needs $\tilde{O}(d^2\log \kappa)$ samples. Through
a lower bound argument,
we show that the dependence on $\kappa$ is necessary. In the approximate
DP setting, we present an algorithm that needs $\widetilde{O}(d^2)$ samples
(no condition number dependence). We leave the following open questions for future work:
\begin{enumerate}

\item All the algorithms presented in this work only apply to the central
model of DP where the data curator is trusted. However, other trust models,
such as local DP~\cite{Warner65, EvfimievskiGS03, KLNRS11}, are of practical interest. What
sample and communication complexity bounds for learning a Gaussian are required in such models?

\item In this work, we focus on the fundamental tasks of mean and covariance estimation. It would be interesting to apply our framework to other statistical tasks (e.g., sparse mean estimation~\cite{georgiev2022privacy} or
regression) to see if it can achieve optimal sample complexities beyond just mean and covariance estimation.

\item For the purpose of learning a Gaussian to within small TV distance, we measure the quality of our covariance estimate with respect to the Manalanobis distance, which require $\Omega (d^2)$ samples. Is there a polynomial-time DP algorithm for estimating the covariance to within small \textit{spectral norm error} with $o(d^2)$ samples?

\end{enumerate}

\section*{Acknowledgement}
Fred Zhang would like to thank Weihao Kong for a helpful discussion about \cite{LiuKKO21,LiuKO22}.

\bibliography{main}

\newcommand{\etalchar}[1]{$^{#1}$}
\begin{thebibliography}{KMS{\etalchar{+}}22b}

\bibitem[AAAK21]{aden2021sample}
Ishaq Aden-Ali, Hassan Ashtiani, and Gautam Kamath.
\newblock On the sample complexity of privately learning unbounded
  high-dimensional gaussians.
\newblock In {\em Algorithmic Learning Theory (ALT)}, 2021.

\bibitem[AL22]{ashtiani2022private}
Hassan Ashtiani and Christopher Liaw.
\newblock Private and polynomial time algorithms for learning gaussians and
  beyond.
\newblock In {\em Conference on Learning Theory (COLT)}, 2022.

\bibitem[AS66]{AS66}
S.~M. Ali and S.~D. Silvey.
\newblock A general class of coefficients of divergence of one distribution
  from another.
\newblock {\em Journal of the Royal Statistical Society: Series B
  (Methodological)}, 28(1):131--142, 1966.

\bibitem[BDKU20]{biswas2020coinpress}
Sourav Biswas, Yihe Dong, Gautam Kamath, and Jonathan Ullman.
\newblock Coinpress: Practical private mean and covariance estimation.
\newblock In {\em Advances in Neural Information Processing Systems (NeurIPS)},
  2020.

\bibitem[BGS{\etalchar{+}}21]{brown2021covariance}
Gavin Brown, Marco Gaboardi, Adam Smith, Jonathan Ullman, and Lydia
  Zakynthinou.
\newblock Covariance-aware private mean estimation without private covariance
  estimation.
\newblock In {\em Advances in Neural Information Processing Systems}, 2021.

\bibitem[BKN10]{BeimelKN10}
Amos Beimel, Shiva~Prasad Kasiviswanathan, and Kobbi Nissim.
\newblock Bounds on the sample complexity for private learning and private data
  release.
\newblock In {\em Conference on Theory of Cryptography (TCC)}, Lecture Notes in
  Computer Science, 2010.

\bibitem[BKSW21]{BunKSW21}
Mark Bun, Gautam Kamath, Thomas Steinke, and Zhiwei~Steven Wu.
\newblock Private hypothesis selection.
\newblock {\em IEEE Transactions on Information Theory}, 67(3):1981--2000,
  2021.

\bibitem[BS98]{BonehS98}
Dan Boneh and James Shaw.
\newblock Collusion-secure fingerprinting for digital data.
\newblock {\em {IEEE} Trans. Inf. Theory}, 44(5):1897--1905, 1998.

\bibitem[BS16]{bun2016concentrated}
Mark Bun and Thomas Steinke.
\newblock Concentrated differential privacy: Simplifications, extensions, and
  lower bounds.
\newblock In Martin Hirt and Adam~D. Smith, editors, {\em Conference on Theory
  of Cryptography (TCC)}, 2016.

\bibitem[BUV14]{BunUlVa14}
Mark Bun, Jonathan Ullman, and Salil Vadhan.
\newblock Fingerprinting codes and the price of approximate differential
  privacy.
\newblock In {\em Proceedings of the 46th Annual ACM Symposium on Theory of
  Computing (STOC)}, 2014.

\bibitem[CSI67]{CI67}
I.~CSISZAR.
\newblock Information-type measures of difference of probability distributions
  and indirect observation.
\newblock {\em Studia Scientiarum Mathematicarum Hungarica}, 2:229--318, 1967.

\bibitem[DK19]{diakonikolas2019recent}
Ilias Diakonikolas and Daniel~M Kane.
\newblock Recent advances in algorithmic high-dimensional robust statistics.
\newblock {\em arXiv preprint arXiv:1911.05911}, 2019.

\bibitem[DKK{\etalchar{+}}18]{diakonikolas2018robustly}
Ilias Diakonikolas, Gautam Kamath, Daniel~M Kane, Jerry Li, Ankur Moitra, and
  Alistair Stewart.
\newblock Robustly learning a gaussian: Getting optimal error, efficiently.
\newblock In {\em Proceedings of the Twenty-Ninth Annual ACM-SIAM Symposium on
  Discrete Algorithms (SODA)}, 2018.

\bibitem[DKK{\etalchar{+}}19]{diakonikolas2019robust}
Ilias Diakonikolas, Gautam Kamath, Daniel Kane, Jerry Li, Ankur Moitra, and
  Alistair Stewart.
\newblock Robust estimators in high-dimensions without the computational
  intractability.
\newblock {\em SIAM Journal on Computing}, 48(2):742--864, 2019.

\bibitem[DKM{\etalchar{+}}06]{DKMMN06}
Cynthia Dwork, Krishnaram Kenthapadi, Frank McSherry, Ilya Mironov, and Moni
  Naor.
\newblock Our data, ourselves: Privacy via distributed noise generation.
\newblock In {\em Annual International Conference on Theory and Application of
  Cryptographic Techniques (EUROCRYPT)}, 2006.

\bibitem[DMNS06]{DworkMNS06}
Cynthia Dwork, Frank McSherry, Kobbi Nissim, and Adam~D. Smith.
\newblock Calibrating noise to sensitivity in private data analysis.
\newblock In {\em Third Theory of Cryptography Conference on Theory of
  Cryptography (TCC)}, 2006.

\bibitem[DMR18]{DMR18a}
Luc Devroye, Abbas Mehrabian, and Tommy Reddad.
\newblock The total variation distance between high-dimensional gaussians.
\newblock {\em arXiv preprint arXiv:1810.08693}, 2018.

\bibitem[DMR20]{DMR18b}
Luc Devroye, Abbas Mehrabian, and Tommy Reddad.
\newblock The minimax learning rates of normal and ising undirected graphical
  models.
\newblock {\em Electronic Journal of Statistics}, 14(1):2338--2361, 2020.

\bibitem[DRV10]{dwork2010boosting}
Cynthia Dwork, Guy~N Rothblum, and Salil Vadhan.
\newblock Boosting and differential privacy.
\newblock In {\em IEEE 51st Annual Symposium on Foundations of Computer Science
  (FOCS)}, 2010.

\bibitem[EGS03]{EvfimievskiGS03}
Alexandre~V. Evfimievski, Johannes Gehrke, and Ramakrishnan Srikant.
\newblock Limiting privacy breaches in privacy preserving data mining.
\newblock In {\em Proceedings of the Twenty-Second {ACM} {SIGACT-SIGMOD-SIGART}
  Symposium on Principles of Database Systems (PODS)}, 2003.

\bibitem[FKP{\etalchar{+}}19]{TCS-086}
Noah Fleming, Pravesh Kothari, Toniann Pitassi, et~al.
\newblock Semialgebraic proofs and efficient algorithm design.
\newblock {\em Foundations and Trends{\textregistered} in Theoretical Computer
  Science}, 14(1-2):1--221, 2019.

\bibitem[GH22]{georgiev2022privacy}
Kristian Georgiev and Samuel~B Hopkins.
\newblock Privacy induces robustness: Information-computation gaps and sparse
  mean estimation.
\newblock In {\em Advances in Neural Information Processing Systems (NeurIPS)},
  2022.

\bibitem[HKM22]{hopkins2021efficient}
Samuel~B Hopkins, Gautam Kamath, and Mahbod Majid.
\newblock Efficient mean estimation with pure differential privacy via a
  sum-of-squares exponential mechanism.
\newblock In {\em Proceedings of the 54th Annual ACM SIGACT Symposium on Theory
  of Computing (STOC)}, 2022.

\bibitem[HKMN22]{HKMN22}
Samuel~B. Hopkins, Gautam Kamath, Mahbod Majid, and Shyam Narayanan.
\newblock Robustness implies privacy in statistical estimation, 2022.

\bibitem[HT10]{HardtTa10}
Moritz Hardt and Kunal Talwar.
\newblock On the geometry of differential privacy.
\newblock In {\em {P}roceedings of the 2010 {ACM} {I}nternational {S}ymposium
  on {T}heory of {C}omputing (STOC)}, 2010.

\bibitem[Hub64]{huber1964robust}
Peter~J Huber.
\newblock Robust estimation of a location parameter.
\newblock {\em The Annals of Mathematical Statistics}, 35(1):73--101, 1964.

\bibitem[KLN{\etalchar{+}}11]{KLNRS11}
Shiva~Prasad Kasiviswanathan, Homin~K. Lee, Kobbi Nissim, Sofya Raskhodnikova,
  and Adam~D. Smith.
\newblock What can we learn privately?
\newblock {\em {SIAM} J. Comput.}, 40(3):793--826, 2011.

\bibitem[KLSU19]{kamath2019privately}
Gautam Kamath, Jerry Li, Vikrant Singhal, and Jonathan Ullman.
\newblock Privately learning high-dimensional distributions.
\newblock In {\em Conference on Learning Theory (COLT)}, 2019.

\bibitem[KMS22a]{KMS22}
Gautam Kamath, Argyris Mouzakis, and Vikrant Singhal.
\newblock New lower bounds for private estimation and a generalized
  fingerprinting lemma.
\newblock In {\em Advances in Neural Information Processing Systems (NeurIPS)},
  2022.

\bibitem[KMS{\etalchar{+}}22b]{kamath2022private}
Gautam Kamath, Argyris Mouzakis, Vikrant Singhal, Thomas Steinke, and Jonathan
  Ullman.
\newblock A private and computationally-efficient estimator for unbounded
  gaussians.
\newblock In {\em Conference on Learning Theory (COLT)}, 2022.

\bibitem[KMV22]{kothari2021private}
Pravesh Kothari, Pasin Manurangsi, and Ameya Velingker.
\newblock Private robust estimation by stabilizing convex relaxations.
\newblock In {\em Conference on Learning Theory (COLT)}, 2022.

\bibitem[KMZ22]{kothari2022polynomial}
Pravesh~K Kothari, Peter Manohar, and Brian~Hu Zhang.
\newblock Polynomial-time sum-of-squares can robustly estimate mean and
  covariance of gaussians optimally.
\newblock In {\em Algorithmic Learning Theory (ALT)}, 2022.

\bibitem[KOTZ14]{KauersOTZ14}
Manuel Kauers, Ryan O'Donnell, Li-Yang Tan, and Yuan Zhou.
\newblock Hypercontractive inequalities via sos, and the frankl--r{\"o}dl
  graph.
\newblock In {\em Proceedings of the Twenty-Fifth Annual ACM-SIAM Symposium on
  Discrete Algorithms}, pages 1644--1658. SIAM, 2014.

\bibitem[KSS18a]{kothari2018robust}
Pravesh~K Kothari, Jacob Steinhardt, and David Steurer.
\newblock Robust moment estimation and improved clustering via sum of squares.
\newblock In {\em Proceedings of the 50th Annual ACM SIGACT Symposium on Theory
  of Computing (STOC)}, 2018.

\bibitem[KSS18b]{KS17}
Pravesh~K Kothari, Jacob Steinhardt, and David Steurer.
\newblock Robust moment estimation and improved clustering via sum of squares.
\newblock In {\em Proceedings of the 50th Annual ACM SIGACT Symposium on Theory
  of Computing}, pages 1035--1046, 2018.

\bibitem[KSU20]{kamath2020private}
Gautam Kamath, Vikrant Singhal, and Jonathan Ullman.
\newblock Private mean estimation of heavy-tailed distributions.
\newblock In {\em Conference on Learning Theory (COLT)}, 2020.

\bibitem[KV18]{KarwaV18}
Vishesh Karwa and Salil~P. Vadhan.
\newblock Finite sample differentially private confidence intervals.
\newblock In {\em Innovations in Theoretical Computer Science Conference
  (ITCS)}, 2018.

\bibitem[Las01]{MR1846160-Lasserre01}
Jean~B Lasserre.
\newblock New positive semidefinite relaxations for nonconvex quadratic
  programs.
\newblock In {\em Advances in Convex Analysis and Global Optimization}, pages
  319--331. Springer, 2001.

\bibitem[LKKO21]{LiuKKO21}
Xiyang Liu, Weihao Kong, Sham~M. Kakade, and Sewoong Oh.
\newblock Robust and differentially private mean estimation.
\newblock In {\em Advances in Neural Information Processing Systems (NeurIPS)},
  2021.

\bibitem[LKO22]{LiuKO22}
Xiyang Liu, Weihao Kong, and Sewoong Oh.
\newblock Differential privacy and robust statistics in high dimensions.
\newblock In {\em Conference on Learning Theory (COLT)}, 2022.

\bibitem[LRV16]{lai2016agnostic}
Kevin~A Lai, Anup~B Rao, and Santosh Vempala.
\newblock Agnostic estimation of mean and covariance.
\newblock In {\em 2016 IEEE 57th Annual Symposium on Foundations of Computer
  Science (FOCS)}, 2016.

\bibitem[McS10]{McSherry10}
Frank McSherry.
\newblock Privacy integrated queries: an extensible platform for
  privacy-preserving data analysis.
\newblock {\em Commun. {ACM}}, 53(9):89--97, 2010.

\bibitem[MT07]{McSherryT07}
Frank McSherry and Kunal Talwar.
\newblock Mechanism design via differential privacy.
\newblock In {\em 48th Annual {IEEE} Symposium on Foundations of Computer
  Science (FOCS)}, 2007.

\bibitem[Nes00]{MR1748764-Nesterov00}
Yurii Nesterov.
\newblock Squared functional systems and optimization problems.
\newblock In {\em High performance optimization}, pages 405--440. Springer,
  2000.

\bibitem[NTZ13]{NikolovTZ13}
Aleksandar Nikolov, Kunal Talwar, and Li~Zhang.
\newblock The geometry of differential privacy: the sparse and approximate
  cases.
\newblock In {\em Annual ACM Symposium on Theory of Computing (STOC)}, 2013.

\bibitem[Par00]{parrilo2000structured}
Pablo~A Parrilo.
\newblock {\em Structured semidefinite programs and semialgebraic geometry
  methods in robustness and optimization}.
\newblock California Institute of Technology, 2000.

\bibitem[RV13]{rudelson2013hanson}
Mark Rudelson and Roman Vershynin.
\newblock Hanson-wright inequality and sub-gaussian concentration.
\newblock {\em Electronic Communications in Probability}, 18:1--9, 2013.

\bibitem[SCV18]{SteinhardtCV18}
Jacob Steinhardt, Moses Charikar, and Gregory Valiant.
\newblock Resilience: {A} criterion for learning in the presence of arbitrary
  outliers.
\newblock In {\em 9th Innovations in Theoretical Computer Science Conference
  (ITCS)}, 2018.

\bibitem[Sho87]{MR939596-Shor87}
Naum~Z Shor.
\newblock Quadratic optimization problems.
\newblock {\em Soviet Journal of Computer and Systems Sciences}, 25:1--11,
  1987.

\bibitem[SS21]{singhal2021privately}
Vikrant Singhal and Thomas Steinke.
\newblock Privately learning subspaces.
\newblock In {\em Advances in Neural Information Processing Systems (NeurIPS)},
  2021.

\bibitem[SV16]{SasonV16}
Igal Sason and Sergio Verd{\'{u}}.
\newblock $f$-divergence inequalities.
\newblock {\em IEEE Transactions on Information Theory}, 62(11):5973--6006,
  2016.

\bibitem[TCK{\etalchar{+}}22]{tsfadia2022friendlycore}
Eliad Tsfadia, Edith Cohen, Haim Kaplan, Yishay Mansour, and Uri Stemmer.
\newblock Friendlycore: Practical differentially private aggregation.
\newblock In {\em International Conference on Machine Learning (ICML)}, 2022.

\bibitem[Tuk60]{tukey60}
John~W. Tukey.
\newblock A survey of sampling from contaminated distributions.
\newblock {\em Contributions to probability and statistics}, 2:448--485, 1960.

\bibitem[Vad17]{Vadhan17}
Salil~P. Vadhan.
\newblock The complexity of differential privacy.
\newblock In {\em Tutorials on the Foundations of Cryptography --- Dedicated to
  Oded Goldreich}, pages 347--450. Springer, 2017.

\bibitem[Ver18]{vershynin2018high}
Roman Vershynin.
\newblock {\em High-dimensional probability: An introduction with applications
  in data science}, volume~47.
\newblock Cambridge University Press, 2018.

\bibitem[Wai19]{wainwright2019high}
Martin~J. Wainwright.
\newblock {\em High-Dimensional Statistics: A Non-Asymptotic Viewpoint}.
\newblock Cambridge University Press, 2019.

\bibitem[War65]{Warner65}
Stanley~L. Warner.
\newblock Randomized response: A survey technique for eliminating evasive
  answer bias.
\newblock {\em Journal of the American Statistical Association},
  60(309):63--69, 1965.

\end{thebibliography}
\bibliographystyle{alpha}

\appendix
\clearpage
\section{Technical Lemmata from KMV22}
\label{sec:kmv-apx}
Here we re-state and modify some useful theorems
and proofs from \cite{kothari2021private}.

\subsection{Truncated Laplace Mechanism}

The Laplace mechanism~\cite{DworkMNS06} is  a widely used mechanism for ensuring (pure) DP. It adds noise  drawn from the Laplace distribution to the output of the algorithm one wants to privatize.

\begin{definition}[Laplace Distribution]
The Laplace distribution with mean $\mu$ and parameter $b$ on $\R$, denoted by $\Lap(\mu, b)$, has the PDF $\frac{1}{2b} e^{-|x-\mu|/b}$. 
\end{definition}

\begin{theorem}[Laplace Mechanism~\cite{DworkMNS06}]
For $\eps > 0$. Consider any function $f:\calX^n\rightarrow\reals^K$. 
For any dataset $x\in\calX^n$, the Laplace Mechanism outputs
$$
f(x) + (L_1, \ldots, L_k),
$$
where $L_1, \ldots, L_k$ are drawn i.i.d. from
$\Lap(0, \Delta/\eps)$ where $\Delta$ is the global sensitivity of the function
$f$. Furthermore, the mechanism satisfies $(\eps, 0)$-DP.
\end{theorem}

We use a ``truncated'' version of the Laplace mechanism
where the noise distribution is shifted and truncated to be non-negative. 

\begin{definition}[Truncated Laplace distribution]
The (negatively) truncated Laplace distribution with mean $\mu$ and parameter $b$, denoted by $\tLap(\mu, b)$ is defined as $\Lap(\mu, b)$ conditioned on the value being negative.
\end{definition}

The truncated Laplace mechanism adds noise drawn from the truncated Laplace distribution and it is known that the mechanism yields $(\eps,\delta)$-DP.
\begin{lemma}[Truncated Laplace Mechanism, Lemma 3.27 of \cite{kothari2021private}] 
\label{lem:tlap-dp}
Let $f: \cY \to \R$ be any function with sensitivity at most $\Delta$. Then the algorithm that adds $\tLap\left(-\Delta\left(1 + \frac{\ln\left(1/\delta\right)}{\eps}\right), \Delta/\eps\right)$ to $f$ satisfies $(\eps, \delta)$-DP.
\label{lem:trunc-lap}
\end{lemma}

Finally, we also state a bound on the tail probability of the truncated Laplace distribution which will be useful in our subsequent analysis.

\begin{lemma} \label{lem:tlap-tail}
 Suppose $\mu < 0$ and $b > 0$. Let $X \sim \tLap(\mu, b)$. Then, for $y < \mu$, we have that
 \[
  \Pr[X < y] = \frac{e^{(y-\mu)/b}}{2 - e^{\mu/b}}.
 \]
\end{lemma}

\subsection{Approximate DP Selection via Exponential Mechanism}
\label{sec:dp-selection}


We describe \SelectionAlg~of Theorem~\ref{thm:dp-apx-selection}. It relies on the truncated Laplace mechanism from the prior subsection.
Specifically, \SelectionAlg~works as follows.
\begin{enumerate}
\item First, run the exponential mechanism to select
$c_1\in\cC$ using budget of $\eps/2$ using the score
function.
\item Using the truncated Laplace mechanism
(Lemma~\ref{lem:trunc-lap}), sample 
$N\sim\tLap\left(-\Delta(1+\frac{2\log(1/\delta)}{\eps}), \frac{2\Delta}{\eps}\right)$ and check if
$\score(c_1, Y) + N\geq \kappa$. If so, output $c_1$. If not,
output $\perp$.
\end{enumerate}

Because $N$ is sampled from a truncated Laplace, we are
guaranteed that $N\leq 0$ so that for any
$c\neq\perp$ obtained from \SelectionAlg, 
$\score(c, Y)\geq \kappa$.

\subsection{Witness Checking}
Crucial to our DP outlier-robust algorithms is that they are \textit{witness-producing}: the algorithms find a sequence of weights on the input that induces a distribution with relevant properties (e.g., certifiable subgaussianity). Essentially, the algorithms search over witnesses  that have certifiable subgaussian   moments or satisfy certifiable hypercontractivity. However, with a small probability, the algorithm might not find a good witness and then reject. We can make the whole procedure DP by observing that if on the dataset $Y$, we can return a good witness, then on a neighboring dataset $Y'$, the same witness should also be good for $Y'$. 

In the next section, we describe a DP mechanism for implementing the \textbf{Witness checking} steps of Algorithms~\ref{alg:approx-dp-mean-est} and \ref{alg:approx-dp-cov-est}. This mechanism and its analysis already appear in \cite{kothari2021private}; we include it in this section for the sake of completeness. As discussed below, we implement a minor modification to Lemma 4.18 of \cite{kothari2021private} (that is already implicit in that work) that improves the final utility guarantee of the procedure. For simplicity, we only describe below the mechanism and its analysis for the case of $\calA^{\textsf{mean}}$ and certifiable subgaussianity. The same final privacy and utility guarantees hold for the case of $\calA^{\textsf{cov}}$ and certifiable hypercontractivity.

\subsubsection{Certifiable Subgaussianity}
We begin by recalling the following results from \cite{kothari2021private}.
\begin{lemma}[Lemma 4.6 of \cite{kothari2021private}]
\label{lem:witness-property-analysis}
Let $\zeta$ be a pseudo-distribution of degree $O(k)$ consistent with $\calA$ on input $Y$ with outlier rate $\eta \ll 1/k$. 
Suppose there exists a $2k$-certifiably $C$-subgaussian distribution $X \subseteq \R^d$ with mean $\mu_*$ of size $n$ such that $|Y \cap X|\geq (1-\eta)n$. Then, for $\eta \leq \eta_0$ for some absolute constant $\eta_0$ and for $\hat{\mu} =\frac{1}{W} \sum_{i = 1}^n  \pE_{\zeta}[w_i] y_i$ where $W = \sum_{i =1}^n \pE[w_i]$, we have: 
\[
\frac{1}{W} \sum_{i = 1}^n \pE_{\zeta}[w_i] \iprod{y_i-\hat{\mu},u}^{2k} \leq (C' k)^k\Paren{\frac{1}{W} \sum_{i = 1}^n \pE_{\zeta}[w_i] \iprod{y_i-\hat{\mu},u}^{2}}^k ,
\]
for $C' \leq C (1+ O(\eta^{1-1/2k})k)$ for small enough $\eta$. 
\end{lemma}

\begin{lemma}[Lemma 4.18 of \cite{kothari2021private}]
\label{lem:witness-check-succeeds}
Let $0 \leq p_i(Y) \leq \frac{1}{(1-2\eta')}$ be a sequence of non-negative weights adding up to $n$ that induce a $2k$-certifiable $C'$-subgaussian distribution on $Y$. Let $p_i(Y')$ be a sequence of non-negative weights adding up to $n$ on $Y'$ adjacent to $Y$ such that $\norm{p(Y)-p(Y')}_1 \leq \beta$ for $\beta \leq \eta_0$. Then, for small enough absolute constant $\eta'>0$, $p_i(Y')$ induces a $2k$-certifiable $C' (1 + O(\beta^{1-1/2k})k)$-subgaussian distribution on $Y$. 
\end{lemma}
Lemma~\ref{lem:witness-check-succeeds} follows directly from the proof of Lemma 4.18 in \cite{kothari2021private} combined with Lemma~\ref{lem:witness-property-analysis}. As done in~\cite{kothari2021private}, after privately
selecting a stable outlier rate, we compute
a pseudo-distribution $\zeta$ minimizing the potential
for subgaussian certificate parameter $C$.
Witness-checking involves the following private
procedure $\calM_2$:
\begin{enumerate}
\item On input $(\calY, \eta,C, \zeta, \beta, \eps, \delta)$,
let $\gamma \sim \tLap(-\Delta(1+\log(1/\delta)), \Delta/\eps)$ for $\Delta = O(Ck\sqrt{L/n})$ large enough and
compute $C' = C + \gamma$. In the next lemma, we show this procedure
satisfies $(\eps, \delta)$-DP by using
Lemma~\ref{lem:trunc-lap} and a bound on the sensitivity of the subgaussianity parameter.
\item Check that
$p = \pE_{\zeta} [w] / \norm{\pE_{\zeta} [w]}_1$
induces a $2k$-certifiably $C'$-subgaussian distribution on $\calY$.
This check is just DP post-processing on the parameter
$C'$.
If the check does not pass, reject (i.e., return $C' = \bot$). 
Otherwise, return $C'$.
\end{enumerate}

We now show that computing $C'$ can be done while
satisfying $(\eps, \delta)$-DP:

\begin{lemma}
\label{lem:witness-checking}
The mechanism $\calM_2$ satisfies $(\epsilon,\delta)$-DP when invoked on input dataset $\calY \subset \R^d$ of size $n \geq O \left( \frac{L k^2 \log^2(1/\delta \beta)}{\epsilon^2} \right)$, outlier rate $\eta \in [0,1/2]$ and pseudo-distribution $\zeta$ satisfying $\calA^{\textsf{mean}}_{C, \eta, n}(\calY)$ that has passed the \textbf{Stable outlier rate selection} step of Algorithm~\ref{alg:approx-dp-mean-est} without halting. It outputs $C' \in \R_{>0} \cup \{\bot\}$ with the following guarantees:
\begin{enumerate}
    \item If the mechanism does not output $\perp$, then $\pE_\zeta [w]$ induces a  $2k$-certifiably $C'$-subgaussian  distribution on $\calY$. 
    \item With probability at least $1-\beta$, the mechanism satisfies $C' \in [C/2,C]$.
\end{enumerate}
 \end{lemma}
 
\begin{proof}
We begin by giving a privacy proof for $\calM_2$, again following \cite{kothari2021private}. Let $C^*(\calY)$ be the smallest $C^*$ for which
a $2k$-certifiably $C^*$-subgaussian distribution on $\calY$
can be induced. Then by Lemma~\ref{lem:witness-check-succeeds} for any neighboring $\calY'$,
$C^*(\calY) - C^*(\calY')\leq C^*(\calY) \cdot O(k\beta^{1-1/2k})$, where $\beta$ is the sensitivity $\beta = \norm{p(\calY) - p(\calY')}_1$. However, we know that Algorithm~\ref{alg:approx-dp-mean-est} has completed the \textbf{Stable outlier rate selection} without halting. Hence, it holds that that $C^*(\calY) \leq C$ and $\beta \leq O(L/n)$ by Lemma~\ref{lem:outlier-rate-select}. Combining this information with the fact that $k \geq 1$, we have $C^*(\calY) - C^*(\calY') \leq \Delta$ for some $\Delta = O(C k \sqrt{L/n})$. Invoking the guarantee of the Truncated Laplace Mechanism (Lemma~\ref{lem:trunc-lap}) for this choice of $\Delta$, we conclude that the sampling step of the the procedure $\calM_2$ satisfies $(\eps, \delta)$-DP. The rest of the procedure is post-processing.

As per utility, by the tail bounds of the truncated Laplace 
distribution (Lemma \ref{lem:tlap-tail}),
with probability at least $1-\beta$, it holds that
$\gamma\geq -O(\frac{\Delta}{\eps}\log(\frac{1}{\delta\beta}))$. Thus, we have that $C' = C + \gamma = \left(1-O \left( \frac{\sqrt{L} k\log(1/\delta \beta)}{\epsilon \sqrt{n}} \right) \right)C \geq C/2$ by assumption on $n$. Finally, by definition of the truncated Laplace distribution, $\gamma \leq 0$, so $C' \leq C$.
\end{proof}

\end{document}